\documentclass{llncs}
\usepackage{makeidx}

\usepackage[small,compact] 
  {titlesecMod}  
\titlespacing{\section}{0pt}{*2}{*1} 

\usepackage{times} 

\usepackage{wrapfig}

\usepackage{graphicx}
\usepackage{stmaryrd }

\usepackage{url}
\usepackage{mathpartir}
\urldef{\mailsa}\path|{alfred.hofmann, ursula.barth, ingrid.haas, frank.holzwarth,|
\urldef{\mailsb}\path|anna.kramer, leonie.kunz, christine.reiss, nicole.sator,|
\urldef{\mailsc}\path|erika.siebert-cole, peter.strasser, lncs}@springer.com|    

\usepackage{pdftexcmds}
\usepackage[cmex10]{amsmath}
\usepackage{mathtools}
\usepackage{amssymb}
\usepackage{latexsym}
\usepackage{makeidx}  
 \usepackage{listings}
\usepackage{semantic}
 \mathligprotect{\emptyset}
\usepackage{syntax}
\usepackage{wrapfig}
\usepackage{graphicx}
\usepackage{multirow}
\usepackage{bigdelim}
\usepackage{delarray}
\usepackage{color}
\usepackage{url}
\usepackage{hyperref}
\usepackage{cleveref}
\usepackage{tabularx}
\usepackage{ifthen}
\newboolean{long} 
\newboolean{verbose}
\newcommand{\para}[2][]{\vspace{0.25\baselineskip}\noindent\textbf{#2}#1~}
\newcommand{\red}[1]{\textcolor{black}{#1}}
\renewcommand{\bullet}
 {\,\begin{picture}(-1,1)(-1,-3)\circle*{3}\end{picture}\ }
 
 \newcommand{\rulename}[1]{\tiny \bf #1}

\newenvironment{itemizeC}%
  {\begin{list}{$\bullet$}%
   {\leftmargin=1.5\parindent \itemsep=0.5ex \topsep=2pt \labelwidth=6em \labelsep=1em
     \parsep=0.5ex \partopsep=0pt}}%
  {\end{list}}

\newcommand{\Set}[1]{ \{ #1 \} }

\newcommand{\Dom}{\mathrm{dom}}

\newcommand{\Union} {\cup}

\newcommand{\ForAll}[2]
           {\forall #1 .\,#2}

\DeclareSymbolFont{euletters}{U}{eur}{m}{n}
\SetSymbolFont{euletters}{bold}{U}{eur}{b}{n}
\DeclareMathSymbol{\varwp}{\mathord}{euletters}{"7D}

\newcommand{\pset}{\varwp} 
\newcommand\restr[2]{{
  \left.\kern-\nulldelimiterspace 
  #1 
  \vphantom{\big|} 
  \right|_{#2} 
  }}

\newcommand{\labarrow}[1]{\xrightarrow{#1}}  

\mathlig{<<}{\langle}
\mathlig{>>}{\rangle}

\newcommand{\whprog}{\texttt{while}}
\newcommand{\swhprog}{\texttt{i-while}}
\newcommand{\assem}{\texttt{RISC}}

\newcommand{\System}{\SemanticDomain{Sys}} 	
\newcommand{\SystemStates}{\mathcal{S}}          	
\newcommand{\Act}{\SemanticDomain{Act}}         	
\newcommand{\LowAct}{\SemanticDomain{LAct}}

\mathlig{-->}{\labarrow{}}					 	
\mathlig{-a->}{\labarrow{a}}
\newcommand{\Loc}{\SemanticDomain{Loc}}            

\newcommand{\faulty}{F}
\newcommand{\tolerant}{T}
\newcommand{\Environment}{\SemanticDomain{Env}}
\newcommand{\errlts}{\SemanticDomain{Err}}          	
\newcommand{\EnvironStates}{\mathcal{E}}		
\newcommand{\et}[1]{\stackrel{#1}{\leadsto}} 		
\newcommand{\faultfun}{\SemanticOp{Fault}}		
\newcommand{\set}[2]{\labarrow{#1}_{#2}}

\newcommand{\SE}[2]{\langle ~#1,~#2 ~\rangle}
\newcommand{\proba}[2]{\ensuremath{\SemanticOp{Pr}}_{#1}(#2)}
\newcommand{\run}[2]{\SemanticOp{run}_{#2}(#1)}

\usepackage{tikz}
\usetikzlibrary{calc,decorations.pathmorphing,shapes}
\newcounter{psarrow}

\newcommand{\Hloc}{H}
\newcommand{\Lloc}{L}

\newcommand{\pprog}[1]{\restr{#1}{\progdom}}
\newcommand{\prog}{P}

\newcommand{\progdom}{\SemanticDomain{ProgLoc}}
\newcommand{\lowdom}{\SemanticDomain{LowLoc}}
\newcommand{\highdom}{\SemanticDomain{HighLoc}}

\newcommand{\pubdata}{\Lloc}
\newcommand{\ppubdata}[1]{\restr{#1}{\lowdom}}
\newcommand{\pprivdata}[1]{\restr{#1}{\highdom}}
\newcommand{\privdata}{\Hloc}
\newcommand{\data}{M}
\newcommand{\ddata}[1]{\restr{#1}{\lowdom \cup \highdom}}

\newcommand{\bistate}[2]{(#1,#2)}
\newcommand{\aug}[1]{\SemanticDomain{#1^{+}}}
\newcommand{\transp}[1]{#1^{\infty}} 
\usepackage{tikz}
\usetikzlibrary{calc,decorations.pathmorphing,shapes}
\newcounter{sarrow}
\newcommand\poarrow[1]{%
\stepcounter{sarrow}%
\mathrel{\begin{tikzpicture}[baseline= {( $ (current bounding box.south) + (0,-0.5ex) $ )}]
\node[inner sep=.5ex] (\thesarrow) {$\scriptstyle #1$};
\path[draw,<-,decorate,
 decoration={zigzag,amplitude=0.7pt,segment length=1.2mm,pre=lineto,pre length=4pt}] 
   (\thesarrow.south east) -- (\thesarrow.south west);
\end{tikzpicture}}%
}
\newcommand{\faultArrow}[1]{\poarrow{#1}}

\newcommand{\gensetrun}[1]{\rho_{#1}}

\newcommand{\enabFamily}{\phi}

\newcounter{esarrow}
\newcommand\trpoarrow[1]{%
\stepcounter{esarrow}%
\mathrel{\begin{tikzpicture}[baseline= {( $ (current bounding box.south) + (0,-0.5ex) $ )}]
\node[inner sep=.5ex] (\thesarrow) {$\scriptstyle #1$};
\path[draw,<<-,decorate,
 decoration={zigzag,amplitude=0.7pt,segment length=1.2mm,pre=lineto,pre length=4pt}] 
   (\thesarrow.south east) -- (\thesarrow.south west);
\end{tikzpicture}}%
}

\newcommand{\termsymbol}{\infty}
\newcommand{\termtransp}[1]{\labarrow{#1}_\termsymbol} 
\newcommand{\whsemtt}[1]{\labarrow{#1}_\termsymbol}
\newcommand{\transtrace}[1]{\mathcal{T}_#1}
\newcommand{\transtracena}{\mathcal{T}}

\usepackage{tikz}
\usetikzlibrary{calc,decorations.pathmorphing,shapes}
\newcounter{pesarrow}

\newcommand{\trfaultArrow}[1]{\trpoarrow{#1}}

\newcommand{\SemanticDomain}[1]{\mathit{#1}}
\newcommand{\SemanticOp}[1]{\mathsf{#1}}

\newcommand{\Fault}{\SemanticOp{Fault}}
\newcommand{\flip}{\SemanticOp{flip}}
\newcommand{\memflip}{\SemanticOp{flip}}

\newcommand{\trace}{\SemanticOp{trace}}

\newcommand{\posft}{Possibilistic }

\newcommand{\wraprule}[1]{[#1]}

\newcommand{\update}{\mapsto} 

\mathlig{|->}{\mapsto}

\mathlig{-tau->}{\labarrow{\tau}}
\mathlig{||-}{\Vdash}
\newcommand{\whsem}[1]{\labarrow{#1}}

\DeclareMathAlphabet{\mathsc}{OT1}{cmr}{m}{sc}

\newcommand{\ct}[2]{\stackrel{#1}{\rightarrow}_{#2}}

\newcommand{\wskip}{\ensuremath{\SemanticOp{skip}}}
\newcommand{\assign}[2]{\ensuremath{#1 := #2}}
\newcommand{\assignName}{:=}
\newcommand{\wif}[3]{\ensuremath{\SemanticOp{if}\  #1\ \SemanticOp{then}\ #2 \ \SemanticOp{else}\ #3}}
\newcommand{\wifa}[3]{
\begin{array}{ll}
\SemanticOp{if} &  #1  \\
                             & \SemanticOp{then}\ #2 \\
                             & \SemanticOp{else}\ #3
\end{array}}

\newcommand{\wifName}{\ensuremath{\SemanticOp{if}}}
\newcommand{\while}[2]{\ensuremath{\SemanticOp{while}\ #1\ \SemanticOp{do}\ #2}}
\newcommand{\whileName}{\ensuremath{\SemanticOp{while}}}
\newcommand{\var}{\ensuremath{\SemanticDomain{Var}}}
\newcommand{\regvar}{\ensuremath{\SemanticDomain{Reg}}}
\newcommand{\allvari}{\ensuremath{\mathcal{V}}}

\newcommand{\wop}[2]{\ensuremath{#1\ \SemanticOp{op}\ #2}}

\newcommand{\ec}{\SemanticOp{R}}
\newcommand{\findot}{\bullet}
\newcommand{\mem}{M}

\newcommand{\whstate}[2]{\langle #1 \stateSplit #2\rangle}
\newcommand{\whlts}{\mathcal{S}_w}
\newcommand{\iwhlts}{\mathcal{S}_{i-w}}

\newcommand{\etypeproduce}{\hookrightarrow}
\newcommand{\seclev}{\SemanticOp{level}}
\newcommand{\esecan}[2]{\langle #1, #2\rangle}
\newcommand{\code}[1]{\{#1\}}
\newcommand{\subst}[2]{#1#2}
\newcommand{\activeReg}{A}

\newcommand{\emptyRegAll}{\{\}}
\newcommand{\regenv}{\ensuremath{\Phi}}
\newcommand{\conn}[3]{
\ifnum\pdfstrcmp{#3}{R}=0
	[#1\stackrel{#3}{\leftarrow}#2]
\else\ifnum\pdfstrcmp{#3}{W}=0
	[#1\stackrel{#3}{\rightarrow}#2]
\else 
	[#1\stackrel{#3}{\leftrightarrow}#2]
	\fi\fi
}

\newcommand{\breakconn}[1]{[#1\not \rightarrow]}
\newcommand{\envinters}{\sqcap}

\newcommand{\writeMode}{W}
\newcommand{\readMode}{R}

\newcommand{\assemprog}{P}
\newcommand{\typeproduce}{\hookrightarrow}
\newcommand{\enta}[2]{|-^{#1}_{#2}}
\newcommand{\eenta}[2]{||-^{#1}_{#2}}

\newcommand{\timee}{\SemanticOp{t}}

\newcommand{\secan}[2]{\langle #2, #1\rangle}
\newcommand{\valHigh}{H}
\newcommand{\valLow}{L}

\newcommand{\termlattice}{\mathcal{L}_\timee}

\newcommand{\writeMap}{\SemanticOp{write}}
\newcommand{\termMap}{\SemanticOp{term}}
\newcommand{\nextTerm}{\gg}

\newcommand{\regcorresp}{\stackrel{\rightarrow}{\equiv}}
\newcommand{\compa}[1]{corr(#1)}
\newcommand{\shado}[1]{shw(#1)}

\newcommand{\tdontc}{\tconst~\valLow}
\newcommand{\ttop}{\tconst~\valHigh}
\newcommand{\tclub}{\uplus}
\newcommand{\tlub}{\mathbin{\mathrlap{\sqcup}\scriptstyle -}}

\newcommand{\wconst}{\SemanticOp{Wr}}
\newcommand{\tconst}{\SemanticOp{Trm}}
\newcommand{\wdontk}{\wconst~\valLow} 
\newcommand{\whigh}{\wconst~\valHigh}
\newcommand{\wlub}{\sqcup}

\newcommand{\strongsim}{\approx}
\newcommand{\typeconc}[7]{
#1, #2 |- #3 \typeproduce #4,~ #7, ~#5, ~#6
}
\newcommand{\typeconce}[6]{
\begin{array}{lllllllll}
#1 & #6 & #2 & \typeproduce & #3 & , & #4 & , & #5 \\
\end{array}}
\newcommand{\condtype}[3]{#1?#2:#3}
\newcommand{\typeconcel}[6]{#1\ \ #6 \ #2 \ \typeproduce \ #3 \ , \ #4 \ , \ #5}

\newcommand{\whtoasse}{\alpha}
\newcommand{\vartoptr}{v2p}

\newcommand{\Constant}{\ensuremath{\mathbb{W}}} 

\newcommand{\Lab}{\ensuremath{\SemanticDomain{Lab}}}

\newcommand{\loadName}{\ensuremath{\SemanticOp{load}}}

\newcommand{\movekName}{\ensuremath{\SemanticOp{movek}}}
\newcommand{\moverName}{\ensuremath{\SemanticOp{mover}}}

\newcommand{\outName}{\ensuremath{\SemanticOp{out}}}
\newcommand{\addName}{\ensuremath{\SemanticOp{add}}}
\newcommand{\jmpName}{\ensuremath{\SemanticOp{jmp}}}
\newcommand{\sub}{\ensuremath{\SemanticOp{sub}}}
\newcommand{\jlez}{\ensuremath{\SemanticOp{jlez}}}
\newcommand{\mul}{\ensuremath{\SemanticOp{mul}}}

\newcommand{\jzName}{\ensuremath{\SemanticOp{jz}}}
\newcommand{\load}[4]{\ensuremath{\loadName^{{#1}}_{{#2}}\ #3 \ #4}}
\newcommand{\storeName}{\ensuremath{\SemanticOp{store}}}
\newcommand{\store}[4]{\ensuremath{\storeName^{#1}_{#2}\ #3 \ #4}}

\newcommand{\movek}[2]{\ensuremath{\SemanticOp{movek}~#1~#2}}
\newcommand{\mover}[2]{\ensuremath{\SemanticOp{mover}~#1~#2}}
\newcommand{\problemInstr}[1]{\ensuremath{\SemanticOp{ \bullet }}}
\newcommand{\jmp}[1]{\ensuremath{\SemanticOp{jmp}~#1}}

\newcommand{\jz}[2]{\ensuremath{\SemanticOp{jz}~#1~#2}}

\newcommand{\out}[2]{\ensuremath{\SemanticOp{out}~#1~#2}}
\newcommand{\nop}{\ensuremath{\SemanticOp{nop}}}
\newcommand{\binOp}{\ensuremath{\SemanticOp{op}}}
\newcommand{\aopName}{\ensuremath{\SemanticOp{op}}}
\newcommand{\res}{\ensuremath{\SemanticOp{res}}}

\newcommand{\concat}[2]{\ensuremath{#1 \concatprograms #2}}

\newcommand{\concatprograms}{\mathbin{+\!\!\!\!\,+}} 

\newcommand{\RegName}{\ensuremath{\SemanticDomain{R}}}

\newcommand{\nextpc}{\ensuremath{\pc^{+}}}
\newcommand{\HeapName}{\mathcal{M}}
\newcommand{\CodeName}{P}
\newcommand{\Heap}[1]{\ensuremath{\HeapName(#1)}} 
\newcommand{\pc}{\ensuremath{\SemanticDomain{pc}}}
\newcommand{\optional}[1]{[#1]} 

\newcommand{\asem}[1]{\labarrow{#1}}
\newcommand{\emptyLab}{\epsilon_{lab}}
\newcommand{\emptya}{\epsilon_I}
\newcommand{\memequiva}{\rightleftharpoons}

\newcommand{\stateSplit}{,} 
\newcommand{\absState}[3]{\ensuremath{\langle #1, \pc\stateSplit #2, #3 \rangle}}
\newcommand{\assemlts}{\mathcal{A}}
\newcommand{\absStatePost}[3]{\ensuremath{\langle #1, \nextpc\stateSplit #2, #3 \rangle}}
\newcommand{\absStateJump}[4]{\ensuremath{\langle #1, #2\stateSplit #3, #4
    \rangle}}

\newcommand{\asselts}{\mathcal{S}_a}

\newcommand{\low}{\ensuremath{\SemanticDomain{low}}}
\newcommand{\high}{\ensuremath{\SemanticDomain{high}}}

\usepackage{tikz}
\usetikzlibrary{arrows,automata}

\setboolean{long}{true}
\setboolean{verbose}{false}   

\author{Filippo Del %
\and David Sands
\and Alejandro Russo}
\begin{document}

\title{Type-Directed Compilation \\ for Fault-Tolerant Non-Interference}

\author{
Filippo Del Tedesco\inst{1}
\and
David Sands\inst{2}
\and
Alejandro Russo\inst{2}
}

\institute{Admeta AB, Sweden\\
\and
Chalmers University of Technology, Sweden\\
}
\maketitle

\begin{abstract}
Environmental noise (e.g. heat, ionized particles, etc.) causes transient faults in hardware, which lead to corruption of stored values. Mission-critical devices require such faults to be mitigated by fault-tolerance -- a combination of techniques that aim at preserving the functional behaviour of a system despite the disruptive effects of transient faults. Fault-tolerance typically has a high deployment cost -- special hardware might be required to implement it -- and provides weak statistical guarantees. It is also based on the assumption that faults are rare. In this paper, we consider scenarios where security, rather than functional correctness, is the main asset to be protected. Our contribution is twofold. Firstly, we develop a theory for expressing confidentiality of data in the presence of transient faults. We show that the natural probabilistic definition of security in the presence of faults can be captured by a possibilistic definition. Furthermore, the possibilistic definition is implied by a known bisimulation-based property, called Strong Security.
Secondly, we illustrate the utility of these results for a simple RISC architecture for which only the code memory and program counter are assumed fault-tolerant. We present a type-directed compilation scheme that produces RISC code from a higher-level language for which Strong Security holds -- i.e. well-typed programs compile to RISC code which is secure despite transient faults. In contrast with fault-tolerance solutions, our technique assumes relatively little special hardware, gives formal guarantees, and works in the presence of an active attacker who aggressively targets parts of a system and induces faults precisely. 
\end{abstract}

\section{Introduction}

Transient faults, or soft errors, are alterations of the state in one or more electronic components, e.g.,\ bit flips in a memory module \cite{Normand:SeuGround:1996}.  In some situations, we are willing to accept that a system may fail due to transient faults, for example, that an occasional message may be lost or corrupted. The problem we address in this paper is how to prevent transient faults (which will be referred to as ``faults'' in the rest of the paper) from compromising the security of the system -- for example, by allowing an attacker to access a secret.

%

It has indeed been shown that bit flips can effectively be used to attack, e.g., the AES encryption algorithm and reveal the cipher key \cite{Wright:aesfault:2003}. More recently, it has been shown that faults occurring in the computation of an RSA signature can permit an attacker to extract the private key from an authentication server \cite{Pellegrini:Rsafault:2010}.
%
Moreover, it is not just the security of crypto systems that is potentially vulnerable. Risks are greatly amplified if the attacker can influence the code executed by a system that can experience faults: it has been demonstrated that a single fault can cause (with high probability) a malicious but well-typed Java applet to violate the fundamental memory-safety property of the virtual machine \cite{SP03}.  


Traditional countermeasures against 
faults aim to make devices fault-tolerant, i.e. able to preserve their functional behavior despite soft errors. Most fault-tolerance mechanisms are based on redundancy: resources are replicated (either in software or in hardware) so that it is possible to detect, if not repair, anomalies. Typically, such solutions are designed for protection against a very simple fault model, in which a limited number of 
faults are assumed. Moreover, the formal guarantees of software-based fault-tolerance mechanisms, according to Perry et al.\ \cite{PLDI07}, are usually not stated -- it is more common that their efficacy is explored statistically.  In reasoning about security it is more difficult to give a statistical model of the environment's behaviour  -- it might be that the environment is an active adversary with a laser and a stopwatch (cf. \cite{OpticalFaults}) rather than just passive background radiation. Similarly, experimental evidence based on typical deployment of representative programs is not likely to be so useful if an attack, as it often does, requires an atypical deployment of an unusual program.

\para[]{Overview of the Paper and Contributions}
In order to be precise about the guarantees that we provide and the assumptions we make, we develop a formal, abstract model  of  fault-prone systems, attackers who can induce faults based on the observations made of the system's public output, and the resulting interaction between the two (Section \ref{subsec:system}). From this, we define our central notion of security in the presence of faults, \emph{probabilistic fault-tolerant non-interference} (PNI). A system is secure in the presence of faults if, for any attacker, the probability of a given public output is independent of the secrets held by the system (Section \ref{sec:security}).

Reasoning directly about fault-tolerant noninterfence is difficult because (i) it demands reasoning about probabilities, and (ii) it quantifies over all possible attackers (in a given class).

\textbf{Our first contribution}
tackles these difficulties with two key theorems. \textbf{Theorem \ref{thm:pniisponi}} (Section \ref{sec:possibft}) shows that PNI can be characterised exactly by a simpler \emph{possibilistic} definition of noninterference, defined by assuming that the exact presence and location of faults are observable to the attacker. This removes difficulty (i), the need to reason about probabilities. \textbf{Theorem~\ref{thm:ssimpliesponi}} (Section \ref{strongsecurity}) then shows that the possibilistic definition is implied by a form of strong bisimulation property, adapting ideas from its use in scheduler independent security of concurrent programs \cite{Sabelfeld:Sands:Multithreaded}. This theorem eliminates difficulty (ii), enabling PNI to be proved without explicit quantification over all attackers.

\textbf{Our second contribution} (Section \ref{Sec:instance})
illustrates the applicability of Theorem~\ref{thm:ssimpliesponi} in a specific
setting: a fault-prone machine modeling a program running on a RISC-like
architecture. The aim is to give a sound characterisation of RISC programs
which, when placed in the context of the machine, give a system which is secure
despite faults.  Our approach combines ideas from security type systems for
simple imperative programs with a \emph{type-directed compilation}
scheme. Roughly speaking, our type-directed compilation of simple imperative
programs guarantees that a well-typed program compiles to RISC code which, for
this architecture, is secure despite faults.
\red{A key assumption for this result to hold is assuming fault-tolerant code memory and program counter, while general purpose registers and data memory are susceptible to bit flips.}

Our innovative perspective on the problem of guaranteeing security in the
presence of transient faults combines a lot of concepts from other works, coming
from both security and dependability literature. In order to clarify these
connections, we discuss related works in Section \ref{Sec:relat}, whereas in
Section \ref{Sec:dis} we focus on the main limitations of our results (including
those shared with other approaches).  We complete the paper in
Section \ref{Sec:conc}, where we briefly summarize the conclusions of our
study.





\section{A Theory for Probabilistic Fault-Tolerant Non-Interference}\label{Sec:allnonint}


We begin (Section \ref{subsec:system}) by setting up a rigorous but abstract semantic model  for  systems that can experience transient faults, and for environments inducing these faults. Then, we establish a model for their interactions and a security definition (Section \ref{sec:security}) for the composition of a system with a fault environment as \emph{probabilistic fault-tolerant non-interference} (PNI).
%
%
We continue (Section \ref{sec:possibft}) by simplifying the security model introduced for PNI in two steps. Firstly, we propose a simpler, nondeterministic fault model for fault-prone systems. This model leads to a more straightforward notion of security called \emph{possibilistic fault-tolerant non-interference} (PoNI)\footnote{A similar but weaker notion of security has been presented in \cite{DelTedesco+:FTNInterference}.}
which is then shown to be equivalent to PNI.
Finally, we introduce (Section \ref{strongsecurity}) the notion of Strong Security (SS), a bisimulation-based security condition which is proved stronger than PoNI,  hence PNI.


\subsection{Probabilistic Fault-Tolerant Non-Interference}
\label{subsec:system}
We model a fault-prone system as a deterministic labelled transition system.
Since we have to reason about bit flips, we model the state of a system as just a collection of bits, each of which is identified by a unique \emph{location}.
The locations are partitioned into a \emph{fault tolerant part} of the system (e.g. memory with error correcting codes), and a faulty part. Only the faulty locations are affected by soft errors from the environment.
The labels of the transition system model outputs of the fault-prone system and ``clock ticks'' ($\tau$) which mark the discrete passage of time. We assume that some outputs can be distinguished by an external observer, whereas the others appear as $\tau$. 

\begin{definition}[Fault-prone System]\label{def:System}
A fault-prone system $\System$ is defined as $\System = \{ \Loc, \Act,  {-->} \subseteq \SystemStates \times \Act \times \SystemStates  \}$, where:
\begin{itemizeC}
\item
The set $\Loc$ contains all the locations (addresses) of the system. It is partitioned into a fault-tolerant subset and a faulty one, namely $\tolerant$ and $\faulty$. Fault-tolerant locations are not affected by soft errors, whereas bits stored in the faulty part of the hardware can get flipped because of transient faults.

\item $\Act$ includes system outputs and a distinguished silent action $\tau$ marking the passage of time. We assume that ``public'' observations (what an attacker can see) are limited to events in $\LowAct \subseteq \Act$, whereas any other action performed by the system is observed as $\tau$.
\item The relation $-->$ formalizes the system behavior. Given the set of all possible states for the system, namely $\SystemStates = \Loc \rightarrow \{ 0,1\}$, the relation $-->$ models how the system evolves.
We write transitions in the usual infix manner $S-a-> S'$ for $S,S' \in \SystemStates$. The system is deterministic in the sense that for any $S \in \SystemStates$ there is at most one state $S'$ and one action $l$ and a transition $S \labarrow{l} S'$.
\end{itemizeC}
\end{definition}



We now introduce a probabilistic fault model, which induces transient faults in $\faulty$. We refer to this model as the \emph{environment}, or, synonymously as the \emph{attacker}.


A simple environment can be modeled as an agent that induces bit flips in the
faulty locations with some fixed (small) probability. Unfortunately, this representation does not capture many realistic scenarios: for example, some locations in the system may be more error-prone than others, and (due to high densities of transistors) the
probability of a location being flipped may be higher if a neighboring
bit is flipped. We could therefore consider environments that can induce
faults according to a probability distribution over
the powerset of faulty locations, where the probability associated with a
given set of locations $L$ is the probability that exactly the bits of
$L$ are flipped before the next computation step occurs. 
However, such static model of the environment is not sufficient to model an active attacker. An active attacker may have physical access to the system (e.g., a tamper-proof smart card), and may be able to influence the likelihood of an error occurring at a specific location and time with high precision (see, e.g., \cite{OpticalFaults}).
What is more, the attacker may do this in a way that depends on the previous 
observations, namely the passage of time and the publicly observable actions. We therefore formalize these capabilities of an environment as follows (recall that $\pset(A)$ here denotes the \emph{powerset} of a given set $A$).

\begin{definition}[Fault environment]
\label{def:Env}

Consider a fault-prone system  
$\System = \{ \Loc, \Act,  {-->}   \}$. A fault environment $(\errlts,\faultfun)$ for $\System$ consists of:
  \begin{itemizeC}
  \item  a labelled transition system
  $\errlts= \langle \EnvironStates, \LowAct \cup \{\tau\}, {\et{}} \subseteq \EnvironStates \times  (\LowAct \cup \{\tau\}) \times \EnvironStates \rangle$, where $\EnvironStates$ represents the set of environment's states;
  \item a function
    $\faultfun
    \in \EnvironStates -> \pset(\faulty) -> [0,1]$ such that  for
    all states $E \in \EnvironStates$, we have that $\Fault(E)$ is a probability
    distribution on sets of locations.
  \end{itemizeC}
\end{definition}

We require that for all states $ E \in \EnvironStates$ and for all actions $a \in \LowAct \cup \{\tau\}$ there exists a unique state $E' \in \EnvironStates$ such that $E \et{a} E'$.
Intuitively, the state $E \in \mathcal{E}$
determines the probability that
a given set of locations (and no others) will experience a fault in the current
execution step of the system.
At each step, the state of the attacker evolves by the observation of ``low'' events and
the passing of time.

\begin{example}
We illustrate the use of Definition \ref{def:Env} by
characterizing the simplest most uniform fault-inducing environment:
a bit flip may occur in any  location with a fixed probability $\epsilon$.
To model this as a fault environment, we need only a trivial transition system
involving the single state $\bullet$ ($\EnvironStates = \{\bullet\}$), whose transitions are
 $\ForAll{a \in \LowAct \cup\{\tau\}} { \bullet \et{a} \bullet }$.
 As for the $\faultfun$ function consider a set $L \subseteq \faulty$ where $|L| = k$ and
$| \faulty| = n$. Then the probability that flips occur in all locations in $L$ and
nowhere else is equal to the probability of faults in every location in
$L$,  namely $\epsilon^{k}$,  multiplied by the probability of the remaining
 locations \emph{not} flipping, namely $(1-\epsilon)^{n-k}$. Hence
\[
\Fault( \bullet ) (L) = \epsilon^{k} (1-\epsilon)^{n-k} \quad \text{where $k
= |L|$ and $n = |\faulty|$}
\]

 \end{example}


We now define,
\ifthenelse{\boolean{long}}
{in outline (details are discussed in Appendix \ref{app:alldefs}),}
{in outline,}
 how fault environments are composed together with fault-prone systems. Some preliminary considerations are needed. We need to model the physical modification performed by the environment on the faulty part of the system. For this, a function $\flip$ is defined as  $\flip(S,L) = S [l |-> \neg S[l], l \in L ]$, which gives the result of flipping the value of every location in $L$ of state $S$  (assuming $L \subseteq \faulty$).
We also need to formalize that an attacker can distinguish only a subset of all possible actions of the system. This is obtained by assuming that there is a function $\low \in \Act \rightarrow \LowAct \cup \{\tau\}$ that behaves as the identity for actions in $\LowAct$, and maps any other action to $\tau$. This provides the public view of the system's output.
Finally, we say that a state $S$ is \emph{stuck} if there is no transition from that state.

\begin{definition}[Fault-prone System and Environment]
\label{def:FS}
Consider a fault-prone system $\System = \{ \Loc , \Act ,  {-->} \}$,
and an environment $\Environment=(\errlts,\faultfun)$ where $\errlts= \langle \EnvironStates,
\LowAct \Union\{\tau\}, \et{}\rangle$. The composition of $\System$ and $\Environment$, defined as $\System \times \Environment = << \SystemStates \times \EnvironStates, (\Act ,[0,1]), \set{}{} >>$ where $\SystemStates$ is the set of all possible states for $\System$, defines a labelled transition system whose states are pairs of the system state and the environment state, and with transitions labelled with an action $a$ and a probability $p$, written $\set{a}{p}$.  Transitions
depend on the state of the system as follows:
\begin{itemizeC}
\item If the system state $S$ is not stuck, it undergoes a transient fault according to the $\flip$ function; then, providing the flipped state is not stuck, the execution takes place, namely
%
%
%
\[
\inference[\rulename{Step}]{
 \pi=\{L~|~L \in \pset(\faulty) \mbox{ and } \flip(S,L)  -a->   S'\}  \\
  p = \Sigma_{L \in \pi} \Fault(E)(L) &
   E \et{\low(a)}E'
}
{ \SE{S}{E} \set{a}{p}  \SE{S'}{E'}}
\]
Observe that, in general, there might be several subsets of $L\subseteq \faulty$ such that $\flip(S,L)$ results in a state that (i) performs the same action $a$ and (ii) performs a transition to the final state $S'$. For this reason, the probability associated with the rule corresponds to the sum of all probabilities associated with locations in the set $\pi$.

\item If the system state is stuck, or it is made stuck by a transient fault, it does not perform any action. However, both the environment and the composition of the system state with the environment state evolve as if the system had performed a $\tau$ action. We enforce these restrictions because an error environment should not be able to distinguish between an active but ``silent'' and a stuck state.  Notice that this way of modeling the composition of systems and environments guarantees that any state of the composition can progress. 
\end{itemizeC}
\end{definition}

This form of transition system is sometimes known as a \emph{fully probabilistic labelled transition system}, or a \emph{labelled markov process}.


\subsection{Defining Security}\label{sec:security}

We can now reason about security of a system operating in an environment: $ \System \times \Environment$. Firstly, we define the observations that the attacker can perform. Then, we define when sensitive data remains secure despite the attacker observations.

Our attacker sees sequences of actions in $\LowAct \cup \{ \tau \}$, called \emph{traces}, and measures their probability, but does not otherwise have access to the state of the fault-prone system.

We say that the sequence $r=Z \ct{a_0}{p_0} Z_1 \ldots Z_{n-1} \ct{a_{n-1}}{p_{n-1}} Z_n$ is a
\emph{run} of size $n$ of a system state $Z \in \System \times \Environment$ and has probability $\proba{}{r}=\Pi_{0 \leq i \leq n-1}p_i$.
The set of all $n$-runs of $Z$ are denoted $\run{Z}{n}$, and we define the set of all runs for $Z$ as $\run{Z}{}=\Union_n \run{Z}{n}$.
Consider a run $r=Z \ct{a_0}{p_0} Z_1 \ldots Z_{n-1} \ct{a_{n-1}}{p_{n-1}} Z_n$ and let $\trace \in \run{Z}{} \rightarrow (\LowAct \cup \{\tau\})^{*}$ be a function such that $\trace(r) = \low(a_0)\ldots \low(a_{n-1})$. For any $t \in (\LowAct \cup \{\tau\})^n$, define $\proba{Z}{t}= \Sigma_{\{r \in \run{Z}{n}| \trace(r)= t\}} \proba{}{r}$. This definition induces a probability distribution over $(\LowAct \cup \{\tau\})^n$.

\begin{proposition}\label{prop:probsimp}
For any state  $Z \in \System \times \Environment$ and for any $n \geq 0$ $\Sigma_{t \in (\LowAct \cup \{\tau\})^n}\proba{Z}{t}=1$.
\end{proposition}
\ifthenelse{\boolean{long}}
{
\begin{proof}
Appendix \ref{Sec:proofprobdistr}.
\end{proof}
}
{}

We can now define the notion of \emph{probabilistic noninterference} \cite{gray:probani:1991} that characterises security for fault-prone systems. For this purpose, we consider the case when the initial state of a fault-prone system is an encoding of three different components: (i) a ``program'', the set of instructions executed by the system, (ii) ``public data'', that stores values known by an attacker and (iii) ``private data'' for confidential information.
 We formalize this partition by defining three mutually disjoint sets $\progdom$, $\lowdom$ and $\highdom$ (such that $\Loc=\progdom \cup \lowdom \cup \highdom$) and, for a system state $S$, by defining the program component $\prog$ as $\pprog{S}$, the public data $\pubdata$ as $\ppubdata{S}$ and $\privdata$ as $\pprivdata{S}$.

Observe that the way locations are partitioned between program and data is orthogonal to the way they are partitioned between fault-tolerant and faulty components. This is because fault-tolerance is orthogonal to the way security is defined.

Our definition of security, \emph{Probabilistic Fault-Tolerant Non-In\-ter\-fer\-ence} (PNI), requires a notion of equivalence to be defined for states that look the same from an attacker's point of view. Two system states $S$ and $S'$ are low equivalent, written as $S=_{\pubdata}S'$, if $\ppubdata{S}= \ppubdata{S'}$. Low equivalence provides us a way for defining when the program component of a state is secure even in the presence of transient faults. Intuitively the definition says that a program component $P$ is secure if the observed probability for any trace is independent of the sensitive data, for all system states where $P$ is the program component.

\begin{definition}[PNI]\label{def:proftni}
Let $\System$ be a fault-prone system and let $\prog$ be a program component of $\System$.  We say that $\prog$ is \emph{probabilistic fault-tolerant noninterfering}, if for any system states $S,S' \in \System$ such that $\pprog{S'}=\pprog{S}=\prog$ and $S=_{\pubdata}S'$,
it holds that for any state $E$ of any environment $\Environment$,
for any $n \geq 0$ and for any $t \in (\LowAct \cup \{\tau\})^n$ we have $\proba{\SE{S}{E}}{t}= \proba{\SE{S'}{E}}{t}$.
\end{definition}
The definition demands that probability of publicly observable traces only depends on values stored in the low locations. Also, it requires that this must hold for any fault-environment. 

\subsection{Possibilistic Characterisation of Fault-Tolerant Non-Interference}\label{sec:possibft}
Reasoning directly about fault-tolerant noninterference is difficult because (i) it demands reasoning about probabilities, and (ii) it quantifies over all possible attackers (in a given class). In this section, we address the first of these problems.

A \emph{possibilistic} model (i.e. not probabilistic) of the interaction between
a fault-prone system and the error environment can be obtained by interleaving
the transitions of the fault prone system with a nondeterministic flipping of
zero or more bits. \red{While this model avoids reasoning about probability
distributions as well as injection of faults by an attacker,}
it is not adequate to directly capture
security, as it is well-known that possibilistic noninterference suffers from
probabilistic information leaks (see e.g. \cite{gray:probani:1991}). In order to
capture PNI precisely, we augment this transition system by making the location
of the faults
\ifthenelse{\boolean{long}}
{observable (details are discussed in Appendix \ref{ref:deftranspblabla}).}
{observable.}
 \begin{definition}[Augmented Fault-prone System]
\label{def:NFS}
Given a fault-prone system $\System = \{ \Loc , \Act ,  {-->} \}$
we define the augmented system  $\aug{\System}$ as $\aug{\System} = \{ \Loc, \pset(\faulty) \times \Act,\faultArrow{\ \ } \}$ , where $\faultArrow{\ \ }$ is defined according to two cases:

\begin{itemizeC}
\item If the system state $S$ is not stuck, it undergoes a nondeterministic transient fault first; then, providing the flipped state is not stuck, execution takes place, namely


\[  \inference[]
  {
  \flip(S,L) -a-> S' \quad L \subseteq \faulty}
  { S  \faultArrow{L,a} S' }
\]
Observe that, compared to the corresponding rule in Definition \ref{def:FS}, we have that $L$ induces a unique transition since it appears in the transition label.
\item If the system state is stuck, or it is made stuck by a transient fault, the transition does not modify it. However, the label attached to the transition is $(L,\tau)$ so that, as in Definition \ref{def:FS}, we make a stuck state indistinguishable from a silently diverging one.
\end{itemizeC}
\end{definition}

The model resembles the composition of fault-prone systems and fault environments presented in Definition \ref{def:FS}, including the fact that it hides the termination of a system configuration. However, it introduces two main differences that influence the way security is defined.  On one hand, the augmented system is purely nondeterministic, and this supports a simpler definition of security. On the other hand, the augmented system has more expressive labels, that include not only the action performed by the system but also information about flipped locations.

We call the sequence $r=S \faultArrow{L_0,a_0} S_1 \ldots S_{n-1}\faultArrow{L_{n-1},a_{n-1}} S_n$ a \emph{possibilistic run} of a system state $S \in \aug{\System}$, and we say that $t=L_0,$$\low(a_0),$$\ldots,L_{n-1},\low(a_{n-1})$ is its corresponding trace. We write $S\trfaultArrow{\ t\ }$ when there exists a run $r$, produced by $S$, that corresponds to $t$. With these conventions, security for augmented systems, \emph{\posft Fault-Tolerant Non-Interfer\-ence} (PoNI), is defined by using the same notation presented in the previous section for fault-prone systems composed with environments.


\begin{definition}[PoNI]\label{def:possec}
We say that a program component $\prog$ enjoys \emph{possibilistic fault-tolerant non-interference}, if for any $S,S' \in \aug{\System}$ such that $\pprog{S'}=\pprog{S}=\prog$ and $S=_{ \Lloc}S'$, for any trace $t$, it holds that $S \trfaultArrow{\ t \ } \mbox{} \Leftrightarrow S' \trfaultArrow{\ t \ }$ \mbox{}.
\end{definition}

The following result says that the definitions of PNI and PoNI coincide.

\begin{theorem}\label{thm:pniisponi}
A program component $P$ enjoys PNI if and only if it enjoys PoNI.
\end{theorem}
\ifthenelse{\boolean{long}}
{
\begin{proof}
Appendix \ref{Sec:thmpniisponi}.
\end{proof}
}
{}

\red{This theorem makes clear that, from an enforcement perspective, security in the presence of transient faults can be achieved by targeting either PNI or PoNI.}

\subsection{Strong Security Implies PNI}\label{strongsecurity}

We now formalize a different notion of security that guarantees PoNI (and hence
PNI) without explicitly modeling the effects of transient faults in a
fault-prone system. This notion, called \emph{Strong Security}, was developed as
a way to capture a notion of scheduler independent compositional security for
multithreaded programs \cite{Sabelfeld:Sands:Multithreaded}.



Strong security is a bisimulation relation over program components of fault-prone systems. Our goal is to relate strong security to the possibilistic security definition established for augmented systems, and show that indeed it is stronger. Before doing so, we need to make sure that the semantics of fault-prone systems hides termination, as it is the case for augmented systems.
In particular, for a fault-prone system $\System= \{ \Loc, \Act,  {-->}\}$, we define its termination-transparent version as $\transp{\System} = \{ \Loc, \Act,  \termtransp{}\}$ where $\termtransp{}$ coincides with $-->$ for active states, but has additional transitions $S \termtransp{\tau} S$ whenever
\ifthenelse{\boolean{long}}
{$S \not -->$ (details are discussed in Appendix \ref{ref:deftranspblabla}).}
{ $S \not -->$.}
 With this, we define strong security for termination-transparent fault-prone systems as follows.

\begin{definition}[Strong Security (SS)]\label{def:strongsec}
Let $\System$ be a fault-prone system and $\transp{\System}= \{ \Loc, \Act,  \termtransp{}\}$ be the corresponding termination-transparent fault-prone system. A symmetric relation $R$ between program components is a strong bisimulation if for any $(\prog,\prog') \in R$ we have that for any two states $S,V$ in $\SystemStates$, if $\pprog{S}=\prog$ and $\pprog{V}=\prog'$ and $S =_{ \Lloc} V$ and $S \termtransp{a} S'$ then $V \termtransp{b} V'$ such that (i) $\low(a)= \low(b)$ and (ii) $S' =_{\Lloc} V'$ and (iii) $(\pprog{S'},\pprog{V'}) \in R$. We say that a program component $\prog$ is strongly secure if there exists a strong bisimulation $R$ such that $(\prog,\prog) \in R$.
\end{definition}

Intuitively, a program component $\prog$ is strongly secure when differences in the private part of the data are neither visible in the computed public data, nor in the transition label. This anticipates the fact that even though an external agent (the error environment, in our scenario) might alter the data component, the program behavior does not reveal anything about secrets.

The next result shows that Strong Security is sufficient to obtain PoNI. Notice, however, that the definition of Strong Security only deals with the modifications that occur in the data part of a fault-prone system. Hence, it only makes sense for the class of systems that host the program component in the fault-tolerant part of the configuration.

\begin{theorem}[Strong security $\Rightarrow$ PoNI]\label{thm:ssimpliesponi}
Let $\prog$  be a program component such that $\progdom \subseteq \tolerant$. If $P$ satisfies  SS then $P$ satisfies PoNI.
\end{theorem}
\begin{proof}
Appendix \ref{Sec:strongandponi}.
\end{proof}

We conclude this section with a negative result by saying that the reverse implication between SS and PoNI does not hold. Strong security requires that the  partition of the state into a program, a low and a high part is preserved (respected) during the whole computation whereas the definition of PNI imposes no special requirement on intermediate states of the computation.

\section{A Type System for FTNI}\label{Sec:instance}
\newcommand*{\Scale}[2][4]{\scalebox{#1}{$#2$}}%

In this section, we present an enforcement mechanism
  capable of synthesizing strongly secure assembly code. The enforcement is given
as a type-directed compilation from a source while-language. 
All in all, we present a concrete instance of our fault-prone system
formalization by defining the \assem\ architecture (Section \ref{sec:procdescr})
and propose a technique to enforce strong security over \assem\ programs
(Section \ref{sec:absttypesystem}), whose soundness is sketched in Section
\ref{sec:abssoundness}.  \red{By achieving strong security, we automatically get
secure RISC code that is robust against transient faults (recall Theorem
\ref{thm:ssimpliesponi}).}

\subsection{A Fault-Prone RISC Architecture}\label{sec:procdescr}

The architecture we are interested in has various hardware components to operate over data and instructions.

Data resides  in the memory and in the register bank. We model the
memory $\HeapName$ as a function $\Constant \rightarrow \Constant$,
where $\Constant$ is the set of all constants that can be represented
with a machine word. The register bank $\RegName$ is modeled as a
function $\regvar \rightarrow \Constant$, where $\regvar$, ranged over
by $r$, $r'$, is a set of register names.

\begin{wrapfigure}{r}{0.75\textwidth}
\centering
\begin{tabular}{l@{\hspace{5pt}}c@{\hspace{5pt}}l@{\hspace{5pt}}c@{\hspace{5pt}}l@{\hspace{5pt}}c@{\hspace{5pt}}l@{\hspace{5pt}}c@{\hspace{5pt}}l@{\hspace{5pt}}c@{\hspace{5pt}}l@{\hspace{5pt}}c@{\hspace{5pt}}l@{\hspace{5pt}}}
$I $ 		& ::= 		& \multicolumn{8}{l}{$\optional{l:} B $}\\
$B$		& ::=		& $\loadName\ r\ k$ 		&	$|$ 	& $\storeName	\ k  \ r$ 	& $|$	& $\jmpName   \ l$ 	& $|$ 	& $\jzName\ l \ r$ 	& $|$ $\nop$ 	 \\
  		&		& $\movekName\ r\ n$	& $|$&  $\moverName\ r\ r'$		& $|$& $\binOp\ r\ r'$ 	& $|$ 	& $\out{ch}{r} $ 	    	& \\
$ch$		& ::=		& $\low$ 				& $|$ & $\high$ 	&&&&&\\
\end{tabular}
\caption{\assem\ instructions syntax\label{table:assemsyntax}}
\end{wrapfigure}
Figure \ref{table:assemsyntax} describes the instruction set of our architecture. We consider that every instruction $I$ could be optionally labeled by a label in the set $\Lab$. Instruction $\load{}{}{r}{k}$ accesses the data memory $\HeapName$ with the pointer $k \in \Constant$ and writes the value pointed by $k$ into register $r \in \regvar$. The corresponding
$\store{}{}{k}{r}$ instruction writes the content of $r$ into the data memory address $k$. Instruction $\jmp{l}$ causes the control-flow to transfer to the instruction labeled as $l$. Instruction $\jz{l}{r}$ performs the jump only if the content of register $r$ is zero. Instruction $\nop$ performs no computation. The instruction $\movek{r}{k}$ writes the constant $k$ to $r$, whereas the instruction $\mover{r}{r'}$ copy the content in $r'$ to $r$.
The instruction $\binOp$ stands for a generic binary operator that combines values in $r$ and $r'$ and stores the result in $r$. Instruction $\out{ch}{r}$ outputs the constant contained in $r$ into the channel $ch$, that can be either $\low$ or $\high$.

The processor fetches \assem\ instructions from the code memory
$\assemprog$, separated from $\HeapName$. The code memory  is modeled
as a list of instructions. We require the code memory to be well-formed,
namely not having two instruction bodies labeled in the same way. A
dedicated \emph{program counter} register  stores the location in
$\assemprog$ hosting the instruction being currently executed. The
value of the program counter is ranged over by $\pc$.


As described in Section \ref{Sec:allnonint}, we partition the
architecture into faulty and fault-tolerant components. Both
$\RegName$ and $\HeapName$ are considered to be faulty: transient faults can strike any
location at any time during the execution. On the other hand, we
assume $\assemprog$ and the program counter are implemented in the
fault-tolerant part of the architecture. The fact that the code memory
is fault-tolerant corresponds to having the machine code in a
read-only memory with ECC, a common
assumption in dependability domain. The requirement on the program counter is a restriction that turns out to be necessary for proving the soundness of our enforcement mechanism.


We instantiate the \assem\ architecture  as a fault-prone system by
defining the semantics of the language as a labelled transition
system.
A state is defined as
a quadruple $\langle \assemprog,\pc, \RegName,\HeapName \rangle$, for $\pc
\in \Constant$,
where the first two elements correspond to the fault-tolerant portion of the hardware. Any action of the system is either an output ($\low!k$ when the output is performed on the publicly observable channel, $\high!k$ otherwise) or the silent action $\tau$. A few examples of transition rules are described in Figure \ref{table:abssemanticsquick} (the full presentation can be found in
\ifthenelse{\boolean{long}}
{Appendix \ref{app:fullassemsema}).}
{the extended version of this paper \cite{deltedesco+typecompftni}).}
 We write $\assemprog(\pc)$ as a shorthand for the instruction at position $\pc$ in $\assemprog$ and $\nextpc$ as a shorthand for $\pc+1$. We assume that the function $\res_\CodeName \in \Lab \rightharpoonup \Constant$ returns  the position at which label $l$ occurs in $P$: $\res_\CodeName(l)=i$ iff $\CodeName(i)=l:B$ for some $B$.


\begin{figure}[t]
\setnamespace{0pt}\setpremisesend{0.25em}\setpremisesspace{0.6em}
\[
\begin{array}{@{}c@{}}
\inference[\rulename{Load}]{
P(\pc)= \loadName\ r\ p
}
{\absState{P}{\RegName}{\HeapName}
 \asem{\tau}
 \absStatePost{P}{\RegName[r \update \Heap{p}]}{\HeapName}
}
\\[3.5ex] 
\inference[\rulename{Store}]{
P(\pc)= \storeName\ p\ r
}
{\absState{P}{\RegName}{\HeapName}
\asem{\tau}
 \absStatePost{P}{\RegName}{\HeapName[p \update \RegName(r)]}
}
\\[3.5ex] 
\inference[\rulename{Jz-S}]{
P(\pc)= \jzName \ l\ r &
\RegName(r)  = 0
}
{\absState{P}{\RegName}{\HeapName}
\asem{\tau}
\absStateJump{P}{\pc\update \res_\CodeName(l)}{\RegName}{\HeapName}
}
\\[3.5ex] 
\inference[\rulename{Jz-F}]{
P(\pc)= \jzName \ l\ r &
\RegName(r) \not =0
}
{\absState{P}{\RegName}{\HeapName}
\asem{\tau}
 \absStatePost{P}{\RegName}{\HeapName}
} \ \
%
\inference[\rulename{Out}]{
P(\pc)= \outName\ ch \  r  &
Reg(r) =  n
}
{
\absState{P}{\RegName}{\HeapName}
\asem{ch!n}
\absStatePost{P}{\RegName}{\HeapName}
}
\end{array}
\]
\caption{Selected rules for \assem\ instructions semantics}\label{table:abssemanticsquick}
\end{figure}

Once the rule $\wraprule{\rulename{Load}}$ is triggered, the register content at $r$ is updated with the memory content at $p$ ($\RegName[r \update \Heap{p}]$), and the program counter is incremented by one ($\nextpc$). Conversely, the rule $\wraprule{\rulename{Store}}$ writes the content at register $r$ in memory location $p$ ($\HeapName[p \update \RegName(r)]$). In case of a jump instruction, neither memory nor registers are modified.  If the guard is 0, as in rule $\wraprule{\rulename{Jz-S}}$, the execution is restarted at the instruction of the label used as the jump argument $\pc\update \res_\CodeName(l)$, otherwise the program counter is just incremented. All previous instructions map to silent actions $\tau$: channels are written by instruction $\wraprule{\rulename{Out}}$ which, on the other hand, leaves both register and memory untouched.

\subsection{Strong Security Enforcement}\label{sec:absttypesystem}
We guarantee strong security for some \assem\ programs via a novel
approach based on type-directed compilation.
Our strategy targets a simple high-level imperative language,
\whprog. For \whprog\ programs we define a type system that performs
two tasks: (i) translation of \whprog\ programs into \assem\ programs
(ii) enforcement of Strong Security on \whprog\ programs. The
compilation is constructed so that the strong security at the level of
\whprog\ programs is preserved by the compilation.


To facilitate the proof of strong security, the
 method is factored into a two-step process:  a
 type-directed translation to an intermediate
 language followed by a simple compilation to \assem.
 In this section, we limit
 ourself to an overview of the method and therefore elide this intermediate step
 from the presentation.


The grammar of the \whprog\ language is presented in Figure \ref{table:whilesyntax}. Both expressions and commands are standard, and assume that the language contains an output command $\outName$.

\begin{figure}[h]
\centering
$\begin{array}{l c  l c l c l c  l c  l c l }
C 		& ::=		& \wskip & | & \assign{x}{E} & | & \wif{E}{C_1}{C_2} & | & \out{ch}{E} & |  & C_1; C_2 & | &   \while{x}{C}  \\

E		& ::=		& k \in \Constant 
                                       & | &  x \in \var & | & \wop{E_1}{E_2} & & & & & &\\
ch		& ::=		& \low				& | & \high  & & & && & & & 	\\
\end{array}$
\caption{\whprog\ programs syntax\label{table:whilesyntax}}
\end{figure}

\para{Typing Expressions} The general structure of the typing
judgement for expressions is presented in Figure
\ref{fig:typingforexpr}.

\begin{wrapfigure}{r}{0.55\textwidth}
\[
\underbrace{\regenv, \activeReg, l}_{\text{(1)}}
||- \overbrace{E}^\text{(2)}
\etypeproduce
\underbrace{\CodeName}_\text{(3)},
\overbrace{ \esecan{\lambda}{n} }^\text{(4)},
\underbrace{r, \regenv'}_\text{(5)}
\]
\caption{General structure for expression typing rules}\label{fig:typingforexpr}
\end{wrapfigure}

The core part of the judgement says that \whprog\ expression $E$ (2) can be
compiled to secure \assem\ program $\CodeName$ (3), with security annotation
(4).  For defining the security annotation, we assume \whprog\ variables and
\assem\ registers are partitioned into two security levels $\{\Lloc, \Hloc\}$
(ordered according to $\sqsubseteq$, which is the smallest reflexive relation
for which $\Lloc \sqsubseteq \Hloc$) according to the function $\seclev \in \var
\cup \regvar \rightarrow \{ \Lloc,\Hloc\}$. The first component $\lambda$ of a
security annotation $(\lambda, n)$ specifies the security level of the registers
that are used to evaluate an expression. The second component $n$ represents the
number of \assem\ computation steps that are necessary to evaluate $E$. The
reason for tracking this specific information about expressions (and in commands
later on) is related to obtain assembly code which avoids timing leaks -- the
type-directed compilation will use $n$, for instance, to do \emph{padding} of
code when needed~\cite{Agat:Timing}. In fact, most of the involved aspects
in our type-system arises from avoiding timing leaks in low-level code.

Because the compilation is defined compositionally, some auxiliary
information (1) is required:
 we firstly need to
know the label $l$ which is to be attached to the first instruction of
$\CodeName$. Also, we have to consider the set $A$ of
registers which cannot be used in the compilation of $E$, since
they hold intermediate values that will be needed after the
computation of $E$ is complete. Finally, we need to keep track of how variables are
mapped to registers. This is done via the register record $\regenv$, which requires
a slightly more elaborate explanation.

\para{Register Records}
In order to allow for efficiently compiled code, expression compilation
builds up a record of associations between \whprog\ variables and \assem\ registers
in a \emph{register record} $\regenv \in \regvar \rightharpoonup
\var$. If $\regenv(r) = x$ then it means that the current value of
variable $x$ is present in register $r$.
The register record produced by a compilation is highly
nondeterministic, meaning that we do not build any particular register
allocation mechanism into the translation.

The rules (in particular the later rules for commands) involve a number of operations on register records which we
briefly describe here. We ensure that a register record is always a partial bijection, namely
 a register is associated to at most one variable and a variable is
 associated with at most one register. We write
 $\subst{\regenv}{\conn{r}{x}{}}$ to denote
 the minimal modification of $\regenv$
 which results in a partial bijection mapping $r$ to $x$. Similarly,
$\subst{\regenv}{\breakconn{r}}$ denotes the removal of any
association to $r$ in $\regenv$. The intersection of records $\regenv \envinters \regenv'$ is just the subset of the
bijections on which $\regenv$ and  $\regenv'$ agree. Finally,
inclusion between register records, written as $\regenv \sqsubseteq
\regenv'$, holds if all associations in $\regenv$ are also found in $\regenv'$.


Beside the compiled expression and its security annotation, each rule returns a modified register record and specifies the actual register where the evaluation of the expression $E$ is found at the end of the execution of $\CodeName$ (5).

\para{Expression Rules}

It will be convenient to
 extend the set of instructions with the empty instruction $\emptya$. Some sample rules for computing the types of
expressions are  presented in Figure \ref{table:abstypesystemexprquick}
(the full presentation can be found in
\ifthenelse{\boolean{long}}
{Appendix \ref{app:allabsrulexp}).}
{the extended version of this paper \cite{deltedesco+typecompftni}).}

\begin{wrapfigure}{r}{0.75\textwidth}
\begin{mathpar}
\Scale[1]{
\inference[\rulename{K}]
{r \not \in \activeReg}
{\regenv, \activeReg, l ||- k \etypeproduce [l:\ \movekName\ r\ k], \esecan{\seclev(r)}{1}, r, \subst{\regenv}{\breakconn{r}}}}\\
\and
\Scale[1]{
\inference[\rulename{V}-cached]
{\regenv(r) = x}
{\regenv, \activeReg, l ||- x \etypeproduce [l:\ \emptya], \esecan{\seclev(r)}{0}, r, \regenv}}
\end{mathpar}
\caption{Selected type system rules for \whprog\ expressions\label{table:abstypesystemexprquick}}
\end{wrapfigure}

In rule $\wraprule{\rulename{K}}$ the constant $k$ is compiled to
code which write the constant to some
register $r$ via the $\movekName\ r\ k$ instruction, providing $r$ is
not already in use ($r \not \in A$). As a result, any previous
association between register $r$ and a variable is lost
($\subst{\regenv}{\breakconn{r}}$). The security level of the result
is simply the level of the register, and the computation time is one.


In rule $\wraprule{\rulename{V}-cached}$ the variable to be compiled is already associated to a register, hence no code is produced.

\begin{wrapfigure}{r}{0.5\textwidth}
$$
\underbrace{\regenv, l}_{\text{(1)}}
|- C
\typeproduce
\CodeName,
\overbrace{ \secan{t}{w} }^\text{(3)},
\underbrace{l', \regenv'}_\text{(2)}
$$
\caption{General structure for command typing rules}\label{fig:typingforcmd}
\end{wrapfigure}

\para{Typing Commands}
The general structure of a typing rule for commands is presented in
Figure \ref{fig:typingforcmd}.  Judgements for commands \
assume a starting label for the code to be produced, and an incoming
register record (1). A compilation will result in a new (outgoing)
register record, and the label of the next instruction following this
block (2) (cf. Figures \ref{table:absttypeatomquick},
\ref{table:absttypeif} and \ref{table:absttypewhcomp}). The security
annotation (3) is similar to that for
expressions;
$w$, the \emph{write effect}, provides information about the security level of variables,
registers, and channels to which the compiled code writes,
and $t$ describes its timing behaviour. However, $w$ and $t$ are drawn
from domains which include possible uncertainty.


The write effect $w$ is described with a label defined in the set $\Set{\whigh,\wdontk}$, with partial ordering
$ \whigh \sqsubseteq  \wdontk $.
%
The value $\whigh$ is for
programs that never write to registers and memory locations
outside of $\Hloc$. The value $\wdontk$ is used when write operations
might occur at any security level.

\begin{wrapfigure}{r}{0.7\textwidth}
\centering
\[
\tconst~0 \sqsubseteq \dots \sqsubseteq \tconst~n
\sqsubseteq \dots \sqsubseteq \tdontc  \sqsubseteq \ttop, ~n \in \mathbb{N}
\] 
\begin{minipage}{0.6 \textwidth}
$\begin{array}{lcl}
t_1 \tlub t_2 & =  &
\left\{
	\begin{array}{ll}
		\tdontc  & \mbox{if } t_1 \sqsubseteq \tdontc
                         \mbox{ and } t_2 \sqsubseteq \tdontc
        \\
		\ttop 
                      & \mbox{otherwise}\\
	\end{array}
\right. \\
t_1 \tclub t_2  & = &
\left\{
	\begin{array}{ll}
		\tconst~n_1+n_2  & \mbox{if } \forall i \in \Set{1,2}\quad t_i= \tconst~n_i  \\
		t_1 \tlub t_2 & \mbox{otherwise}
	\end{array}
\right.
\end{array}$
\end{minipage}
\caption{Termination Partial Ordering \label{fig:termlattice}}
\end{wrapfigure}

The timing behavior of a command is described by an element of the
 partial order (and associated operations) defined in Figure
 \ref{fig:termlattice}. We use timing $t = \tconst~n$, for $n \in \mathbb{N}$, when termination of the code is
guaranteed in exactly $n$ steps (and hence is independent of any
secrets); $t = \tdontc$ is used for
programs whose  timing characterization does not depend on secret
values, but whose exact timing is either unimportant, or difficult to
calculate statically.
When secrets might directly influence the timing behavior of a program, the label $\ttop$ is used.

\para{Command Rules}
We now introduce some of the actual rules for commands.
 The concatenation of code memories $\CodeName$ and $\CodeName'$ is
 written $\concat{\CodeName}{\CodeName'}$ and is well defined if the
 resulting program remains well-formed. It will be convenient to
 extend the set of labels $\Lab$ with a special empty label
 $\emptyLab$ such that $\emptyLab:B$ simply denotes $B$. Also, we consider that the empty instruction $\emptya$ is such that if $\CodeName=[B,I_1,\dots,I_n]$ we define $\concat{[l:\emptya]}{\CodeName}$ as $[l:B,I_1,I_2,\dots,I_n]$.

\begin{wrapfigure}{r}{0.75\textwidth}
\begin{mathpar}
\inference[\rulename{\assignName}]
{\regenv, \emptyRegAll, l ||- E \etypeproduce \CodeName, \esecan{\seclev(x)}{n}, r, \regenv'
& \regenv''=\subst{\regenv'}{\conn{r}{x}{}}} 
{\typeconc{\regenv}{l}{\assign{x}{E}}
{
 \{
 \CodeName ~\concatprograms ~
 {[}\storeName\ \vartoptr(x) \ r{]}
 \}
}
{\emptyLab}{\regenv''}{
tp
}
}
\\
\text{where } tp = \begin{cases} \secan{\tconst~n+1}{\whigh} &\text{if $\seclev(x)=\Hloc$} \\
                   \secan{\tdontc}{\wdontk} & \text{otherwise}. \end{cases}
\end{mathpar}
\caption{Type system rule for assignment}\label{table:absttypeatomquick}
\end{wrapfigure}


The rule $\wraprule{\rulename{\assignName}}$ in Figure
\ref{table:absttypeatomquick} requires that the security level of
expression $E$ matches the level of the variable $x$
($\esecan{\seclev(x)}{n}$). If this is possible, the compilation is
completed by storing the value of $r$ into the pointer corresponding
to $x$ via the instruction $\storeName\ \vartoptr(x) \ r$ (assuming
there exists an injective function $\vartoptr \in \var \rightarrow
\Constant$ which maps \whprog\ variables to memory locations), and the
register record is updated by associating $r$ and $x$
 ($\subst{\regenv'}{\conn{r}{x}{}}$).
The resulting security annotation depends on the level of $x$:  when $\seclev(x)=\Hloc$ the security annotation is $\secan{\tconst~n+1}{\whigh}$, otherwise it is $\secan{\tdontc}{\wdontk}$. Rules for $\wskip$ and $\outName$ can be found in
\ifthenelse{\boolean{long}}
{Appendix \ref{app:absatomicstat}.}
{the extended version of this paper \cite{deltedesco+typecompftni}.}

\begin{figure*}[tb]
\begin{mathpar}
\setnamespace{0pt}\setpremisesend{0.25em} 
\inference[\rulename{\wifName}-any]
{\regenv, \emptyRegAll, l ||- E \etypeproduce \CodeName_0,
  \esecan{\lambda}{n_0} , r, \regenv_1 &
\forall i \in \Set{1,2}\quad
\typeconc{\regenv_1}{
  \emptyLab}{C_i}{\CodeName_i}{l_i}{\regenv_{i+1}}{ \secan{t_i}{w_i}}
\\ 
w_i \sqsubseteq \writeMap(\lambda)
& 
\mathit{br}, \mathit{ex} \mbox{ fresh}
}
{
\typeconc{\regenv}{l}{\wifa{E}{C_1}{C_2}}{
\left \{
\begin{array}{l}
\CodeName_0 \concatprograms 
{[}\jzName\ br \ r{]} \concatprograms \\
\concat{\CodeName_1}{[l_1:\ \jmpName\ ex]} \concatprograms \\
\concat{br: \CodeName_2}{[l_2:\ \nop]}
\end{array}
\right \}
}
{ex}{ \regenv_2 \envinters \regenv_3}
{
\left \langle
\begin{array}{c}
\writeMap(\lambda)\\
\termMap(\lambda)\tlub t_1 \tlub t_2
\end{array}
\right \rangle
}}
\\ 

\inference[\rulename{\wifName}-\valHigh]
{\regenv, \emptyRegAll, l ||- E \etypeproduce \CodeName_0, \esecan{\valHigh}{n_0} , r, \regenv_1 \\
\forall i \in \Set{1,2} \quad
  \typeconc{\regenv_1}{\emptyLab}{C_i}{\CodeName_i}{l_i}{\regenv_{i+1}}{\secan{\tconst~n_i}{\whigh}}
\\
m = n_0 + max(n_1,n_2) + 2 & \mathit{br}, \mathit{ex} \mbox{ fresh}
}
{\typeconc{\regenv}{l}{\wifa{E}{C_1}{C_2}}
{
\left \{
\begin{array}{l}
\CodeName_0 \concatprograms {[}\jzName\ br \ r{]} \\
\concatprograms
\concat{\CodeName_1}{\concat{l_1:\nop^{n_2-n_1}}{[\jmpName\ ex]}}
\\
\concatprograms
\concat{\CodeName_2}{\concat{l_2:\nop^{n_1-n_2}}{[\nop]}}
\end{array}
\right \}
}
{ex}
{\regenv_2 \envinters \regenv_3}
{\left \langle
\begin{array}{c}
\whigh\\
\tconst~m
\end{array}
\right \rangle
}}
\end{mathpar}

\caption{Type rules for $\wifName$}\label{table:absttypeif}
\end{figure*}

The rule $\wraprule{\rulename{\wifName}-any}$ in Figure
\ref{table:absttypeif} builds the translation of the $\wifName$
statement by joining together several \assem\ fragments.
The basic idea of this rule is that it follows Denning's classic
condition for certifying information flow security
\cite{Denning:Denning:Certification}: if the conditional involves high
data then the branches of the conditional cannot write to anything
except high variables. This is obtained by imposing the side condition $w_i \sqsubseteq
\writeMap(\lambda)$, where $w_i$ is the write-effect of the respective
branches, and $\writeMap$ is a function mapping the security level of the guard into its corresponding write-effect (such that $\writeMap(\Hloc)= \whigh$ and $\writeMap(\Lloc)= \wdontk$).  In this rule, in contrast to the
$\wraprule{\rulename{\wifName}-\valHigh}$ rule for the conditional,
the timing properties of the two branches may be
different, so we do not attempt to return an accurate timing. Hence,
the use of the operator $\tlub$ which just records whether the timing
depends on only low data, or possibly any data (notice that the security level of the guard is mapped into its corresponding timing label by the function $\termMap$, such that $\termMap(\Lloc)= \tdontc$ and $\termMap(\Hloc)= \ttop$). The compilation of the conditional code into \assem\ is fairly
straightforward: compute the expression ($P_0$) into register $r$, jump to
else-branch ($P_2$) if $r$ is zero, otherwise fall through to ``then''
branch ($P_1$) and
then jump out of the block. The resulting register record of the
whole command compilation is the common part of
the register records resulting from the respective branches.

The rule $\wraprule{\rulename{\wifName}-\valHigh}$ (Figure
\ref{table:absttypeif}) allows the system to be more permissive.
It deals with a conditional
expression which computes purely with high data -- a so-called
\emph{high conditional}. This rule, when applicable, compiles the
high conditional in a way that
guarantees that its timing behaviour is \emph{independent of the high
  data}. This is important since it is the only way that we can permit
a computation to securely write to low variables after a high conditional.
This is related to timing-sensitive information-flow typing rules for
high conditionals by Smith \cite{Smith:CSFW01}. The basic strategy is to
compute  the timing of each branch ($n_1$ and $n_2$ respectively)
and pad the respective
branches in the compiled code with sequences of $\nop$ instructions so that they become equally
long, where $\nop^m$ is a sequence of $m$ consecutive $\nop$ instructions when
$m>0$, and is $\emptya$ otherwise.



The rule $\wraprule{\rulename{;}}$ for sequential composition (Figure
\ref{table:absttypewhcomp}) is largely standard: the label and
register records are passed sequentially from inputs to outputs, and
the security types are combined in the obvious way. The only twist,
the side condition, encodes the key idea in the type system of Smith
\cite{Smith:CSFW01}. If the computation of the first command has
timing behaviour which might depend on high data ($t_1 = \ttop$), then the
second command cannot be allowed to write to low data ($w_2 =
\whigh$), as this would otherwise reveal information about the high
data through timing of low events.


\begin{figure*}[tb]
\begin{mathpar}
\inference[\rulename{seq}]
{\typeconc{\regenv}{l}{C_1}{\CodeName_1}{l_1}{\regenv_1}{\secan{t_1}{w_1}} &
\typeconc{\regenv_1}{l_1}{ C_2}{\CodeName_2}{ l_2}{\regenv_2}{\secan{t_2}{w_2}} \\
t_1=\ttop \Rightarrow w_2=\whigh
}
{\typeconc{\regenv}{l}{C_1;C_2}{
\left \{
\begin{array}{l}
\CodeName_1 \concatprograms
\CodeName_2
\end{array} \right \}}{l_2}{\regenv_2}{\secan{t_1 \tclub t_2}{w_1 \wlub w_2}}
}
\\
\inference[\rulename{\whileName}]
{
\lambda = \seclev(x)=\seclev(r) &
 & t=\ttop \Rightarrow \writeMap(\lambda)=\whigh
&
w \sqsubseteq \writeMap(\lambda) \\
\regenv_B \sqsubseteq \subst{\regenv}{\conn{r}{x}{}}  &
\regenv_B \sqsubseteq \subst{\regenv_E}{\conn{r}{x}{}}
 & lp, ex \mbox{ fresh} \\
\typeconc{\regenv_B}{\emptyLab}{C}{\CodeName}{l'}{\regenv_E}{\secan{t}{w}}
& \CodeName_i  =[ \loadName\ r\ \vartoptr(x),\ \storeName\ \vartoptr(x) \ r]
}
{\typeconc
{\regenv}
{l}
{ \while{x}{C}}
{
\left \{
\begin{array}{l}
l: \ \CodeName_i  \concatprograms [lp:\jzName\ ex\ r]
\concatprograms {} \CodeName \concatprograms {} \\{}
l': \ \CodeName_i \concatprograms [\jmpName\ lp]
\end{array}
\right \}
}
{ex}
{\regenv_B}
{
\left \langle
\begin{array}{c}
\writeMap(\lambda) \\
\termMap(\lambda)\tlub t
\end{array}
\right \rangle
}}
\end{mathpar}
\caption{Type rules for sequential composition ($;$) and $\whileName$}\label{table:absttypewhcomp}
\end{figure*}

The compilation of the $\whileName$ command (Figure
\ref{table:absttypewhcomp})
is quite involved for two reasons.  Firstly, as one would expect in a typing rule for a
looping construct, there are technical conditions
relating the register record at the beginning of the loop, and the register
record on exit. This is because we need a single description of the exit
register record $\regenv_B$ which approximates both the register record at the
start of the loop body ($\subst{\regenv}{\conn{r}{x}{}}$) and the
register record after computing the loop body and putting $x$ back
into register $r$ ($\subst{\regenv_E}{\conn{r}{x}{}}$).
 Secondly, for technical reasons relating purely to the
proof of correctness (security), the code is (i) a little less compact that one
would expect to write due to an unnecessarily repeated subexpression, and (ii) contains a redundant
instruction $\storeName\ x \ r$ immediately after having loaded $x$
into $r$.
The lack of compactness is due
to the fact that the proof goes via an intermediate language that
cannot represent the ideal version of the code. The redundant
instruction
establishes a particular invariant that is needed in the
proof: not only is $x$ in register $r$, but it arrived there as the
result of writing $r$ into $x$.
 The security concerns are taken care of by ensuring
that the security level of the whole loop is consistent with the
levels of the branch variable $x$ and the branch register $r$, and
that if the timing of the body might depend on high data, then the
level of the loop variable (and hence the whole expression) must be $H$.

All the rules introduced in this section are used in Appendix \ref{appendix:example} to compile a simple \whprog\ program that computes a hash function.

\subsection{Soundness}\label{sec:abssoundness}



We open this section by instantiating the definition of strong security for \assem\ programs, which requires to view the  \assem\ machine as an instance of a
fault-prone system (Definition \ref{def:System}). For this we consider the
set of locations $\Loc$  of the \assem\ fault-prone system to be the
names of the individual bits comprising the registers and memories. So,
for example, a general purpose register $r$ corresponds to some set of
locations $r_0,\ldots r_{31}$ (for a word-size of 32). With this
correspondence, the set of states of the \assem\ system are isomophic
to the set of functions $\Loc -> \Set{0,1}$.
As mentioned earlier, the fault-prone locations $F$
are those which correspond to the general purpose registers and the
data memory.

For the definition of security we must additionally partition the
locations into the program $\progdom$, the low locations $\lowdom$, and the high
locations $\highdom$:  $\progdom$ comprises the locations of the code \emph{and} the program
counter register, $\lowdom$ the locations of the low variables and
registers, and $\highdom$ the  locations of the high variables and registers.

Since assembly programs are run starting at their first instruction, the following
slightly specialised version of strong security is appropriate:
\begin{definition}[Strong Security for \assem\ programs]
We say that an assembly program $\CodeName$ is strongly secure if $(\CodeName,0)$ is strongly secure according to Definition \ref{def:strongsec} instantiated on the fault-prone system $\assemlts=\{ \Loc,\{ch!k| ch \in \{\low,\high\} \mbox{ and } k \in \Constant\}\cup \{\tau\},\asem{}\}$.
\end{definition}

The type system defined in Section \ref{sec:absttypesystem} guarantees that any type-correct \whprog\ program is compiled into a strongly secure \assem\ program. This is formalized as follows.

\begin{theorem}[Strong security enforcement]\label{thm:strsenforce}
Let $C$ be a \whprog\ program, and suppose $\{\}, \emptyLab |- C \typeproduce \CodeName, \secan{t}{w}, l,  \regenv$. Then $\CodeName$ is strongly secure.
\end{theorem}
\ifthenelse{\boolean{long}}
{
\begin{proof}
See Section \ref{sec:realtypesystem}.
\end{proof}
}
{}

According to Theorem \ref{thm:strsenforce}, we can obtain strongly secure \assem\ programs from type-correct \whprog\ programs. Theorem \ref{thm:ssimpliesponi} (Section \ref{sec:possibft}) states that strong security is a sufficient condition to guarantee PoNI. The two results together express a strategy to translate \whprog\ programs into \assem\ programs that enjoy PoNI. We state this formally by instantiating the definition of PoNI for \assem\ programs.

\begin{definition}[PoNI for \assem\ programs]
We say that an assembly program $\CodeName$ enjoys PoNI if $(\CodeName,0)$ is PoNI according to Definition \ref{def:possec} instantiated on the fault-prone system $\assemlts=\{ \Loc,\{ch!k| ch \in \{\low,\high\} \mbox{ and } k \in \Constant\}\cup \{\tau\},\asem{}\}$.

\end{definition}

\begin{corollary}[PoNI enforcement on \assem\ programs]
Let $C$ be a \whprog\ program, and suppose $\{\}, \emptyLab |- C \typeproduce \CodeName, \secan{t}{w}, l,  \regenv$. Then $\CodeName$ is PoNI.
\end{corollary}
\begin{proof}
Direct application of Theorem \ref{thm:strsenforce} and Theorem \ref{thm:ssimpliesponi}.
\end{proof}

\section{Related Work}\label{Sec:relat}

\para{Fault Tolerant Non-Interference}
The only previous work of which we are aware that aims to prevent transient faults from violating noninterference is by Del Tedesco et al.~\cite{DelTedesco+:FTNInterference}. The enforcement approach of that paper is radically different from the approach studied here, and the two approaches are largely complimentary. Here we highlight the differences and tradeoffs: 
\begin{itemizeC}
\item 
Targeting a similar RISC machine, the implementation mechanism of \cite{DelTedesco+:FTNInterference} is a combination of software fault isolation \cite{Wahbe:1993:originalsfi} and a black-box non-interference technique called secure multi-execution \cite{Devriese:2010}. This can be applied to any program, but only preserves the behaviour of noninterfering memory-safe programs. Verifying that a program is memory-safe would have to be done separately, but could be achieved by compiling correctly from a memory-safe language.  
\item In \cite{DelTedesco+:FTNInterference}, fault-tolerance is assumed for the code memory but not in the program counter register. The cost of this is that the method described in that paper can only tolerate up to a statically chosen number of faults, whereas in the present work we can tolerate any number. 
\item The security property enforced by the method described in \cite{DelTedesco+:FTNInterference} is a restriction of PoNI to runs with a limited number of faults. However, the work does not justify this definition with respect to the more standard notion of probabilistic noninterference. The limitation in the number of faults, together with our result, shows that the established security property is strictly weaker than PNI.  
\end{itemizeC}

\para{Strong Security for Fault Tolerance}
This paper is not the first to realise that Sabelfeld and Sands's strong security \cite{Sabelfeld:Sands:Multithreaded} implies security in the presence of faults. Mantel and Sabelfeld \cite{heiko:andrei:distributedbisim} used strong security in a state-based encoding of 
channel-based communication. They observed that strong security is not affected by faults occurring in message transmission. Another way to think of this is that strong security of individual threads implies strong security of their composition; a faulty environment is itself a strongly secure thread, simply because it has no ability to read directly from secrets in the state. 

\para{Related Type Systems}
The type-directed compilation presented here combines several features which are inspired by existing non-interference type systems for sequential and concurrent programming languages. Our security notion is timing sensitive and has some similarities with 
Agat's~\cite{Agat:Timing}  type-directed source-to-source transformation method that maps a source program into an equivalent target program where timing leaks are eliminated by padding. Similar ideas were shown to apply to a type system for strong security \cite{Sabelfeld:Sands:Multithreaded}.
Our padding mechanism is different from Agat's, since it is based on  counting the number of computation steps in the branches of a high conditional expressions, and our system is more liberal, since it allows e.g. loops with a secret guard. These distinguishing features are both present in Smith's type system for a concurrent language \cite{Smith:CSFW01} (see also \cite{Castellani:Boudol:TCS02}).

\para{ Non-interference for Low-level Programming Languages}
Medel et al.~\cite{TCS05} propose a type system for a RISC-like assembly language capable to enforce (termination and time insensitive) non-interference. 
Enforcing the same security condition, Barthe et al.~\cite{barthe:typesprescomp:2004} introduce a stack-based assembly
language equipped with a type system. 
Subsequent work~\cite{Barthe:2007} shows a compilation strategy which enforces non-interference across all the intermediate steps until reaching a JVM-like language. Bartuti et al.~\cite{Barbuti:2002} use a different notion for  confidentiality, called $\sigma$-security, which is enforced by 
abstract interpretation.

\para{Dependability}
The need for a stronger connection between security and dependability has been stated in many works (e.g.~\cite{erland:secandavail:2006,connectdepsec:2004}). 
Interestingly, it can be observed that many solutions for dependability are based on information-flow security concepts. 
In  \cite{Rushby99:partitioning}, well-known concepts from the information-flow literature are introduced as building blocks to achieve dependability goals. 
In \cite{Weber:FaultToleranceAsNoninterference:1989}, 
non-interference-like definitions are used to express fault tolerance in terms of program semantics. 
On the other hand, one could argue that the security domain has been influenced by dependability principles as well. For instance, our enforcement is sound only if fault-tolerant hardware components are deployed for the code memory and the program counter. 
	
\para{Language-Based Techniques for Fault-tolerance}
The style of our work -- in terms of the style of formalisation, the use of programming language techniques, and the level of semantic precision in the stated goals -- is in the spirit of Perry et al's fault-tolerant typed assembly language \cite{PLDI07}. Because we need to reason about security and not conventional fault tolerance, our semantic model of faults is necessarily much more involved than theirs, which is purely deterministic.

\para{Security and Transient Faults}
We have not been the first ones to consider the implications of
transient faults for security -- Bar-El et al.~ \cite{naccache:sorcerer:2006} survey a variety of methods that can be used to induce transient faults on circuits that manipulate sensitive data. 
 Xu et al.~\cite{vulnerabilities:errors:2001} study the effect of a single bit flip that strikes the opcode of x86 control flow instructions;  
their work states the non-modifiability of the source code, which is a crucial assumption in our framework. Bao et al.~\cite{Bao:TransientFaultsPublicKey:1998}  illustrate several transient-fault based attacks on crypto-schemes. 
Their protection mechanisms either involve some form of replication or a more intensive usage of randomness (to increase the unpredictability of the result). In a similar scenario, Ciet et al.~\cite{Ciet:EllipticTransientFaults:2005} 
show how the parameters of an elliptic curve crypto-system can be compromised by transient faults, and illustrate how a comparison mechanism is sufficient to prevent the attack from being successful. 
Canetti et al.\ \cite{Canetti+:Maintaining} discuss security in the presence of transient faults for cryptographic protocol implementations where they focus on how random number generation is used in the code.

Our approach relies on fault-tolerant support for the program counter. While it seems a bit restrictive, there are fault-tolerant solutions for registers (e.g.~\cite{Oliveira:RegStrong:2007,sparc64viiifx:2010}).

\section{Limitations}\label{Sec:dis}
The hardware model discussed here is similar to those introduced in
\cite{PLDI07,DelTedesco+:FTNInterference} and, in common with many informal models of faults, has similar
shortcomings: faults occurring at lower levels e.g. in combinatorial circuits, are not modelled. 
It has been argued \cite{Wang+:Characterizing} that these non-memory elements of a
processor have much lower sensitivity to faults than state elements, but in our attacker model this does not say so much.

For timing channels discussed in Section \ref{sec:absttypesystem} we make a large simplifying assumption: that the time to compute an instruction is constant. This ignores cache effects, so either implies a greatly simplified achitecture, or a need for further refinements to the method to ensure that cache effects are mitigated by preloading or using techniques from \cite{Agat:Timing}.


The language we are able to compile is too small to be practical, containing
neither functions nor arrays. The type system can perhaps be extended to handle
such features, based on our understanding of other security type systems, and
with the help of some additional fault-tolerant components (an extra fault-tolerant register should be sufficient to guarantee SS in the presence of dynamic pointers). However, the main
challenge was in the non-trivial correctness proof 
\ifthenelse{\boolean{long}}
{(Appendix \ref{sec:realtypesystem})}
{\cite{deltedesco+typecompftni}}. 
Clearly, mechanically verifiable proofs will facilitate extending our system 
to other features as well as verifying our confidence in the formal results. 


\section{Conclusion}\label{Sec:conc}


We formalize security in presence of transient faults as Probabilistic
Fault-Tolerant Non-Interference (PNI). We simplify it by reducing it to 
possibilistic framework (PoNI), and we show that another well-known security
condition, called Strong Security \cite{Sabelfeld:Sands:Multithreaded}, implies
it. We explore a concrete instance of our formalism. We consider a simple RISC
architecture, in which the only fault-tolerant components are the program
counter and the code memory. We define a type system that maps programs written
in a toy while-language to the assembly language executed by our architecture
and, at the same time, ensures that the produced code enjoys Strong Security
(hence PNI).

\bibliography{local,conferences,zippedbiblio}
\bibliographystyle{IEEEtran} 

\newpage

To avoid unintentional printing of the lengthy appendices of this submission, they have been removed from this version. The paper plus appendices can be obtained from  \texttt{arxiv.org}.

\newpage

\appendix

\section{Full Definition of Fault-prone Systems with Environments}\label{app:alldefs}

We formally account the composition of a Fault-prone System together with an error environment as follows.

 \begin{definition}[Fault-prone System with an Environment]
\label{def:FSFormal}
Consier a fault-prone system $\System = \{ \SystemStates , \Act ,  {-->} \}$, 
and an environment $\Environment=(\errlts,\faultfun)$ where $\errlts= \langle \EnvironStates,
\LowAct \Union\{\tau\}, \et{}\rangle$. The composition of $\System$ and $\Environment$, defined as   
$\System \times \Environment = << \SystemStates \times \EnvironStates, (\Act ,[0,1]), \set{}{} >>$ where $\SystemStates$ is the set of all possible states for $\System$, defines a labelled transition system whose states are pairs of the system state and the environment state, and with transitions labelled with an action $a$ and a probability $p$, written $\set{a}{p}$. Transitions are determined according to the following rules: 


\begin{center}
\[
\inference[\rulename{Step}]{ 
 \exists l S \labarrow{l} S' & \pi=\{L~|~L \in \pset(\faulty) \mbox{ and } \flip(L,S)  -a->   S'\}  \\
  p = \Sigma_{L \in \pi} \Fault(E)(L) &
   E \et{\low(a)}E' 
}
{ \SE{S}{E} \set{a}{p}  \SE{S'}{E'}}
\]
\\
\[
\inference[\rulename{Stuck-1}]{
 \exists l S \labarrow{l} S' & L \in \pset(\faulty) & 
 \flip(L,S)  \not \rightarrow &
 p = \Fault(E)(L) & 
  E \et{\tau}E'  
}
{ \SE{S}{E} \set{\tau}{p}  \SE{\flip(L,S)}{E'}}
\]
\\
\[
\inference[\rulename{Stuck-2}]{ 
 S  \not --> &
E \et{\tau}E' 
}
{ \SE{S}{E} \set{\tau}{1}  \SE{S}{E'}}
\]
\end{center}
\end{definition}

In rule $\wraprule{\rulename{Step}}$ some locations in the fault-prone component of the system are flipped, then the execution can take place. In general, there might be several subsets of $L\subseteq \faulty$ such that $\flip(L,S)$ results in a state that (i) performs the same action $a$ and (ii) performs a transition to the final state $S'$. For this reason the probability associated to the rule corresponds to the sum of all probabilities associated to locations in the set $\pi$. 

We assume that a stuck configuration cannot be resumed. In rule $\wraprule{\rulename{Stuck-1}}$ a flip on given subset of locations produces a stuck configuration with a certain probability. In rule $\wraprule{\rulename{Stuck-2}}$ a stuck configuration remains as such regardless of any possible flip of the faulty component.

\section{Proof of Proposition \ref{prop:probsimp}}\label{Sec:proofprobdistr}

In order to prove Proposition \ref{prop:probsimp}, we prove an auxiliary result which ensures that our assignment of probability to $n$-sized runs is a probability distribution.

\begin{proposition}\label{prop:probrun}
Let $\System$ be a fault-prone system and $\Environment$ an environment. Then
$\forall Z \in \System \times \Environment$ and $\forall n \geq 0$ $\Sigma_{r \in \run{Z}{n}} \proba{}{r}=1$.
\end{proposition}
\begin{proof}
We prove the statement by induction on $n$.

\framebox[1.1\width]{Base case} 

When $n=0$, the set $\run{Z}{0}$ contains only one element, the empty run, and its probability is 1.

\framebox[1.1\width]{Inductive step} 

Consider $n>0$. 

Any run $r \in \run{Z}{n}$ can be written as $r=Z \ct{a_0}{p_0} Z_1 \ldots Z_{n-1} \ct{a_{n-1}}{p_{n-1}} Z_n=(Z \ct{a_0}{p_0} Z_1 \ldots Z_{n-1}).(Z_{n-1} \ct{a_{n-1}}{p_{n-1}} Z_n)$, where $r'= Z \ct{a_0}{p_0} Z_1 \ldots Z_{n-1}$ is a prefix of $r$ in $\run{Z}{n-1}$ with probability $\proba{}{r'}=p_{r'}$.

Consider the set $R_{r'} \subseteq \run{Z}{n}$ of all $n$-sized runs from $Z$ that has $r'$ as prefix. Since for all subsets in $\pset(\faulty)$ there is a transition rule in $\System \times \Environment$, we have $\Sigma_{r \in \run{Z_{n-1}}{1}} \proba{}{r}=1$. Hence we have that $\Sigma_{r \in R_{r'}} \proba{}{r}=p_{r'}$ and we can conclude that $\Sigma_{r \in \run{Z}{n}} \proba{}{r} =  \Sigma_{r' \in \run{Z}{n-1}} \Sigma_{r \in R_{r'}} \proba{Z}{r} $ $= \Sigma_{r' \in \run{Z}{n-1}}  p_{r'}=1$, where the last result holds by inductive hypothesis.
\end{proof}

\begin{proof}[Proof of Proposition \ref{prop:probsimp}]
Recall that for any trace $t \in (\LowAct \cup \{\tau\})^n$, $\proba{Z}{t}= \Sigma_{\{r \in \run{Z}{n}| \trace(r)= t\}} \proba{}{r}$. Since for any run $r$ there is a trace $t \in (\LowAct \cup \{\tau\})^n$ such that $\trace(r)= t$, we have as a result  that   $\Sigma_{t \in (\LowAct \cup \{\tau\})^n } \proba{Z}{t} = \Sigma_{r \in \run{Z}{n}} \proba{}{r}=1$, where the last equality holds because of Proposition \ref{prop:probrun}.  
\end{proof}

\section{Full definitions of Augmented Fault-prone and Termination Transparent Systems}\label{ref:deftranspblabla}

An augmented fault-prone system is formally described as follows. 

 \begin{definition}[Augmented Fault-prone System]
\label{def:NFSFormal}
Given a fault-prone system $\System = \{ \Loc , \Act ,  {-->} \}$
we define the augmented system  $\aug{\System}$ as $\aug{\System} = \{ \Loc , \Act \times \pset(\faulty),\faultArrow{\ \ } \}$ by the following rules:


\begin{center}
\[
  \inference[]
  { \exists l S \labarrow{l} S' & \flip(S,L) -a-> S' \quad L \subseteq \faulty} 
  { S  \faultArrow{L,a} S' }
  \]
  \\
  \[
  \inference[]
  {\exists l S \labarrow{l} S' &  \flip(S,L) \not --> &  L \subseteq \faulty} 
  { S  \faultArrow{L,\tau} \flip(S,L) }
  \]
\\
  \[
  \inference[]
  { S \not --> \quad L \subseteq \faulty} 
  { S  \faultArrow{L,\tau} S }
  \]
  
\end{center}
\end{definition}

A termination transparent system is formalized as follows.

\begin{definition}[Termination Transparent System]
For a fault-prone system $\System= \{ \Loc, \Act,  {-->}\}$ we define its termination-transparent version as $\transp{\System} = \{ \Loc, \Act,  \termtransp{}\}$ where $\termtransp{}$ is defined with the following rules:
\[
  \inference[]
  { S -a-> S'} 
  { S \termtransp{a} S' }
  \inference[]
  { S \not --> } 
  { S  \termtransp{\tau} S }
  \]
 
\end{definition}

\section{Proof of Theorem \ref{thm:pniisponi}}\label{Sec:thmpniisponi}
Before proving the theorem in question, we need to define some auxiliary concepts.
 
\begin{definition}[Enabling set]
Let $Z=\SE{S}{E}$ and $Z'=\SE{S'}{E'}$ be a pair of states in $\System \times \Environment$ such that $Z \set{a}{p} Z'$. We say that $L$ is an enabling set (of locations) for $Z \set{a}{p} Z'$ in the following cases:
\begin{itemize}
\item the transition is derived from rule $\wraprule{\rulename{Step}}$ and $L \in \pi$;
\item the transition is derived from rule $\wraprule{\rulename{Stuck-1}}$ and $L$ is the argument in $\flip(S,L)$; 
\item the transition is derived from rule $\wraprule{\rulename{Stuck-2}}$.
\end{itemize}
\end{definition}

\begin{definition}[Enabling sequence]
Let $r$ be a run for $\SE{S_0}{E_0}$ in $\System \times \Environment$ such that $\SE{S_0}{E_0} \set{a_0}{p_0} \SE{S_1}{E_1} \dots \set{a_{n-1}}{p_{n-1}} \SE{S_n}{E_n} $. The sequence $\mathcal{L}=L_0\dots L_{n-1}$ such that $\forall 0 \leq i \leq n-1$ $L_i$ is an enabling set for $\SE{S_i}{E_i} \set{a_i}{p_i} \SE{S_{i+1}}{E_{i+1}}$ is called an enabling sequence for $r$. We define the probability of $\mathcal{L}$ as $\proba{r}{\mathcal{L}}=\Pi_{0 \leq i \leq n-1} \Fault(E_i)(L_i)$. We define the set of all enabling sequences for $r$ as $\enabFamily(r)$.
\end{definition}

We also need the following intermediate result.

\begin{lemma}\label{lemma:sumofprob}
Let $r$ be a run for $Z$ in $\System \times \Environment$. Then we have that $\proba{}{r}= \Sigma_{\mathcal{L} \in \enabFamily(r)} \proba{Z}{\mathcal{L}}$.
\end{lemma}

We can now prove Theorem \ref{thm:pniisponi}. 

\begin{proof}[Proof of Theorem \ref{thm:pniisponi}]
Suppose $\prog$ enjoys PoNI, we now show it enjoys PNI as well.\\
Consider a faulty system $\System$, an error environment $\Environment=(\errlts,\Fault)$ and two states  $Z=\SE{S}{E}$ and $Z'=\SE{S'}{E}$, for $S,S' \in \System$ and $E \in \errlts$. Assume $\pprog{S}=\pprog{S'}=P$ and $S=_\pubdata S'$. \\
We first  show that for any $n \geq 0$ and for any trace $t \in (\LowAct \cup \{\tau \})^n$, $\proba{Z}{t}  \leq \proba{Z'}{t}$, and hence by symmetry that $\proba{Z}{t} = \proba{Z'}{t}$. \\
We prove the inequality by relating the probability of a trace to the probability determined by the enabling sequences that corresponds to it. \\
Consider a trace $t$ such that $\proba{Z}{t}>0$. 
Let $\gensetrun{Z}(t)=\{r \in \run{Z}{}| \trace(r)=t\}$ be the (nonempty) set of runs from $Z$ whose trace is $t$. \newline
We have $\proba{Z}{t}= \Sigma_{r \in \gensetrun{Z}(t)} \proba{}{r} =  \Sigma_{r \in \gensetrun{Z}(t)} \Sigma_{\mathcal{L} \in \enabFamily(r)} \proba{Z}{\mathcal{L}}$, where the first equality holds by definition and the second one follows from Lemma \ref{lemma:sumofprob}.\\
We now show that all enabling sequences for $Z$ are also enabling sequences for $Z'$.  Let $\kappa_Z= \cup_{r \in \gensetrun{Z}(t)}  \enabFamily(r)$ be the set of all enabling sequences for $t$ in $Z$ and let $\mathcal{L}$ be an enabling sequence for a run $r \in \run{Z}{}$. Since $\prog$ is PoNI, there must be a run $r'$ from $Z'$ such that $\mathcal{L}$ is an enabling sequence for $r'$, and $\trace(r')=t$. Hence, for the set $\kappa_{Z'}= \cup_{r' \in \gensetrun{Z'}(t)}  \enabFamily(r')$ we have that $\kappa_Z \subseteq \kappa_{Z'}$.\\
Also, observes that for any $\mathcal{L} \in \kappa_Z$. $\proba{Z}{\mathcal{L}}= \proba{Z'}{\mathcal{L}}$, since $\trace(u)= \trace(u')$.\\
Then $\proba{Z}{t}=  \Sigma_{\mathcal{L} \in \kappa_Z} \proba{Z}{\mathcal{L}} \leq \Sigma_{\mathcal{L} \in \kappa_{Z'}} \proba{Z'}{\mathcal{L}}= \proba{Z'}{t}$. 
\newline
We continue by showing that PNI implies PoNI by proving  the contrapositive. 
Suppose that $\prog$ is not PoNI. Then there must a fault-prone system $\System$
, two states $S,S'$ such that $\pprog{S}=\pprog{S'}=\prog$ and $S=_\pubdata S'$, together with a location set $\lambda \in \pset(\faulty)$, a trace $t=L_0, a_0, \ldots, L_j, a_j$ and $a,b \in \LowAct \cup \{\tau \}$ such that $a \not = b$, $S \trfaultArrow{t,\lambda,a}$ and $S' \trfaultArrow{t,\lambda,b}$. 
\newline
Define an error environment $\Environment=(\errlts,\Fault)$ such that $\errlts = \langle \{L_i|L_i \in t \} \cup \{\lambda\},\LowAct \cup \{\tau\}, \{L_i \et{a} L_{i+1}|0 \leq i \leq j-1 \mbox{ and } a \in \LowAct \cup \{\tau\}\} \cup \{L_j \et{a} \lambda|a \in \LowAct \cup \{\tau\}\} \rangle$ and $\Fault(\lambda)(\lambda) = \Fault(L)(L)=1$. Essentially, $\Environment$ deterministically traverses all flipped locations included in $t$, and terminates in $\lambda$, regardless of the actions performed by the fault-prone system.
\newline
Consider now the composition of $\System$ with $\Environment$ and let $Z=\SE{S}{L_0}$ and $Z'=\SE{S'}{L_0}$. Let $t'=a_0.\ldots.a_j.a$ be a trace in $(\LowAct \cup \{\tau\})^{*}$ obtained from $t$ by (i) striping flipped locations and (ii) appending the action $a$ at the end. Then there exists a unique run $r \in \run{Z}{}$ such that $\trace(r)=t'$ and $\proba{}{r}= \proba{Z}{t'}=1 \not = \proba{Z'}{t'}=0$. The inequality between $\proba{Z}{t'}$ and $\proba{Z'}{t'}$ follows from the hypothesis of $\prog$ not being PoNI, which implies that there is no $r' \in \run{Z'}{}$ such that $\trace(r')= t'$. 
\end{proof}

\section{Proof of Theorem \ref{thm:ssimpliesponi}}\label{Sec:strongandponi}
Rather than showing that Strong Security implies  PoNI directly, we take an indirect approach. 

First we characterize the semantics of a termination-transparent system in terms of ``transition traces'', borrowing ideas from \cite{Brookes:transitionsem:1996}. Then we define an ad-hoc security property, called Strong Trace-based Security within this semantic model. We finally show that Strong Security implies Strong Trace-based Security, which in turn implies PoNI. 

For improving readability we represent $\ddata{S}$ as $\data$, the data component,  therefore the state of a fault-prone system is represented as $\bistate{\prog}{\data}$. We adapt the concept of low equality between states to data components by saying that $\data =_{\pubdata} \data'$ if $\ppubdata{\data}= \ppubdata{\data'}$.

\begin{definition}[Transition trace semantics]
Let $\System$ be a fault-prone system and let $\transp{\System}= \{ \Loc, \Act,  \termtransp{}\}$ be its termination-transparent version. The $n$-step transition-trace semantic of a program component $\prog_0$ is defined as the set $\transtrace{n}(\prog_0)$, such that $\transtrace{n}(\prog_0)=\{(\data_0,a_0,\data_0'),(\data_1,a_1,\data_1')\dots(\data_{n-1},a_{n-1},\data_{n-1}')|$ $ \forall 0 \leq i \leq n-1$$ \bistate{P_i}{\data_i} $$\termtransp{a_i} \bistate{P_{i+1}}{\data_i'}\}$. The transition trace semantics of $P_0$ is defined as $\transtracena(\prog_0)= \cup_{n} \transtrace{n}(\prog_0)$.
\end{definition}

In the transition trace model, the semantics of a program component $P$ is built in sequences of steps. In particular, at any step, the program component is executed on a certain data component, then the data component is modified and the execution is restarted. Observe that the model is very similar to the way a fault-prone system and an error environment interact with each other. This is even more clear when viewing the modification of the data component as the effect of its interaction with the error environment. 

We say that two transition traces $$
\begin{array}{c}
t=(\data_0,a_0,\data_0'),(\data_1,a_1,\data_1')\dots(\data_{n-1},a_{n-1},\data_{n-1}') \\ 
t'=(N_0,b_0,N_0'),(N_1,b_1,N_1')\dots(N_{n-1},b_{n-1},N_{n-1}')
\end{array}$$ 
are input low-equivalent, written $t=_I t'$, if $\forall 0 \leq i < n$, $\data_i =_{\pubdata} N_i$, whereas they are output low-equivalent, written $t=_O t'$ if $\forall 0 \leq i < n$, $\low(a_i) = \low(b_i)$ and $\data_i' =_{\pubdata} N_i'$.

\begin{definition}[Strong Trace-based Security (StbS)]
We say that a program component  $\prog$ is $n$-Strong Trace-based Secure if for any two transition traces $t,t'\in \transtrace{n}(\prog)$, if $t=_I t'$ then $t=_O t'$. We say that a program component $\prog$ is Strong Trace-based Secure if it is $n$-Strong Trace-based Secure for any $n \in \mathbb{N}$.
\end{definition}

We now show how to use the notion of Strong Trace-based Security to bridge the gap between Strong Security and PoNI. We show that Strong Security implies Strong Trace-based Security first. 
 
\begin{lemma}[SS implies StbS]\label{lemma:SStoStbS}
Let $\prog$ be a program component. If $\prog$ enjoys SS then $\prog$ enjoys StbS.
\end{lemma}
\begin{proof}
We define some notation first. We refer to the $i$-th triple in a transition trace $t$ as $t_i$, and to  the program component used to evaluate it as $\prog_{t_i}$ (for a trace $t=(\data_0,a_0,\data_0'),(\data_1,a_1,\data_1')\dots(\data_{n-1},a_{n-1},\data_{n-1}')$ we therefore say that the $i$-th triple $(\data_{i},a_{i},\data_{i}')$ is induced by $\bistate{\prog_{t_i}}{\data_i} \termtransp{a_i} \bistate{\prog_{t_{i+1}}}{\data_{i}'}$). 

Consider a program component $\prog$ and two $n$-transition traces $$t=(\data_0,a_0,\data_0'),(\data_1,a_1,\data_1')\dots(\data_{n-1},a_{n-1},\data_{n-1}')$$ and $$t'=(N_0,b_0,N_0'),(N_1,b_1,N_1')\dots(N_{n-1},b_{n-1},N_{n-1}')$$ in $\transtrace{n}(\prog)$. 

We want to show that if $\prog$ enjoys SS and $t=_I t'$, then $t=_O t'$. 

Starting from a strong bisimulation $R$ for $(\prog,\prog)$, the idea of the proof is to infer properties of $t'$ by unwinding $R$ for $n$-steps. We proceed by showing that for all $0 \leq i <n$ we have that $(\prog_{t_i},\prog_{t'_i}) \in R$. For $i=0$ $\prog_{t_0}=\prog_{t'_0}=\prog$ and $(\prog,\prog) \in R$. If $(\prog_{t_i},\prog_{t'_i})\in R$, then by definition of $R$ we have that if $\bistate{\prog_{t_i}}{\data_i} \termtransp{a_i} \bistate{\prog_{t_{i+1}}}{\data_{i}'}$ and $\data_i =_{\pubdata} N_i$ then $\bistate{\prog_{t_i'}}{N_i} \termtransp{t_i} \bistate{\prog_{t_{i+1}'}}{N_{i}'}$ and $(\prog_{t_{i+1}},\prog_{t'_{i+1}}) \in R$. But this is the case for $t$ and $t'$, since $t=_I t'$.

The statement of the lemma is therefore proved by recalling that two program components $\prog$ and $\prog'$ in a strong bisimulation $R$ are such that their executions from low equivalent data result in (i) low equivalent data and (ii) low equivalent actions. 
\end{proof}

We now discuss the relation between Strong Trace-based Security and PoNI. In general it is not true that Strong Trace-based Security is stronger than PoNI. Consider, for example, the class of systems such that $\progdom \not \subseteq \tolerant$. Due to transient faults, a completely innocuous program component can be converted into a harmful one, even when it enjoys Strong Trace-based Security. 

Surprisingly, this is not the only constraint that we must impose to the systems of our interest. We must also require that they show a uniform behavior for termination, as shown in the following example. 

\begin{example}
Consider the fault-prone system in Figure \ref{fig:stbsponi}. For each state $S=\{b_i \rightarrow \{0,1\} | i \in \{0,1,2\}\}$ we consider $\progdom= \{b_0\}$, $\highdom= \{b_1\}$ and $\lowdom= \{b_2\}$. 
We also assume that the states where $\prog$ is 1 are stuck and therefore are omitted.


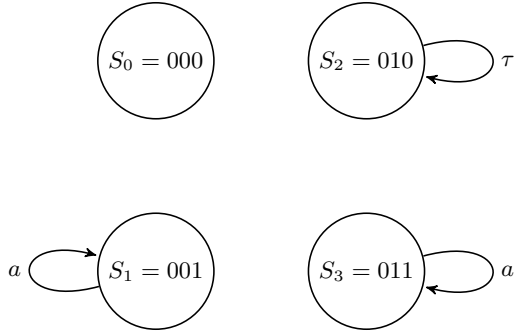
\begin{figure}[h]
\centering
\begin{tikzpicture}[->,>=stealth',shorten >=1pt,auto,node distance=2.8cm,
                    semithick]
  \tikzstyle{every state}=[text=black]

  \node[state] 	(A)                    {$S_0=000$};
  \node[state]     (B) [right of=A] {$S_2=010$};
  \node[state]     (D) [below of=A] {$S_1=001$};
  \node[state]     (C) [below of=B] {$S_3=011$};

  \path 
        (B) edge [loop right] node {$\tau$} (B)
        (C) edge [loop right] node {$a$} (C)
        (D) edge [loop left] node {$a$} (D);

\end{tikzpicture}
  \caption{StbS does not imply PoNI in general}\label{fig:stbsponi}
\end{figure}

The system is Strongly Trace-based Secure: all states are either stuck or perform a transition on themselves, therefore low equivalence is preserved. The only difference in the output behavior is observable between $S_0$ and $S_2$: the former is stuck, the latter perform a $\tau$ transition. Nonetheless they result indistinguishable in the termination transparent version of the system. However, the system is not PoNI. In fact, if a bit flip on $b_2$ can transform $S_2$ into $S_3$ so we have $S_0 \faultArrow{\{b_2\},\tau}$ but $S_2 \faultArrow{\{b_2\},a}$. 
\end{example}

From now onwards we focus our attention on ``standard fault-prone'' systems. For such systems, we have that the whole program component is fault-tolerant (formally  $\progdom \subseteq \tolerant$) and it is either stuck or active, regardless of the data component. 

\begin{definition}[Standard fault-prone systems]
A fault-prone system is called standard if $\progdom \subseteq \tolerant$ and, for all  $\prog$, either for any data component $\data$ there exists an action $l$ such that $\bistate{P}{\data}\labarrow{l}$ or for any data component $\data$ the system is stuck, namely $\bistate{P}{\data} \not -->$.
\end{definition}

For the class of systems of our interest, Strong Trace-based Security is indeed stronger than PoNI.

\begin{lemma}[Strong Trace-based Security implies PoNI]\label{thm:ssimplponi}
Let $\prog$ be a program component in a standard fault-prone system $\System$. If $\prog$ enjoys StbS then $\prog$ enjoys PoNI.
\end{lemma}
\begin{proof}

We prove this lemma by showing the contrapositive. Suppose that $\prog$ is not PoNI. Then there must be two states $S= \bistate{\prog}{\data_0}$ and $S'= \bistate{\prog}{N_0}$ in the augmented version of a fault-prone system $\System$ such that $\data_0=_\pubdata N_0$, and two runs that exit from them whose corresponding traces violate the security condition. Let $$r= \bistate{\prog}{\data_0} \faultArrow{L_0,a_0} \bistate{\prog_1}{\data_1} \dots \bistate{\prog_{j-1}}{\data_{j-1}} \faultArrow{L_{j-1},a_{j-1}}  \bistate{\prog_{j}}{\data_{j}} \faultArrow{L,a}$$ and 
$$r'= \bistate{\prog}{N_0} \faultArrow{L_0,b_0} \bistate{\prog^1}{N_1} \dots \bistate{\prog^{j-1}}{N_{j-1}} \faultArrow{L_{j-1},b_{j-1}}  \bistate{\prog^{j}}{N_{j}} \faultArrow{L,b}$$

be the runs in question, such that $\forall 0 \leq i < j$ $\low(a_i)=\low(b_i)$ but $\low(a) \not = \low(b)$. 

Before continuing, we observe that, in the initial configuration, $\prog$ cannot be stuck. Also, it must be that at most one run between $r$ and $r'$ contains a sequence of stuck configurations. Both conditions are  necessary to have $\low(a) \not = \low(b)$. 

We now show that it is possible to build two transition traces for $\prog$ that violate Strong Trace-based Security. Recall that the $\flip$ function can be applied only to locations in $\faulty$, and that we consider systems whose faulty locations are restricted to the data component ($\faulty \subseteq \lowdom \cup \highdom$). Hence, when $S=(\prog,\data)$, we write $\flip(S,L)$ as $(\prog,\flip(\data,L))$, and we focus on the data component $\flip(\data,L)$ when necessary. 

We proceed by distinguishing two cases, depending on whether or not a stuck configuration is traversed by $r$ (equivalently $r'$).

\emph{Case 1: no stuck configurations are traversed in either $r$ or $r'$}.

Consider the following two transition traces 
$$t= (\memflip(\data_0,L_0),a_0,\data_1),(\memflip(\data_1,L_1),a_1,\data_2)\dots(\memflip(\data_j,L),a,\data_{j+1})$$
and
$$t'= (\memflip(N_0,L_0),b_0,N_1),(\memflip(N_1,L_1),b_1,N_2)\dots (\memflip(N_j,L),b,N_{j+1})$$

Since $r$ does not traverse stuck configurations, all transitions are computed by an application of the rule $\wraprule{\rulename{Step}}$. This means that for any transition we have  $\bistate{\prog_{i}}{\data_{i}} \faultArrow{L_{i},a_{i}}  \bistate{\prog_{i+1}}{\data_{i+1}}$ if $\bistate{\prog_{i}}{\memflip(\data_{i},L_i)} \labarrow{a_{i}}  \bistate{\prog_{i+1}}{\data_{i+1}}$. Hence $t$ is in $\transtracena(\prog)$. By applying a similar argument for $r'$, we conclude that $t'$ is in $\transtracena(\prog)$ as well. 

Observe that $\memflip$ preserves low equivalence between data components (if $\data_i =_\pubdata N_i$ then for all set of locations $L$ it is true that $\memflip(\data_i,L)=_\pubdata \memflip(N_i,L))$. Considering that $\data_0 =_\pubdata N_0$, there are two possible cases. Either there exists $k$, such that  $0 \leq k < j$ and $(\memflip(\data_k,L_k),a_k,\data_{k+1})$ and $(\memflip(N_k,L_k),b_k,N_{k+1})$ and $\data_{k+1} \not =_\pubdata N_{k+1}$, or $\forall k 0 \leq k \leq j$ $\data_k \ =_\pubdata N_k$ and $\low(a) \not = \low(b)$. In both cases Strong Trace-based Security is violated.

\emph{Case 2: there is a stuck configuration in $r$.}

We consider the case in which a configuration in $r$ is stuck. The symmetric case for $r'$ is similar, and it is omitted. 

Since $\System$ is a ``standard fault-prone'' system, the rule $\wraprule{\rulename{Stuck-2}}$ cannot be applied in the first step. Let $1 \leq w \leq j$ be the index of the first stuck state $\bistate{\prog_w}{\data_w}$ in $r$. Consider the following transition traces
$$
\begin{array}{lcl}
t & = & (\memflip(\data_0,L_0),a_0,\data_1),\dots,(\memflip(\data_{w-1},L_{w-1}),a_w,\data_{w}),\\
	&	& (\memflip(N_w,L_w),\tau,\memflip(N_w,L_w)),\dots,(\memflip(N_j,L),\tau,\memflip(N_j,L))\\
t' & = & (\memflip(N_0,L_0),b_0,N_1),\dots, (\memflip(N_{w-1},L_{w-1}),b_{w-1},N_w),\\
	&	& (\memflip(N_w,L_w),b_w,N_{w+1}), \dots ,(\memflip(N_j,L),b,N_{j+1})
\end{array}
$$

As observed in the previous case, since $\memflip$ preserves low equivalence of data components, there are the following cases to be considered. Either there exists $k$ such that $0 \leq k < w$ and $(\memflip(\data_k,L_k),a_k,\data_{k+1})$ and $(\memflip(N_k,L_k),b_k,N_{k+1})$ and $\data_{k+1} \not =_\pubdata N_{k+1}$, or there exists $k$ such that $w \leq k < j$ and $(\memflip(N_k,L_k),\tau,\memflip(N_k,L_k))$ and $(\memflip(N_k,L_k),b_k,N_{k+1})$ and $\memflip(N_k,L_k) \not =_\pubdata N_{k+1}$, or $\tau \not = \low(b)$. In all cases Strong Trace-based Security is violated.

\end{proof}

\begin{proof}[Proof of Theorem \ref{thm:ssimpliesponi}]
Directly obtained by applying Lemma \ref{lemma:SStoStbS} and Lemma \ref{thm:ssimplponi}.
\end{proof}

\section{\assem\ instructions semantics}\label{app:fullassemsema}

Figure \ref{table:abssemantics} formalizes the semantics of the \assem\ language.

\begin{figure}[t] 
\setnamespace{0pt}\setpremisesend{0.25em}\setpremisesspace{0.7em}
\[
\begin{array}{@{}c@{}}
\inference[\rulename{Load}]{ 
P(\pc)= \loadName\ r\ p 
}
{\absState{P}{\RegName}{\HeapName} 
 \asem{\tau} 
 \absStatePost{P}{\RegName[r \update \Heap{p}]}{\HeapName}
}
\\[3ex] 
\inference[\rulename{Store}]{ 
P(\pc)= \storeName\ p\ r 
}
{\absState{P}{\RegName}{\HeapName} 
\asem{\tau}
 \absStatePost{P}{\RegName}{\HeapName[p \update \RegName(r)]}
}
\\[3ex] 
\inference[\rulename{Jmp}]{ 
P(\pc)= \jmpName\ l 
}
{\absState{P}{\RegName}{\HeapName} 
\asem{\tau} 
 \absStateJump{P}{\pc\update \res_\CodeName(l)}{\RegName}{\HeapName}
}
\\[3ex] 
\inference[\rulename{Jz-S}]{ 
P(\pc)= \jzName \ l\ r &
\RegName(r) = 0 
}
{\absState{P}{\RegName}{\HeapName} 
\asem{\tau} 
\absStateJump{P}{\pc\update \res_\CodeName(l)}{\RegName}{\HeapName}
}
\\[3ex] 
\inference[\rulename{Jz-F}]{ 
P(\pc)= \jzName \ l\ r &
\RegName(r) \not =0 
}
{\absState{P}{\RegName}{\HeapName} 
\asem{\tau} 
 \absStatePost{P}{\RegName}{\HeapName}
}
\\[3ex] 
\inference[\rulename{Nop}]{ 
P(\pc)= \nop  
}
{
\absState{P}{\RegName}{\HeapName} 
\asem{\tau} 
 \absStatePost{P}{\RegName}{\HeapName}
}
\\[3ex] 
\inference[\rulename{Move-k}]{ 
P(\pc)= \movek{r}{n}&
}
{\absState{P}{\RegName}{\HeapName} 
\asem{\tau} 
 \absStatePost{P}{\RegName[r \update n]}{\HeapName}
}
\\[3ex] 

\inference[\rulename{Move-r}]{ 
P(\pc)= \mover{r}{r'}&
}
{\absState{P}{\RegName}{\HeapName} 
\asem{\tau} 
 \absStatePost{P}{\RegName[r \update \RegName(r')]}{\HeapName}
}
\\[3ex] 
\inference[\rulename{Op}]{ 
P(\pc)= \aopName\ {r}\ {r'} 
}
{\absState{P}{\RegName}{\HeapName} 
\asem{\tau}
 \absStatePost{P}{\RegName[r \update \RegName(r) \ \aopName \ \RegName(r')]}{\HeapName}
}
\\[3ex] 
\inference[\rulename{Out}]{ 
P(\pc)= \outName\ ch \  r  & 
Reg(r) =  n 
}
{
\absState{P}{\RegName}{\HeapName} 
\asem{ch!n} 
\absStatePost{P}{\RegName}{\HeapName}
}
\end{array}
\]
\caption{\assem\ instructions semantics}\label{table:abssemantics}
\end{figure}

\section{Typing rule for expressions}\label{app:allabsrulexp}

Figure \ref{table:abstypesystemexpr} shows all the (abstract) rules for the compilation of \whprog\ expressions into \assem\ code.

\begin{figure*}[h]
\centering
$\begin{array}{c}
\vspace{0.5cm}

\Scale[0.95]{
\inference[\rulename{K}]
{r \not \in \activeReg}
{\regenv, \activeReg, l ||- k \etypeproduce [l:\ \movekName\ r\ k], \esecan{\seclev(r)}{1}, r, \subst{\regenv}{\breakconn{r}}}}\\

\vspace{0.5cm}

\Scale[0.95]{
\inference[\rulename{V}-cached]
{\conn{r}{x}{} \in \regenv}
{\regenv, \activeReg, l ||- x \etypeproduce [l:\ \emptya], \esecan{\seclev(r)}{0}, r, \regenv}}\\

\vspace{0.5cm}
\Scale[0.95]{
\inference[\rulename{V}-uncached]
{r \not \in \activeReg & \seclev(r) \sqsupseteq \seclev(x)} 
{\regenv, \activeReg, l ||- x \etypeproduce [l:\ \loadName\ r\ \vartoptr(x)], \esecan{\seclev(r)}{1}, r, \subst{\regenv}{\conn{r}{x}{\readMode}}}}\\

\Scale[0.95]{
\inference[\rulename{C}]
{\regenv, \activeReg, l ||- E_1 \etypeproduce \CodeName_1, \esecan{\lambda}{n_1}, r, \regenv_1 & 
\regenv_1, \activeReg\cup\{r\}, \emptyLab ||- E_2 \etypeproduce \CodeName_2, \esecan{\lambda}{n_2}, r', \regenv_2}
{\regenv, \activeReg, l ||- \wop{E_1}{E_2} \etypeproduce \concat{\concat{\CodeName_1}{\CodeName_2}}{[\binOp\ r\ r']}, \esecan{\lambda}{n_1+n_2+1}, r, \subst{\regenv_2}{\breakconn{r}}}}\\

\end{array}$
\caption{(Abstract) Type system for \whprog\ expressions\label{table:abstypesystemexpr}}
\end{figure*}

\section{Typing rule for atomic statements}\label{app:absatomicstat}

All rules for typing atomic statements are presented in Figure \ref{table:absttypeatom}.
\begin{figure*}[htp]
\centering
$\begin{array}{c}
\vspace{0.2cm}
\inference[\rulename{\wskip}]
{ - }
{\typeconc{\regenv}{l}{\wskip}{[l:\ \nop]}{\emptyLab}{\regenv}{\secan{\tconst~1}{\whigh}}}
\\
\vspace{0.2cm}
\inference[\rulename{\assignName}]
{\regenv, \emptyRegAll, l ||- E \etypeproduce \CodeName, \esecan{\seclev(x)}{n}, r, \regenv'
& \regenv''=\subst{\regenv'}{\conn{r}{x}{\writeMode}}} 
{\typeconc{\regenv}{l}{\assign{x}{E}}{
 \left \{ 
\begin{array}{l}
\CodeName \concatprograms \\ 
{[}\storeName\ \vartoptr(x) \ r{]}
\end{array} \right \}}{\emptyLab}{\regenv''}{
\begin{array}{rcl}
\seclev(x)=\Hloc & ? & \secan{\tconst~n+1}{\whigh}\\
			  & :  & \secan{\tdontc}{\wdontk}\\
\end{array}
}}
 \\
\vspace{0.2cm}
\inference[\rulename{\outName}]
{\regenv, \emptyRegAll,l ||- E \etypeproduce \CodeName, \esecan{\seclev(ch)}{n} , r, \regenv'}
{\typeconc{\regenv}{l}{ \out{ch}{E}}{ 
 \left \{ 
\begin{array}{l}
\CodeName \concatprograms \\ 
{[}\out{ch}{r}{]}
\end{array} \right \}}
{\emptyLab}{\regenv'}{
\begin{array}{rcl}
\seclev(ch)=\Hloc & ? & \secan{\tconst~n+1}{\whigh} \\
			     & : & \secan{\tdontc}{\wdontk}
\end{array}
}}
\end{array}$
\caption{(Abstract) Type system for atomic \whprog\ commands}\label{table:absttypeatom}
\end{figure*}

\section{Proof of Theorem \ref{thm:strsenforce}}\label{sec:realtypesystem} 

The type system presented in Section \ref{sec:absttypesystem} is an abstraction of the method that we use to enforce Strong Security over \assem\ programs. In this section we describe such method, together with all the necessary results to prove the Theorem \ref{thm:strsenforce}.

We begin (Section \ref{Sec:Wh-iWh}) by introducing the \swhprog\ language, an intermediate representation that bridges the translation between \whprog\ and \assem\ programs. We also define the type system, whose abstraction is presented in Section \ref{sec:absttypesystem}, that translates \whprog\ programs into \swhprog\ programs. For a type correct \whprog\ program, the type system ensures that (i) it is strongly secure (Proposition \ref{prop:strongsectype}) and (ii) is mapped to a strongly secure \swhprog\ program (Proposition \ref{thm:typeprediction}).

In Section \ref{Sec:iWh-asse} we define a strategy for translating \swhprog\ programs into \assem\ programs. The strategy ensures that the target \assem\ program is semantically equivalent to the source \swhprog\ program (Proposition \ref{prop:iwhileandrisc}), hence Strong Security can be inferred for the \assem\ program via Strong Security enjoyed by the source \swhprog\ program.

\subsection{From \whprog\ programs to \swhprog\ programs}\label{Sec:Wh-iWh}

The syntax of \swhprog\ programs is presented in Figure \ref{table:swhilesyntax}. 

\begin{figure}[htb]
\centering
$\begin{array}{lc  l c l c l c l c l l l }
D 		& ::=		& \wskip & | & \assign{x}{r} & | & \assign{r}{F} & |  & D_1; D_2 & | & \out{ch}{r} \ \ | \\
		&		& \multicolumn{9}{l}{\while{r}{\{D;\wskip\}} \ \ | \  \  \wif{r}{\{D_1;\wskip\}}{\{D_2;\wskip\}}} \\
F		& ::=		& k \in \mathbb{N} & | &  r \in \regvar & | & x \in \var & | & \wop{r_1}{r_2}\\
ch		& ::=		& \low				& | & \high 	&&&&&\\
\end{array}$
\caption{\swhprog\ programs syntax\label{table:swhilesyntax}}
\end{figure}

\swhprog\ programs are a subclass of \whprog\ programs that take into account features of the \assem\ architecture. The first feature is that an \swhprog\ program distinguishes between \emph{pure} and  \emph{register} variables, respectively in the sets $\var$ and $\regvar$. This distinction equips the \swhprog\ language with the concept of ``register" that is used in the target machine. The second feature is a more restrictive syntax for conditionals and loops, that enables a simpler translation between \swhprog\ and \assem\ programs. Specifically, both $\whileName$ and $\wifName$ commands in \swhprog\ syntax require the guard to be a register variable, and require the branches (in the case of $\wifName$) or the loop body (in the case of $\whileName$) to be terminated by a $\wskip$ command.

\begin{figure}[htp]
\centering
$\begin{array}{c}
\vspace{0.5cm}

\Scale[0.93]{
\inference[\rulename{K}]
{r \in \regvar & r \not \in \activeReg}
{\regenv, \activeReg \eenta{\regvar}{\var} k \etypeproduce \code{\assign{r}{k}}, \esecan{\seclev(r)}{1}, r, \subst{\regenv}{\breakconn{r}}}}\\

\vspace{0.5cm}

\Scale[0.93]{
\inference[\rulename{V}-cached]
{x \in \var & r \in \regvar &  \conn{r}{x}{} \in \regenv } 
{\regenv, \activeReg \eenta{\regvar}{\var} x \etypeproduce \code{\findot},\esecan{\seclev(r)}{0}, r, \regenv }}\\
\vspace{0.5cm}

\Scale[0.93]{
\inference[\rulename{V}-uncached]
{x \in \var & r \in \regvar & r   \not \in \activeReg & \seclev(r) \sqsupseteq \seclev(x)} 
{\regenv, \activeReg \eenta{\regvar}{\var} x \etypeproduce \code{\assign{r}{x}}, \esecan{\seclev(r)}{1}, r, \subst{\regenv}{\conn{r}{x}{\readMode}}}}\\
\vspace{0.5cm}

\Scale[0.93]{
\inference[\rulename{C}]
{r,r' \in \regvar \\
\regenv, \activeReg \eenta{\regvar}{\var} E_1 \etypeproduce \code{D_1}, \esecan{\lambda}{n_1}, r, \regenv_1 \\ 
\regenv_1, \activeReg\cup\{r\} \eenta{\regvar}{\var} E_2 \etypeproduce \code{D_2}, \esecan{\lambda}{n_2}, r', \regenv_2} 
{\regenv, \activeReg \eenta{\regvar}{\var} \wop{E_1}{E_2} \etypeproduce \code{D_1;D_2;\assign{r}{\wop{r}{r'}}}, \esecan{\lambda}{n_1+n_2+1}, r, \subst{\regenv_2}{\breakconn{r}}}}\\
\end{array}$
\caption{Type system for \whprog\ expressions\label{table:exptypesystem}}
\end{figure}

In Figures \ref{table:exptypesystem}, \ref{table:typesystematom} and \ref{table:typesystem} the type system that performs both the translation of a \whprog\ program $C$ into an \swhprog\ program $D$ and the security analysis on $C$ is described.  We now discuss these aspects in details.

\begin{figure}[htp]
\centering
$\begin{array}{c}

\vspace{0.5cm}
\inference[\rulename{\wskip}]
{ - }
{\regenv \enta{\regvar}{\var}  \wskip \typeproduce \code{\wskip}, \secan{\tconst~(1,1)}{\whigh}, \regenv}\\

\vspace{0.5cm}
\Scale[0.94]{
\inference[\rulename{\assignName}]
{x \in \var &r\in \regvar & \regenv, \emptyRegAll \eenta{\regvar}{\var} E \etypeproduce \code{D}, \esecan{\seclev(x)}{n}, r, \regenv' }
{
\regenv \enta{\regvar}{\var} \assign{x}{E} \typeproduce \code{D;\assign{x}{r}}, 
\condtype{\seclev(x)=\valHigh}{\secan{\tconst~(1,n+1)}{\whigh}}{\secan{\tdontc}{\wdontk}}, 
\subst{\regenv'}{\conn{r}{x}{\writeMode}}}}\\


\vspace{0.5cm}
\Scale[0.94]{
\inference[\rulename{\outName}]
{ r\in \regvar & \regenv, \emptyRegAll \eenta{\regvar}{\var} E \etypeproduce \code{D}, \esecan{\seclev(ch)}{n}, r, \regenv'}
{
\regenv \enta{\regvar}{\var} \out{ch}{E} \typeproduce \code{D;\out{ch}{r}}, 
\condtype{\seclev(ch)=\valHigh}{\secan{\tconst~(1,n+1)}{\whigh}}{\secan{\tdontc}{\wdontk}}, 
\regenv'
}}\\
\end{array}$
\caption{Type system for atomic \whprog\ statements\label{table:typesystematom}}
\end{figure}

\begin{figure}[htp]
\centering
$\begin{array}{c}
\vspace{0.5cm}
\inference[\rulename{;}]
{\regenv \enta{\regvar}{\var}  C \typeproduce \code{D}, \secan{t_1}{w_1}, \regenv_1 &
\regenv_1 \enta{\regvar}{\var}  C' \typeproduce \code{D'}, \secan{t_2}{w_2}, \regenv_2 \\
t_1=\ttop \Rightarrow w_2=\whigh
}
{\regenv \enta{\regvar}{\var}  C;C' \typeproduce \code{D;D'}, \secan{t_1 \tclub t_2}{w_1 \wlub w_2},\regenv_2}\\

\vspace{0.5cm}
\Scale[0.9]{
\inference[\rulename{\wifName}-any]
{r \in \regvar \\
\regenv, \emptyRegAll \eenta{\regvar}{\var} E \etypeproduce \code{D_g}, \esecan{\lambda}{n_g} , r, \regenv_1\\ 
\regenv_1 \enta{\regvar}{\var}  C_t \typeproduce \code{D_t}, \secan{t_1}{w_1}, \regenv_2 &
\regenv_1 \enta{\regvar}{\var}  C_e \typeproduce \code{D_e}, \secan{t_2}{w_2}, \regenv_3 \\
D_t'=D_t;\wskip & 
D_e'=D_e;\wskip \\
\writeMap(\lambda) \sqsupseteq w_i &
\regenv_F=\regenv_2 \envinters \regenv_3
}
{\typeconce{\regenv}{
\wifa{E}{C_t}{C_e}
}{
\left \{ 
\begin{array}{l}
D_g; \\
\wifa{r}{D_t'}{D_e'}
\end{array} 
\right \}
} 
{
\left \langle 
\begin{array}{l}
\writeMap(\lambda) \wlub w_1 \wlub w_2\\
\termMap(\lambda)\tlub t_1 \tlub t_2 
\end{array} 
\right \rangle
}
{\regenv_F}{\enta{\regvar}{\var}}}}\\

\vspace{0.5cm}

\Scale[0.9]{
\inference[\rulename{\wifName}-\valHigh]
{r \in \regvar \\
\regenv, \emptyRegAll \eenta{\regvar}{\var} E \etypeproduce D_g, \esecan{\valHigh}{n_g}, r, \regenv_1 \\ 
\regenv_1 \enta{\regvar}{\var}  C_t \typeproduce \code{D_t}, \secan{\tconst~(m,n_t)}{\whigh} , \regenv_2&
\regenv_1 \enta{\regvar}{\var}  C_e \typeproduce \code{D_e}, \secan{\tconst~(m,n_e)}{\whigh}, \regenv_3 \\
D_t'=D_t;\wskip^{n_e-n_t};\wskip &
D_e'=D_e;\wskip^{n_t-n_e};\wskip \\
\regenv_F=\regenv_2 \envinters \regenv_3
}
{\typeconce{\regenv}{\wifa{E}{C_t}{C_e}}{
\left \{ 
\begin{array}{l}
D_g; \\
\wif{r}{D_t'}{D_e'}
\end{array} 
\right \}
} 
{
\left \langle 
\begin{array}{c}
\whigh \\
\tconst~(m+1,\\
n_g+max(n_t,n_e)+2) 
\end{array} 
\right \rangle
}
{\regenv_F}{\enta{\regvar}{\var}}}}\\


\Scale[0.9]{
\inference[\rulename{\whileName}]
{x \in \var & r \in \regvar \\
D_0=\assign{r}{x};\assign{x}{r}; \\
\regenv'=\subst{\regenv}{\conn{r}{x}{\writeMode}} &
\regenv_B \sqsubseteq \regenv' & \regenv_B \sqsubseteq \subst{\regenv_E}{\conn{r}{x}{\writeMode}}\\
\regenv_B \enta{\regvar}{\var}  C \typeproduce \code{D}, \secan{t}{w} , \regenv_E\\
\lambda = \seclev(x)=\seclev(r) & 
\writeMap(\lambda) \sqsupseteq w & t=\ttop \Rightarrow \writeMap(\lambda)=\whigh
}
{\typeconce{\regenv}{\while{x}{C}}{
\left \{ 
\begin{array}{l}
D_0; \\
\while{r}{\code{D;D_0;\wskip}}
\end{array} 
\right \}
} 
{
\left \langle 
\begin{array}{c}
\writeMap(\lambda) \wlub w \\
\termMap(\lambda)\tlub t
\end{array} 
\right \rangle
}
{\regenv_B}{\enta{\regvar}{\var}}}}\\

\end{array}$
\caption{Type system for non-atomic \whprog\ statements\label{table:typesystem}}
\end{figure}

\subsubsection{Type System: Translation}\label{ts:transl}

The type system described in Figures \ref{table:exptypesystem}, \ref{table:typesystematom}  and \ref{table:typesystem} converts a \whprog\ program into an \swhprog\ program. 

In general, observe that none of the rules involve labels, since these are no longer part of the target language. However, the rules still require that a register record is carried along the compilation. We assume, similarly to what has been done in Section \ref{sec:absttypesystem}, that there is a register record  function $\regenv \in \regvar \rightharpoonup (\{\readMode,\writeMode \} \times \var)$ which not only associates variables to registers, but also records the modality (read/write) in which the association was created. Notice that in this context $\regvar$ are just variables and have not direct hardware interpretation. Finally, observe that the rules are parametrized in the set of pure and register variables ($\var$ and $\regvar$). Implementing our translation strategy does not require this parametrization, since the set of registers that are used in the \assem\ machine is known beforehand. However it turns out to be a useful tool for implementing our proof strategy: first we show that a type-correct \whprog\ program $C$ is strongly secure, then we show that the correspondent \swhprog\ program $D$ is strongly secure by retyping it under a different set of register variables, that are solely used for stating the security property of $D$ (all the details are formalized in Proposition \ref{thm:typeprediction}).  

The rules in Figure \ref{table:exptypesystem}, that formalize the compilation of expressions, are very similar to the rules in Figure \ref{table:abstypesystemexpr}, beside the fact that operations over registers and memory locations are replaced by assignments between pure and register variables. In order to represent a compilation that produces no code (see rule $\wraprule{\rulename{V}-cached}$) we explicitly extend the syntax of \swhprog\ programs (i.e. \whprog\ programs) with $\findot$, the empty statement, that  represents the \swhprog\ correspondent of the $\emptya$ statement used for \assem\ programs. However, since $\findot$ has no semantic meaning, we assume that the composition operator $;$ strips the occurrences of $\findot$ when composing commands together. This is formalized by defining \emph{structural equivalence} to be the least congruence\footnote{A congruence relation, in this context, is an equivalence relation on commands that is closed under the syntactic constructors of the language.} relation $\equiv$ between programs such that $\findot; C \equiv C ; \findot \equiv C$. 

Commands translation, implemented by the rules in Figures \ref{table:typesystematom} and \ref{table:typesystem}, are also similar to the rules in Figures \ref{table:absttypeatom}, \ref{table:absttypeif} and \ref{table:absttypewhcomp}, hence we only focus on the main differences.

The compilation of $\wif{E}{C_t}{C_e}$ performed by the $\wraprule{\rulename{\wifName}-any}$ rule produces the code $D_g$ for $E$ first, that corresponds to evaluating $E$ in register variable $r$. Then, subprograms $D_t$ and $D_e$ are obtained from $C_t$ and $C_e$ respectively. In order to comply with the $\swhprog$ syntax, a $\wskip$ instruction is appended to both $D_t$ and $D_e$, obtaining respectively programs $D_t'$ and $D_e'$ that are used in the rule output $D_g; \wif{r}{D_t'}{D_e'}$.  In the rule $\wraprule{\rulename{\wifName}-\valHigh}$ we follow a padding strategy which is similar to the one used for \assem\ programs: we say $\wskip^n$ is a sequence of $n$ $\wskip$ commands when $n>0$, whereas it is $\findot$ for $n \leq 0$.  
In the compilation of a $\while{x}{C}$ command, we deploy the fragment of code $D_0$ to establish the invariant property on the register record. Moreover, we append $\wskip$ to $D;D_0$ in order to make the code compliant with $\swhprog$ syntax.


\subsubsection{Type System: Security analysis}\label{ts:sec}

In order to reason about security of \whprog\ programs, we distinguish between public and secret data. In particular, as we did in Section \ref{sec:absttypesystem}, we assume that variables are labeled in a way that does not change during the execution and it is determined by the function $\seclev \in \regvar \cup \var \rightarrow \{\Lloc, \Hloc\}$.  

The security information calculated by the type system is essentially the same one presented for the type system in Section \ref{sec:absttypesystem}. A type-correct \whprog\ expression is decorated with a label $(\lambda,n)$, which indicate the security level of the register variables that are used to evaluate $E$ and the number of \swhprog\ steps that are required to evaluate $E$, exactly as in Figure \ref{table:abstypesystemexpr}. A type-correct \whprog\ program is associated to a label $(w,t)$ describing its write effect (label $w$) and timing behavior (label $t$). Write labels are taken from the
two-element set $\Set{\whigh,\wdontk}$, with partial ordering
$ \whigh \sqsubseteq  \wdontk $,  and have the same semantics introduced in Section \ref{sec:absttypesystem}.  
The timing label $t$ is an element from the partial order $\termlattice$ which is described in Figure \ref{fig:exttermlattice} (notice that, for improving readability, the order relation is just sketched). Compared to the  partial order presented in Figure \ref{fig:termlattice}, we define a refined annotation for programs that are certainly terminating in a finite number of steps. In particular we label such programs with values $\tconst~(m,n)$, which specifies that (i) the source \whprog\ program terminates in $m$ steps and (ii) the compiled \swhprog\ program terminates in $n$ steps, no matter which values the variables are set to. As for rules parametrization, this feature is a tool for our proof strategy. In fact, the type system has to track timing behavior of both the original program and its compiled version because it might happen that the timing behavior is modified in the recompilation of  a compiled program. 

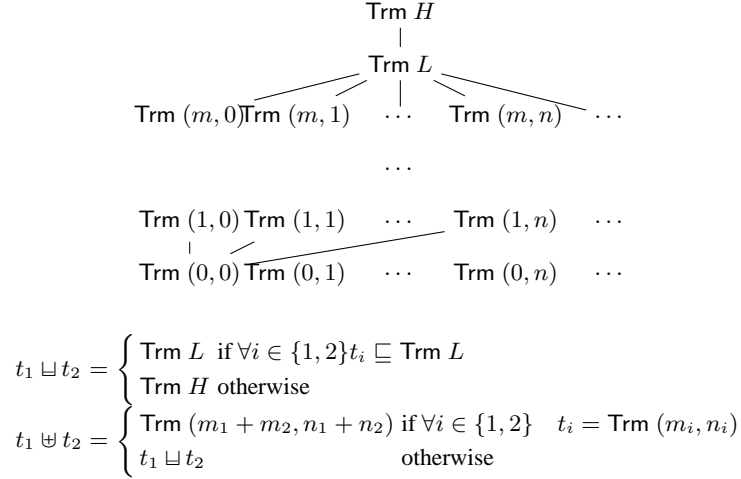
\begin{figure}[htbp]
\centering
\begin{minipage}{0.9\textwidth}
\centering
\begin{tikzpicture}[scale=.7]
  \node (top) at (0,0) {$\ttop$};
  \node (mid) at (0,-1) {$\tdontc$};
  \node (bottom1a) at (-4,-5) {$\tconst~(0,0)$};
  \node (bottom2a) at (-2,-5) {$\tconst~(0,1)$};
  \node (bottom3a) at (0,-5) {$\dots$};
  \node (bottom4a) at (2,-5) {$\tconst~(0,n)$};
  \node (bottom5a) at (4,-5) {$\dots$};
   \node (bottom1b) at (-4,-4) {$\tconst~(1,0)$};
  \node (bottom2b) at (-2,-4) {$\tconst~(1,1)$};
  \node (bottom3b) at (0,-4) {$\dots$};
  \node (bottom4b) at (2,-4) {$\tconst~(1,n)$};
  \node (bottom5b) at (4,-4) {$\dots$};
  \node (middle) at (0,-3) {$\dots$};
   \node (bottom1c) at (-4,-2) {$\tconst~(m,0)$};
  \node (bottom2c) at (-2,-2) {$\tconst~(m,1)$};
  \node (bottom3c) at (0,-2) {$\dots$};
  \node (bottom4c) at (2,-2) {$\tconst~(m,n)$};
  \node (bottom5c) at (4,-2) {$\dots$};
  \draw (top) -- (mid)  -- (bottom1c); 
  \draw (mid)  -- (bottom2c);
   \draw (mid)  -- (bottom3c);
    \draw (mid)  -- (bottom4c);
     \draw (mid)  -- (bottom5c);
     \draw (bottom1a) -- (bottom1b);
     \draw (bottom1a) -- (bottom2b);
     \draw (bottom1a) -- (bottom4b);
\end{tikzpicture}  
\vspace{0.5cm}
\end{minipage}
\begin{minipage}{0.9\textwidth}
\centering
$\begin{array}{lcl}
t_1 \tlub t_2 	& = & 
\left\{
	\begin{array}{ll}
		\tdontc  		& \mbox{if }  \forall i \in \Set{1,2} t_i \sqsubseteq \tdontc \\
		\ttop 	& \mbox{otherwise}\\
	\end{array}
\right. \\
t_1 \tclub t_2	& = &
\left\{
	\begin{array}{ll}
		\tconst~(m_1+m_2,n_1+n_2)  	& \mbox{if } \forall i \in \Set{1,2}\quad t_i= \tconst~(m_i,n_i) \\
		t_1 \tlub t_2 & \mbox{otherwise}\\
	\end{array}
\right. \\
\end{array}$
\end{minipage}
\caption{Termination partial order\label{fig:exttermlattice}}
\end{figure}

We now explain how security annotations are computed by the type system in Figures \ref{table:typesystematom} and \ref{table:typesystem}.

The $\wskip$ commands takes exactly one step to be completed in both its original and its compiled version. Moreover, it does not modify the value of any of the variables used in the program. For this reason, $\wraprule{\rulename{\wskip}}$ rule assigns $\secan{\tconst~(1,1)}{\whigh}$ as security annotation for $\wskip$.

The security annotation for an assignment command $\assign{x}{E}$ depends on the security level of the variable $x$. If $\seclev(x)= \valHigh$, we require that all instructions used to evaluate the expression perform write actions solely on $\valHigh$ variables. This, in turn, requires that $E$ has an associated type $\esecan{\valHigh}{n}$, which not only specifies the information about written variables, but also states that $E$ has been compiled into $n$ \swhprog\ instructions. Hence, for $\seclev(x)= \valHigh$, rule $\wraprule{\rulename{\assignName}}$ assigns $\secan{\tconst~(1,n+1)}{\whigh}$ to $\assign{x}{E}$, since a single instruction is mapped to $n+1$ \swhprog\ instructions ($n$ instructions correspond to $E$, the last instruction $\assign{x}{r}$ counts for one). If $\seclev(x)= \valLow$, then we require that written variables are at security level $\valLow$\footnote{In fact security could be established in a more liberal setting, however this strictness simplifies arguments in proofs.}, beside making sure that no content from $\valHigh$ variables is used. In this case the final type for $\assign{x}{E}$ is $\secan{\tdontc}{\wdontk}$. 

The rule $\wraprule{\rulename{\outName}}$ follows a similar argument, with the role of $x$ taken by the channel $ch$.

The rule $\wraprule{\rulename{;}}$ requires that whenever a component $C_1$ has a timing behavior described by $\ttop$, the following component $C_2$ induces a $\whigh$ write effects, in order to avoid timing channel leaks. The annotation computed by rule $\wraprule{\rulename{;}}$ considers the least upper bound of the writing effects of $C_1$ and $C_2$, and uses an extended least upper bound operator $\tclub$ (cf. Figure \ref{fig:exttermlattice}).

The rule $\wraprule{\rulename{\wifName}-any}$ prevents implicit flows from happening in the program by applying the same strategy presented for the type system in Figure \ref{table:absttypeif}. In particular, the security label $\lambda$ of an expression is translated into a write effect by the function $\writeMap$ (cf. Section \ref{sec:absttypesystem}), and write effects of both branches are expected to be lower than $\writeMap(\lambda)$. The resulting security annotation for the $\wif{E}{C_t}{C_e}$ command is computed in terms of the least upper bound of the security annotations for $E$, $C_t$ and $C_e$. In particular, the write effect $\writeMap(\lambda) \wlub w_1 \wlub w_2$ corresponds to the least upper bound of all write effects, whereas the time behavior $\termMap(\lambda)\tlub t_1 \tlub t_2$ is determined from the time behavior of branches, together with the corresponding time behavior of the guard, according to the function $\termMap$ (where $\termMap$ is defined in Section \ref{sec:absttypesystem}). Notice that the exact label for termination would be $\tconst~(0, n_g) \tclub \termMap(\lambda) \tclub (t_1 \tclub \tconst~(0,1)) \tlub (t_2 \tclub \tconst~(0,1))$, because of the code $D_g$ that is executed at the loop entrance and the $\wskip$ commands that are appended at the end of $D_t$ and $D_e$. However, we can omit this information because $\termMap$ always returns a label in $\{\tdontc,\ttop\}$, hence the exact number of steps for $D_g$ and $\wskip$ is irrelevant. 

When it is known that the $\wifName$ statement involves only write actions on $\valHigh$ variables, the timing behavior can be computed more accurately, as for the type system in Figure \ref{table:absttypeif}. In particular, as shown by rule $\wraprule{\rulename{\wifName}- \valHigh}$, if $C_t$ and $C_e$ are associated to a security type $\secan{\tconst~(m,n_t)}{\whigh}$ and $\secan{\tconst~(m,n_e)}{\whigh}$, the resulting annotation for the statement becomes $\secan{\tconst~(m+1,\max(n_t,n_e)+n_g+2)}{\whigh}$. The first timing value, namely $m+1$, adds one expression evaluation step to the timing behavior of the branches, which is \emph{required} to have the same value. The second timing value, namely $\max(n_t,n_e)+n_g+2$ considers the expression compilation (factor $n_g$), the final $\wskip$ command and the register evaluation (they count for the factor 2), together with the biggest factor between $n_t$ and $n_e$ for branches. Notice that alignment of branches is performed by applying the usual padding strategy. 

As for the type system in Figure \ref{table:absttypewhcomp}, the rule $\wraprule{\rulename{\whileName}}$ prevents implicit flow from happening by enforcing several constraints. Indirect information flows involving the loop guard are prevented by requiring the write effect of the loop to be lower than the write effect of the guard. The timing channel induced by the loop body is prevented by requiring the write effect of the body to be $\whigh$ if its timing behavior depends on secrets.  

\subsubsection{Type System: Results}

In order to formalize the properties of the type system we proceed with specifying the semantics of the \whprog\ language. Since we are targeting the notion of Strong Security, we instantiate the abstract machine for the \whprog\ language as a fault-prone system. 

Observe that we formulate the semantic definition generally enough to be suitable for \swhprog\ programs as well. Since \swhprog\ programs are a subclass of \whprog\ programs, the only difference to be considered is the set of variables on which programs are defined in the two languages. We therefore define the rules in Figure \ref{table:whilesem} to be parametric on $\allvari$, the set of variables, such that $\allvari= \var$ for \whprog\ programs and $\allvari = \var \cup \regvar$ for \swhprog\ programs. 

We define the memory as an element $\mem$ from the set $\{\var \cup \regvar \rightarrow \mathbb{N}\}$ of mappings between variables and values in $\mathbb{N}$.  The semantics for the \whprog\ language as a fault-prone system is described by the LTS $\whlts=\{\{\langle C, \mem \rangle  \}, \whsem{}, \{ch!n|ch \in \{\low,\high\} \mbox{ and } n \in \mathbb{N}\} \cup \{\tau\} \}$ \, where $C$ is the fault-tolerant part of configurations, $\mem$ is the fault-prone part and $\whsem{}$ is obtained by rules in Figure \ref{table:whilesem}. Notice that rules for sequential composition are expressed using evaluation contexts, defined as $\ec ::= [-] | \ec;C$. The rule $\wraprule{\rulename{c-2}}$ requires a particular attention: since we need a fine control over the number of execution steps, rule $\wraprule{\rulename{c-2}}$ simultaneously executes statements that terminates in one step (like $\wskip, \assignName, \outName$) and stripes off the terminating command $\findot$ they are reduced to. In this way, the number of reductions performed by a \whprog\ program according to \whprog\ semantics are easily mapped to the number of steps performed by the corresponding \assem\ program in the \assem\ semantics.

\begin{figure}[h]
\centering
$\begin{array}{c}
\vspace{0.2cm}
\inference[\rulename{skip}]
{}
{\whstate{\wskip}{\mem}\whsem{\tau} \whstate{\findot}{\mem}
}
\\
\vspace{0.2cm}
\inference[\rulename{:=}]
{v \in \allvari & |[E_\allvari |](M)=n}
{\whstate{\assign{v}{E}}{\mem} \whsem{\tau} \whstate{\findot}{\mem[v \backslash n]}
}
\\
\vspace{0.2cm}
\inference[\rulename{out}]
{|[E_\allvari |](M)=n}
{\whstate{\out{ch}{E}}{\mem} \whsem{ch!n} \whstate{\findot}{\mem}
}
\\
\vspace{0.2cm}
\inference[\rulename{if-1}]
{|[E_\allvari |](M)\not = 0 }
{\whstate{\wif{E}{C_1}{C_2}}{\mem}\whsem{\tau} \whstate{C_1}{\mem}
}
\\
\vspace{0.2cm}
\inference[\rulename{if-2}]
{ |[E_\allvari  |](M) = 0}
{\whstate{\wif{E}{C_1}{C_2}}{\mem} \whsem{\tau} \whstate{C_2}{\mem}
}
\\
\vspace{0.2cm}
\inference[\rulename{w-1}]
{v\in \allvari &  M(v)\not = 0}
{\whstate{\while{v}{C}}{\mem}\whsem{\tau} \whstate{C;\while{v}{C}}{\mem}
}
\\
\vspace{0.2cm}
\inference[\rulename{w-2}]
{v\in \allvari & M(v) = 0 }
{\whstate{\while{v}{C}}{\mem} \whsem{\tau} \whstate{\findot}{\mem}
}
\\
\vspace{0.2cm}
\inference[\rulename{c-1}]
{\whstate{C}{\mem} \whsem{l} \whstate{C'}{\mem'} & C' \not = \findot}
{\whstate{\ec[C]}{\mem} \whsem{l} \whstate{\ec[C']}{\mem'}
}
\\
\vspace{0.2cm}
\inference[\rulename{c-2}]
{\whstate{C_1}{\mem} \whsem{l} \whstate{\findot}{\mem'}}
{\whstate{\ec[C_1;C_2]}{\mem} \whsem{l} \whstate{\ec[C_2]}{\mem'}
}
\\[3ex] 
\end{array}$
\caption{\whprog\ and \swhprog\ programs semantics\label{table:whilesem}}
\end{figure}

We want to show that any type-correct \whprog\ program is SS. First we instantiate the definition of Strong Security for \whprog\ programs.

\begin{definition}[Strong Security for \whprog\ programs]
We say that a \whprog\ program $C$ is strongly secure if it is strongly secure according to Definition \ref{def:strongsec} instantiated for the fault-prone system $\whlts$.
\end{definition}
 
The first property of the type system presented in Figures \ref{table:exptypesystem}, \ref{table:typesystematom} and \ref{table:typesystem} is that any type-correct program is strongly secure. We formalize this result as follows.

\begin{proposition}[Type System Enforces Strong Security]\label{prop:strongsectype}
Let $C$ be a \whprog\ program. If there exists $\regenv$ such that $\regenv \enta{\regvar}{\var}  C \typeproduce \code{D}, \secan{t}{w}, \regenv'$, then $C$ is strongly secure. 
\end{proposition} 

Proving this statement requires few auxiliary results. The first one, formalized in Lemma \ref{lemma:comweak},  illustrates a property of the type system which we refer to as ``upward closure'': if a command is type-correct with respect to a register record $\regenv$, it also results type-correct with respect to any other register record $\regenv'$ that is bigger than $\regenv$. Also, the output register record corresponding to $\regenv'$ results bigger than the one corresponding to $\regenv$. 

In order to prove Lemma \ref{lemma:comweak}, two auxiliary results are required. The first one formalizes some properties of the operations over register records.

\begin{lemma}\label{lemma:regenvoper}
Let $\regenv$ and $\regenv'$ two register records such that $\regenv' \sqsupseteq \regenv$. Then:
\begin{itemize}
\item $\subst{\regenv'}{\breakconn{r}} \sqsupseteq \subst{\regenv}{\breakconn{r}}$;
\item $\subst{\regenv'}{\conn{r}{x}{\mu}} \sqsupseteq  \subst{\regenv}{\conn{r}{x}{\mu}}$.
\end{itemize}
\end{lemma}
\begin{proof}
The first statement follows directly by the fact that inclusion is preserved by set difference. Proving the second statement is slightly more challenging, since the update operation requires the bijection property of register records to be preserved. In particular, implementing $\subst{\regenv}{\conn{r}{x}{\mu}}$ requires that (i) any existing association $\conn{r'}{x}{}$ is removed, and (ii) the new association $\conn{r}{x}{\mu}$ is added to the register record. The statement is proved by distinguishing the following cases:
\begin{itemize}
\item assume $x$ is not associated in $\regenv$. This implies that there is no $r' \in \regvar$ such that $\conn{r'}{x}{} \in \regenv$. We distinguish four cases:
\begin{itemize}
\item $x$ is not associated in $\regenv'$. Then  $\subst{\regenv'}{\conn{r}{x}{\mu}} \sqsupseteq  \subst{\regenv}{\conn{r}{x}{\mu}}$ is true because $\restr{\subst{\regenv'}{\conn{r}{x}{\mu}}}{\regvar \setminus \{r\}}= \restr{\regenv'}{\regvar \setminus \{r\}}$ and $\restr{\subst{\regenv}{\conn{r}{x}{\mu}}}{\regvar \setminus \{r\}}= \restr{\regenv}{\regvar \setminus \{ r \}}$.

\item there exists $r' \not = r$ such that $\conn{r'}{x}{} \in \regenv'$. Then $\subst{\regenv'}{\conn{r}{x}{\mu}} \sqsupseteq  \subst{\regenv}{\conn{r}{x}{\mu}}$ because $r'$ is unassociated in $\regenv$. 

\item $\conn{r}{x}{\mu'} \in \regenv'$ but $\mu' \not = \mu$. Then $\subst{\regenv'}{\conn{r}{x}{\mu}} \sqsupseteq  \subst{\regenv}{\conn{r}{x}{\mu}}$ because $\restr{\subst{\regenv'}{\conn{r}{x}{\mu}}}{\regvar \setminus \{r \}}= \restr{\regenv'}{\regvar \setminus \{r\}}$ and $\restr{\subst{\regenv}{\conn{r}{x}{\mu}}}{\regvar \setminus \{r\}}= \restr{\regenv}{\regvar \setminus \{ r \}}$.

\item $\conn{r}{x}{\mu} \in \regenv'$. Then $\subst{\regenv'}{\conn{r}{x}{\mu}}=\regenv'$. Hence $\subst{\regenv'}{\conn{r}{x}{\mu}}=\regenv' \sqsupseteq  \subst{\regenv}{\conn{r}{x}{\mu}}$ since $\restr{\subst{\regenv}{\conn{r}{x}{\mu}}}{\regvar \setminus \{r\}}= \restr{\regenv}{\regvar \setminus \{r\}}$.
\end{itemize} 
\item Assume $x$ is associated in $\regenv$, hence there exists $r' \in \regvar$ such that $\conn{r'}{x}{\mu'} \in \regenv$. Since $\regenv' \sqsupseteq \regenv$, $\conn{r'}{x}{\mu'} \in \regenv'$. We distinguish three cases:
\begin{itemize}
\item $r' \not =r$. Then $\subst{\regenv'}{\conn{r}{x}{\mu}}  \sqsupseteq  \subst{\regenv}{\conn{r}{x}{\mu}}$ since (i) $\restr{\subst{\regenv'}{\conn{r}{x}{\mu}}}{\regvar \setminus \{r,r'\}}= \restr{\regenv'}{\regvar \setminus \{ r,r' \}}$ and $\restr{\subst{\regenv}{\conn{r}{x}{\mu}}}{\regvar \setminus \{r,r'\}}= \restr{\regenv}{\regvar \setminus \{ r,r' \}}$ and (ii) $r'$ results unassociated in both $\subst{\regenv'}{\conn{r}{x}{\mu}}$ and $\subst{\regenv}{\conn{r}{x}{\mu}}$. 

\item $r'=r$ but $\mu' \not = \mu$. Then  $\subst{\regenv'}{\conn{r}{x}{\mu}} \sqsupseteq \subst{\regenv}{\conn{r}{x}{\mu}}$ holds because $\restr{\subst{\regenv'}{\conn{r}{x}{\mu}}}{\regvar \setminus \{r \}}= \restr{\regenv'}{\regvar \setminus \{r\}}$ and $\restr{\subst{\regenv}{\conn{r}{x}{\mu}}}{\regvar \setminus \{r \}}= \restr{\regenv}{\regvar \setminus \{r\}}$ .

\item If $r'=r$ and $\mu' = \mu$. Then $\subst{\regenv'}{\conn{r}{x}{\mu}} =  \regenv' \sqsupseteq  \regenv = \subst{\regenv}{\conn{r}{x}{\mu}}$.\\
\end{itemize} 
\end{itemize}
\end{proof}

The second result that supports the proof of Lemma \ref{lemma:comweak} formalizes the notion of ``upward closure'' for expressions.

\begin{lemma}[Register Record Upward Closure for Expressions]\label{lemma:expweak}
Let $E$ be a \whprog\ expression. If there exists $\regenv$ such that $\regenv, \activeReg \eenta{\regvar}{\var} E \etypeproduce \code{D}, \esecan{\lambda}{n}, r, \regenv_\alpha$ then $\forall \regenv' \sqsupseteq \regenv$ $\regenv', \activeReg \eenta{\regvar}{\var} E \etypeproduce \code{D}, \esecan{\lambda}{n}, r, \regenv_\beta$ such that $\regenv_\beta \sqsupseteq \regenv_\alpha$.
\end{lemma}
\begin{proof}
We prove the proposition by induction on the structure of $E$ and by cases on the last rule applied in the type derivation. 

\framebox[1.1\width]{Base case} 

Case $E=k$. We know that there exists $\regenv$ such that $\regenv, \activeReg \eenta{\regvar}{\var} k \etypeproduce \code{\assign{r}{k}}, \esecan{\seclev(r)}{1}, r, \subst{\regenv}{\breakconn{r}}$ and consider $\regenv' \sqsupseteq \regenv$. Then $\regenv', \activeReg \eenta{\regvar}{\var} k \etypeproduce \code{r:=k}, \esecan{\seclev(r)}{1}, r, \subst{\regenv'}{\breakconn{r}}$ and $\regenv_\beta = \subst{\regenv'}{\breakconn{r}} \sqsupseteq  \subst{\regenv}{\breakconn{r}} = \regenv_\alpha$ by Lemma \ref{lemma:regenvoper}. \\

Case $E=x$. Assume that the rule $\wraprule{\rulename{V}-cached}$ is used. Then there exists $\regenv$ such that $\conn{r}{x}{\mu} \in \regenv$ and  $\regenv, \activeReg \eenta{\regvar}{\var} x \etypeproduce \code{\findot},\esecan{\lambda}{0}, r, \regenv$. Consider $\regenv' \sqsupseteq \regenv$. Then $\conn{r}{x}{\mu} \in \regenv'$ and $\regenv', \activeReg \eenta{\regvar}{\var} x \etypeproduce \code{\findot},\esecan{\lambda}{0}, r, \regenv'$ and $\regenv_\beta = \regenv' \sqsupseteq  \regenv = \regenv_\alpha$ by hypothesis.\\
Assume that the rule $\wraprule{\rulename{V}-uncached}$ is used instead. Then there exists $\regenv$ such that $\regenv, \activeReg \eenta{\regvar}{\var} x \etypeproduce \code{\assign{r}{x}}, \esecan{\seclev(r)}{1}, r, \subst{\regenv}{\conn{r}{x}{\readMode}}$. Let $\regenv' \sqsupseteq \regenv$. Then $\regenv', \activeReg \eenta{\regvar}{\var} x \etypeproduce \code{\assign{r}{x}}, \esecan{\seclev(r)}{1}, r, \subst{\regenv'}{\conn{r}{x}{\readMode}}$ and $\regenv_\beta =  \subst{\regenv'}{\conn{r}{x}{\readMode}} \sqsupseteq   \subst{\regenv}{\conn{r}{x}{\readMode}} = \regenv_\alpha$ by Lemma \ref{lemma:regenvoper}. \\

\framebox[1.1\width]{Inductive step} 

Case $E=\wop{E_1}{E_2}$. We know that there exists $\regenv$ such that $\regenv, \activeReg \eenta{\regvar}{\var} \wop{E_1}{E_2} \etypeproduce \code{D_1;D_2;\assign{r}{\wop{r}{r'}}}, \esecan{\lambda}{n_1+n_2+1}, r, \subst{\regenv_2}{\breakconn{r}}$, under the assumption that $\regenv, \activeReg \eenta{\regvar}{\var} E_1 \etypeproduce \code{D_1}, \esecan{\lambda}{n_1}, r, \regenv_1$ and $\regenv_1, \activeReg\cup\{r\} \eenta{\regvar}{\var} E_2 \etypeproduce \code{D_2}, \esecan{\lambda}{n_2}, r', \regenv_2$. Consider $\regenv' \sqsupseteq \regenv$. By applying the inductive hypothesis  on $E_1$ we obtain $\regenv', \activeReg \eenta{\regvar}{\var} E_1 \etypeproduce \code{D_1}, \esecan{\lambda}{n_1}, r, \regenv_1'$, such that $\regenv_1' \sqsupseteq \regenv_1$. By applying the inductive hypothesis on $E_2$ we obtain $\regenv_1', \activeReg\cup\{r\} \eenta{\regvar}{\var} E_2 \etypeproduce \code{D_2}, \esecan{\lambda}{n_2}, r', \regenv_2'$ such that $\regenv_2' \sqsupseteq \regenv_2$. Hence $\regenv', \activeReg \eenta{\regvar}{\var} \wop{E_1}{E_2} \etypeproduce \code{D_1;D_2;\assign{r}{\wop{r}{r'}}}, \esecan{\lambda}{n_1+n_2+1}, r, \subst{\regenv_2'}{\breakconn{r}}$ and $\regenv_\beta = \subst{\regenv_2'}{\breakconn{r}} \sqsupseteq  \subst{\regenv_2}{\breakconn{r}} = \regenv_\alpha$ by Lemma \ref{lemma:regenvoper}.

\end{proof}

We continue with the formalization (and the proof) of ``upward closure'' for commands.

\begin{lemma}[Register Record Upward Closure for Commands]\label{lemma:comweak}
Let $C$ be a \whprog\ program. If there exists $\regenv$ such that $\regenv \enta{\regvar}{\var}  C \typeproduce \code{D}, \secan{t}{w}, \regenv_\alpha$ then $\forall \regenv' \sqsupseteq \regenv$ $\regenv' \enta{\regvar}{\var}  C \typeproduce \code{D}, \secan{t}{w}, \regenv_\beta$ such that $\regenv_\beta \sqsupseteq \regenv_\alpha$.
\end{lemma}
\begin{proof}
We prove the proposition by induction on the structure of $C$ and by cases on the last rule applied in the type derivation.

\framebox[1.1\width]{Base case} 

Case $C=\wskip$. We know that $\regenv \enta{\regvar}{\var}  \wskip \typeproduce \code{\wskip}, \secan{\tconst~(1,1)}{\whigh}, \regenv$. Consider $\regenv' \sqsupseteq \regenv$. Then $\regenv' \enta{\regvar}{\var}  \wskip \typeproduce \code{\wskip}, \secan{\tconst~(1,1)}{\whigh}, \regenv'$ and $\regenv_\beta = \regenv' \sqsupseteq \regenv = \regenv_\alpha$ by hypothesis. \\

Case $C= \assign{x}{E}$. We know that $\regenv \enta{\regvar}{\var} \assign{x}{E} \typeproduce \code{D;\assign{x}{r}}, \condtype{\seclev(x)=\valHigh}{\secan{\tconst~(1,n+1)}{\whigh}}{\secan{\tdontc}{\wdontk}}, 
\subst{\regenv_e}{\conn{r}{x}{\writeMode}}$, under the assumption that $\regenv, \emptyRegAll \eenta{\regvar}{\var} E \etypeproduce \code{D}, \esecan{\seclev(x)}{n}, r, \regenv_e$. Consider $\regenv' \sqsupseteq \regenv$. By applying Lemma \ref{lemma:expweak} we obtain that $\regenv', \emptyRegAll \eenta{\regvar}{\var} E \etypeproduce \code{D}, \esecan{\seclev(x)}{n}, r, \regenv'_e$, such that $\regenv'_e \sqsupseteq \regenv_e$, and $\regenv' \enta{\regvar}{\var} \assign{x}{E} \typeproduce \code{D;\assign{x}{r}}, \condtype{\seclev(x)=\valHigh}{\secan{\tconst~(1,n+1)}{\whigh}}{\secan{\tdontc}{\wdontk}}, 
\subst{\regenv_e'}{\conn{r}{x}{\writeMode}}$ and  $\regenv_\beta= \subst{\regenv_e'}{\conn{r}{x}{\writeMode}} \sqsupseteq \subst{\regenv_e}{\conn{r}{x}{\writeMode}} = \regenv_\alpha$ by Lemma \ref{lemma:regenvoper}. \\

Case $C= \out{ch}{E}$. The derivation $\regenv \enta{\regvar}{\var} \out{ch}{E} \typeproduce \code{D;\out{ch}{r}}, 
\condtype{\seclev(ch)=\valHigh}{\secan{\tconst~(1,n+1)}{\whigh}}{\secan{\tdontc}{\wdontk}}, 
\regenv_e$ holds, under the assumption that that $\regenv, \emptyRegAll \eenta{\regvar}{\var} E \etypeproduce \code{D}, \esecan{\seclev(ch)}{n}, r, \regenv_e$. Consider $\regenv' \sqsupseteq \regenv$. By applying Lemma \ref{lemma:expweak} we obtain that $\regenv', \emptyRegAll \eenta{\regvar}{\var} E \etypeproduce \code{D}, \esecan{\seclev(ch)}{n}, r, \regenv'_e$ such that $\regenv'_e \sqsupseteq \regenv_e$. Hence $\regenv' \enta{\regvar}{\var} \out{ch}{E} \typeproduce \code{D;\out{ch}{r}}, 
\condtype{\seclev(ch)=\valHigh}{\secan{\tconst~(1,n+1)}{\whigh}}{\secan{\tdontc}{\wdontk}}, 
\regenv_e'$ and $\regenv_\beta = \regenv_e' \sqsupseteq \regenv_e  = \regenv_\alpha$. \\

\framebox[1.1\width]{Inductive step}

Case $C=\wif{E}{C_t}{C_e}$. An $\wifName$ statement can be typed according to two rules, $\wraprule{\rulename{\wifName}-any}$ and $\wraprule{\rulename{\wifName}-\valHigh}$. \\
Assume the rule $\wraprule{\rulename{\wifName}-any}$ is used.  By hypothesis, we have that  the derivation $\typeconcel{\regenv}{\wif{E}{C_t}{C_e}}{\code{D_g;\wif{r}{D_t'}{D_e'}}}{\secan{\termMap(\lambda)\tlub t_1 \tlub t_2 }{\writeMap(\lambda) \wlub w_1 \wlub w_2}}{\regenv_2 \envinters \regenv_3}{\enta{\regvar}{\var}}$ holds, for $D_t'=D_t;\wskip$ and $D_e'=D_e;\wskip$, providing that $\regenv, \emptyRegAll \eenta{\regvar}{\var} E \etypeproduce \code{D_g}, \esecan{\lambda}{n_g} , r, \regenv_1$ and $\regenv_1 \enta{\regvar}{\var}  C_t \typeproduce \code{D_t}, \secan{t_1}{w_1}, \regenv_2$ and $\regenv_1 \enta{\regvar}{\var}  C_e \typeproduce \code{D_e}, \secan{t_2}{w_2}, \regenv_3$. Consider $\regenv' \sqsupseteq \regenv$. By applying Lemma \ref{lemma:expweak} on $E$ we obtain $\regenv', \emptyRegAll \eenta{\regvar}{\var} E \etypeproduce \code{D_g}, \esecan{\lambda}{n_g} , r, \regenv_1'$ such that $\regenv_1' \sqsupseteq \regenv_1$. By applying the inductive hypothesis on $C_t$ we obtain $\regenv_1' \enta{\regvar}{\var}  C_t \typeproduce \code{D_t}, \secan{t_1}{w_1}, \regenv_2'$ such that $\regenv_2' \sqsupseteq \regenv_2$. By applying the inductive hypothesis on $C_e$ we obtain $\regenv_1' \enta{\regvar}{\var}  C_e \typeproduce \code{D_e}, \secan{t_2}{w_2}, \regenv_3'$ such that $\regenv_3' \sqsupseteq \regenv_3$. Hence we conclude that  $\typeconcel{\regenv'}{\wif{E}{C_t}{C_e}} {\code{D_g;\wif{r}{D_t'}{D_e'}}}{\secan{\termMap(\lambda)\tlub t_1 \tlub t_2 }{\writeMap(\lambda) \wlub w_1 \wlub w_2}}{\regenv_2' \envinters \regenv_3'}{\enta{\regvar}{\var}}$ and $\regenv_\beta= \regenv_2' \envinters \regenv_3' \sqsupseteq \regenv_2 \envinters \regenv_3 = \regenv_\alpha$. \\
Consider the case in which the rule $\wraprule{\rulename{\wifName}-\valHigh}$ is used instead. 
We know that $\typeconcel{\regenv}{\wif{E}{C_t}{C_e}}{\code{D_g;\wif{r}{D_t'}{D_e'}}}{\secan{\tconst~(m+1,n_g+max(n_t,n_e)+2)}{\whigh}}{\regenv_2 \envinters \regenv_3}{\enta{\regvar}{\var}}$, for $D_t'=D_t;\wskip^{n_e-n_t};\wskip$ and $D_e'=D_e;\wskip^{n_t-n_e};\wskip$, providing that $\regenv, \emptyRegAll \eenta{\regvar}{\var} E \etypeproduce \code{D_g}, \esecan{\valHigh}{n_g}, r, \regenv_1$ together with $\regenv_1 \enta{\regvar}{\var}  C_t \typeproduce \code{D_t}, \secan{\tconst~(m,n_t)}{\whigh} , \regenv_2$ and $\regenv_1 \enta{\regvar}{\var}  C_e \typeproduce \code{D_e}, \secan{\tconst~(m,n_e)}{\whigh}, \regenv_3$. By applying Lemma \ref{lemma:expweak} on $E$ we obtain $\regenv', \emptyRegAll \eenta{\regvar}{\var} E \etypeproduce \code{D_g}, \esecan{\valHigh}{n_g} , r, \regenv_1'$ such that $\regenv_1' \sqsupseteq \regenv_1$. By applying the inductive hypothesis on $C_t$ we obtain $\regenv_1' \enta{\regvar}{\var}  C_t \typeproduce \code{D_t}, \secan{\tconst~(m,n_t)}{\whigh} , \regenv_2'$ such that $\regenv_2' \sqsupseteq \regenv_2$. By applying the inductive hypothesis on $C_e$ we obtain $\regenv_1' \enta{\regvar}{\var}  C_e \typeproduce \code{D_e}, \secan{\tconst~(m,n_e)}{\whigh}, \regenv_3'$ such that $\regenv_3' \sqsupseteq \regenv_3$. Hence we can conclude that the derivation $\typeconcel{\regenv'}{\wif{E}{C_t}{C_e}}{\code{D_g;\wif{r}{D_t'}{D_e'}}}{\secan{\tconst~(m+1,n_g+max(n_t,n_e)+2)}{\whigh}}{\regenv_2' \envinters \regenv_3'}{\enta{\regvar}{\var}}$ holds and $\regenv_\beta= \regenv_2' \envinters \regenv_3' \sqsupseteq \regenv_2 \envinters \regenv_3 = \regenv_\alpha$. \\

Case $C=\while{x}{C'}$. 
By considering the rule definition we know that $\typeconcel{\regenv}{\while{x}{C}}{\code{D_0; \while{r}{\code{D;D_0;\wskip}}}}{\secan{\termMap(\lambda)\tlub t}{\writeMap(\lambda) \wlub w}}{\regenv_B}{\enta{\regvar}{\var}}$ 
for $D_0= \assign{r}{x};\assign{x}{r}$, $\regenv_{*}= \subst{\regenv}{\conn{r}{x}{\writeMode}}$, $\regenv_B \sqsubseteq \regenv_{*}$ and $\regenv_B \sqsubseteq \subst{\regenv_E}{\conn{r}{x}{\writeMode}}$, providing that $\regenv_B \enta{\regvar}{\var}  C \typeproduce \code{D}, \secan{t}{w} , \regenv_E$. Consider $\regenv' \sqsupseteq \regenv$. By Lemma \ref{lemma:regenvoper} we know that $\subst{\regenv'}{\conn{r}{x}{\writeMode}} \sqsupseteq \subst{\regenv}{\conn{r}{x}{\writeMode}}$, hence the same $\regenv_B$ used for $\regenv$ can be used for $\regenv'$ to conclude that  that $\typeconcel{\regenv'}{\while{x}{C}}{\code{D_0; \while{r}{\code{D;D_0;\wskip}}}}{\secan{\termMap(\lambda)\tlub t}{\writeMap(\lambda) \wlub w}}{\regenv_B}{\enta{\regvar}{\var}}$ and $\regenv_\beta= \regenv_B \sqsupseteq \regenv_B = \regenv_\alpha$. \\

Case $C=C_1;C_2$. We know that $\regenv \enta{\regvar}{\var}  C_1;C_2 \typeproduce \code{D_1;D_2}, \secan{t_1 \tclub t_2}{w_1 \wlub w_2},\regenv_2$, providing that $\regenv \enta{\regvar}{\var}  C_1 \typeproduce \code{D_1}, \secan{t_1}{w_1}, \regenv_1$ as well as $\regenv_1 \enta{\regvar}{\var}  C_2 \typeproduce \code{D_2}, \secan{t_2}{w_2}, \regenv_2$ hold. Consider $\regenv' \sqsupseteq \regenv$. By applying the inductive hypothesis on $C_1$ we obtain $\regenv' \enta{\regvar}{\var}  C_1 \typeproduce \code{D_1}, \secan{t_1}{w_1}, \regenv_1'$ such that $\regenv_1' \sqsupseteq \regenv_1$. By applying the inductive hypothesis on $C_2$ we obtain $\regenv_1' \enta{\regvar}{\var}  C_2 \typeproduce \code{D_2}, \secan{t_2}{w_2}, \regenv_2'$ such that $\regenv_2' \sqsupseteq \regenv_2$. Hence $\regenv' \enta{\regvar}{\var}  C_1;C_2 \typeproduce \code{D_1;D_2}, \secan{t_1 \tclub t_2}{w_1 \wlub w_2},\regenv_2'$ and $\regenv_\beta= \regenv_2' \sqsupseteq \regenv_2 = \regenv_\alpha$.

\end{proof}

The second result that supports the proof of Proposition \ref{prop:strongsectype} is subject reduction. In order to characterize subject reduction we have to reason about the connection between the operational semantics and the type system. For this purpose it is convenient to generalize the syntax to include $\findot$ (the terminated program) as a possible sub term of any sequential composition term. We extend the type system to include the following typing rule for $\findot$:
$$\inference[\rulename{T}]{}{\regenv \enta{\regvar}{\var} \findot \typeproduce \code{\findot}, \secan{\tconst~(0,0)}{\whigh}, \regenv}$$ 
We choose this particular typing value since $\findot$ represents the final configuration, which does not perform any reduction step.

\begin{lemma}[Extended typing is compatible with structural equivalence]
For any \whprog\ program $C$, if $\regenv \enta{\regvar}{\var}  C \typeproduce \code{D}, \secan{t}{w}, \regenv'$ and $C \equiv C'$, then $\regenv \enta{\regvar}{\var}  C' \typeproduce \code{D'}, \secan{t}{w}, \regenv'$ such that $D \equiv D'$. 
\end{lemma}
\begin{proof} Induction on derivation, which follows easily from the fact that the typing of $\findot;C$, $C;\findot$  and $C$ are all equivalent. 
\end{proof}

This compatibility property allows us to implicitly view a typing derivation as applying to structural equivalence classes of terms (much in the same way that one reasons informally about alpha equivalence in languages with variable binding). When reasoning about induction on the size of a derivation, we will view the size of a derivation to be the size of the $\wraprule{\rulename{T}}$-rule free derivation - i.e. the smallest derivation of the given typing. 

\begin{lemma}[Reduction Context Typing]\label{lemma:ctxandtypes}
For any \whprog\ command $C$ and reduction context $R[]$, there exists $\regenv$ such that $\regenv \enta{\regvar}{\var}  R[C] \typeproduce \code{\dots}, \secan{t}{w}, \regenv'$ if and only if there exist $\regenv$ and pairs $t_1,t_2$ and $w_1,w_2$ such that $\regenv \enta{\regvar}{\var}  C \typeproduce \code{\dots}, \secan{t_1}{w_1}, \regenv_1$ and $\regenv_1 \enta{\regvar}{\var}  R[\findot] \typeproduce \code{\dots}, \secan{t_2}{w_2}, \regenv'$ and $t_1 = \ttop \Rightarrow w_2= \whigh$ for $w = w_1 \wlub w_2$ and $t = t_1 \tclub t_2$. 
\end{lemma}
\begin{proof}
We prove the lemma by induction on the structure of $R$.

\framebox[1.1\width]{Base case} 

Assume $R=[]$. Then $R[C]=C$. \\
($\Rightarrow$)  Assume $\regenv \enta{\regvar}{\var}  R[C] \typeproduce \code{\dots}, \secan{t}{w}, \regenv'$. Since $R[C]=C$ then $\regenv \enta{\regvar}{\var}  C \typeproduce \code{\dots}, \secan{t}{w}, \regenv'$ hence  $t_1=t$ and $w_1=w$. Also, $R[\findot]= \findot$, and we know that $\forall \regenv$ $\regenv \enta{\regvar}{\var} \findot \typeproduce \code{\findot}, \secan{\tconst~(0,0)}{\whigh}, \regenv$, hence $t_2= \tconst~(0,0)$ and $w_2=\whigh$. Notice that $w= w_1= w_1\wlub w_2$ and $t= t_1= t_1 \tclub t_2$.\\
($\Leftarrow$)  Assuming $\regenv \enta{\regvar}{\var}  C \typeproduce \code{\dots}, \secan{t}{w}, \regenv'$ the following type derivation is correct
$$\inference{\regenv \enta{\regvar}{\var}  C \typeproduce \code{\dots}, \secan{t}{w}, \regenv' & \regenv' \enta{\regvar}{\var} \findot \typeproduce \code{\findot}, \secan{\tconst~(0,0)}{\whigh}, \regenv'}{\regenv \enta{\regvar}{\var}  C;\findot \typeproduce \code{\dots}, \secan{t}{w}, \regenv' }$$
and $C;\findot \equiv C=R[C]$.

\framebox[1.1\width]{Inductive Step} 

Assume $R[]=R_1[];C_1$. Then there exists a command $C'$ such that  $R[C] = R_1[C'];C_1$. \\
($\Rightarrow$)  Assume $\regenv \enta{\regvar}{\var}  R[C] \typeproduce \code{\dots}, \secan{t}{w}, \regenv'$  Then the following type derivation must exist:

$$
\inference{
\regenv \enta{\regvar}{\var}  R_1[C'] \typeproduce \code{\dots}, \secan{t_\alpha}{w_\alpha}, \regenv_\alpha \\
\regenv_\alpha \enta{\regvar}{\var}  C_1 \typeproduce \code{\dots}, \secan{t_\beta}{w_\beta}, \regenv_\beta \\
w= w_\alpha \wlub w_\beta & t= t_\alpha \tclub t_\beta}
{\regenv \enta{\regvar}{\var}  R_1[C'];C_1 \typeproduce \code{\dots}, \secan{t}{w}, \regenv'}
$$ 
such that $t_\alpha = \ttop \Rightarrow w_\beta=\whigh$.
We now focus on $R_1[C']$. By applying the inductive hypothesis we know that there exist $t_a,t_b$ and $w_a,w_b$ such that $\regenv \enta{\regvar}{\var}  C' \typeproduce \code{\dots}, \secan{t_a}{w_a}, \regenv_1$ and $\regenv_1 \enta{\regvar}{\var}  R_1[\findot] \typeproduce \code{\dots}, \secan{t_b}{w_b}, \regenv_\alpha$ and $w_\alpha = w_a \wlub w_b$ and $t_\alpha= t_a \tclub t_b$. Hence, the following type derivation 
$$
\inference{
\regenv_1 \enta{\regvar}{\var}  R_1[\findot] \typeproduce \code{\dots}, \secan{t_b}{w_b}, \regenv_\alpha \\ 
\regenv_\alpha \enta{\regvar}{\var}  C_1 \typeproduce \code{\dots}, \secan{t_\beta}{w_\beta}, \regenv_\beta }
{\regenv_1 \enta{\regvar}{\var}  R_1[\findot];C_1 \typeproduce \code{\dots}, \secan{t_b \tclub t_\beta}{w_b \wlub w_\beta}, \regenv'}
$$ 
is correct. In particular, if $t_b = \ttop$, then $t_\alpha = \ttop$, hence $w_\beta = \whigh$ for hypothesis. \\
($\Leftarrow$) Assume that $\regenv \enta{\regvar}{\var}  C' \typeproduce \code{\dots}, \secan{t_a}{w_a}, \regenv_1$ together with $\regenv_1 \enta{\regvar}{\var}  R_1[\findot];C_1 \typeproduce \code{\dots}, \secan{t_c}{w_c}, \regenv'$ such that $t_a = \ttop \Rightarrow w_c= \whigh$. Then the following type derivation $$
\inference{
\regenv \enta{\regvar}{\var}  C' \typeproduce \code{\dots}, \secan{t_a}{w_a}, \regenv_1 \\ 
\regenv_1 \enta{\regvar}{\var}  R_1[\findot];C_1 \typeproduce \code{\dots}, \secan{t_c}{w_c}, \regenv' }
{\regenv \enta{\regvar}{\var}  R_1[C'];C_1 \typeproduce \code{\dots}, \secan{t_a \tclub t_c}{w_a \wlub w_c}, \regenv'}
$$
is correct. 
\end{proof}

We can now present all the details related to subject reduction.

\begin{proposition}[Subject Reduction]\label{prop:subjred}
Let $C$ be a \whprog\ program. If there exists $\regenv$ such that $\regenv \enta{\regvar}{\var}  C \typeproduce \code{D}, \secan{t_\alpha}{w_\alpha}, \regenv_\alpha$ and there exists $\mem$ such that $\whstate{C}{\mem} \whsem{l} \whstate{C'}{\mem'}$, then there exists $\regenv'$ such that $\regenv' \enta{\regvar}{\var}  C' \typeproduce \code{D'}, \secan{t_\beta}{w_\beta}, \regenv_\beta$ and:
\begin{itemize}
\item $w_\beta \sqsubseteq w_\alpha$ and $t_\beta \sqsubseteq t_\alpha$;
\item $\regenv_\beta \sqsupseteq \regenv_\alpha$.
\end{itemize}
\end{proposition}
\begin{proof}
We prove the proposition by induction on the structure of $C$ and by cases on the last rule applied in the type derivation. In the proof we omit the explicit representation of the code production since it is not relevant in this context. 

\framebox[1.1\width]{Base case} 

In this case $C \in \{ \wskip, \assign{x}{E}, \out{ch}{E}\}$. For any such $C$ we have that $\forall \mem$ $\whstate{C}{\mem} \whsem{l} \whstate{\findot}{\mem'}$ for suitable $l$ and $\mem'$. Since $\forall \regenv$ $\regenv \enta{\regvar}{\var} \findot \typeproduce \code{\findot}, \secan{\tconst~(0,0)}{\whigh}, \regenv$ the statement trivially holds. 

\framebox[1.1\width]{Inductive step} 

Case $C=\wif{E}{C_1}{C_2}$. An $\wifName$ statement can be typed according to two rules, $\wraprule{\rulename{\wifName}-any}$ and $\wraprule{\rulename{\wifName}-\valHigh}$. \\
Assume the rule $\wraprule{\rulename{\wifName}-any}$ is used. 
Then 
$\typeconcel{\regenv}{\wif{E}{C_t}{C_e}}{\code{\dots}} {\secan{\termMap(\lambda)\tlub t_1 \tlub t_2}{\writeMap(\lambda) \wlub w_1 \wlub w_2 }}{\regenv_F}{\enta{\regvar}{\var}}$    
holds under the assumption that 
$\regenv, \emptyRegAll \eenta{\regvar}{\var} E \etypeproduce \code{\dots}, \esecan{\lambda}{n_g} , r, \regenv_1$ 
and  
$\regenv_1 \enta{\regvar}{\var}  C_t \typeproduce \code{\dots}, \secan{t_1}{w_1}, \regenv_2$ 
and 
$\regenv_1 \enta{\regvar}{\var}  C_e \typeproduce \code{\dots}, \secan{t_2}{w_2}, \regenv_3$
, such that $\regenv_F=\regenv_2 \envinters \regenv_3$. 
Assume that $\mem$ is such that  $\whstate{\wif{E}{C_t}{C_e}}{\mem} \whsem{\tau} \whstate{C_t}{\mem}$. The statement is true for $\regenv'=\regenv_1$ since $ w_\beta= w_1 \sqsubseteq \writeMap(\lambda) \wlub w_1 \wlub w_2 = w_\alpha$ and $t_\beta= t_1 \sqsubseteq \termMap(\lambda)\tlub t_1 \tlub t_2 = t_\alpha$ and $\regenv_\beta= \regenv_2 \sqsupseteq \regenv_2 \envinters \regenv_3 = \regenv_F = \regenv_\alpha$. The case for $\mem$ such that $\whstate{\wif{E}{C_t}{C_e}}{\mem} \whsem{\tau} \whstate{C_e}{\mem}$ is analogous, hence it is omitted. \\
Consider the case in which the rule $\wraprule{\rulename{\wifName}-\valHigh}$ is used instead.
Then we know 
$\typeconcel{\regenv}{\wif{E}{C_t}{C_e}}{\code{\dots}}{\secan{\tconst~(m+1,n_g+max(n_t,n_e)+2)}{\whigh}}{\regenv_F}{\enta{\regvar}{\var}}$ under the assumption that $\regenv, \emptyRegAll \eenta{\regvar}{\var} E \etypeproduce \code{\dots}, \esecan{\valHigh}{n_g}, r, \regenv_1$ and $\regenv_1 \enta{\regvar}{\var}  C_t \typeproduce \code{\dots}, \secan{\tconst~(m,n_t)}{\whigh}, \regenv_2$ and $\regenv_1 \enta{\regvar}{\var}  C_e \typeproduce \code{\dots}, \secan{\tconst~(m,n_e)}{\whigh}, \regenv_3$ such that $\regenv_F=\regenv_2 \envinters \regenv_3$.
Assume that $\mem$ is such that  $\whstate{\wif{E}{C_t}{C_e}}{\mem} \whsem{\tau} \whstate{C_t}{\mem}$. The statement is true for $\regenv'=\regenv_1$ since $ w_\beta= \whigh \sqsubseteq \whigh = w_\alpha$ and $t_\beta= \tconst~(m,n_t) \sqsubseteq \tconst~(m+1,n_g+max(n_t,n_e)+2) = t_\alpha$ and $\regenv_\beta= \regenv_2 \sqsupseteq \regenv_2 \envinters \regenv_3 = \regenv_F = \regenv_\alpha$. The case for $\mem$ such that $\whstate{\wif{E}{C_t}{C_e}}{\mem} \whsem{\tau} \whstate{C_e}{\mem}$ is analogous, hence it is omitted. \\

Case $C=\while{x}{C'}$. 
Then we know that $\typeconcel{\regenv}{\while{x}{C'}}{\code{\dots}}{\secan{\termMap(\lambda)\tlub t}{\writeMap(\lambda) \wlub w}}{\regenv_B}{\enta{\regvar}{\var}}$. This is true under the assumption that, for $\regenv_{*}= \subst{\regenv}{\conn{r}{x}{\writeMode}}$, there exists $\regenv_B$ such that 
$\regenv_B \sqsubseteq \regenv_{*}$ and $\regenv_B \sqsubseteq \subst{\regenv_E}{\conn{r}{x}{\writeMode}}$
for  which $\regenv_B \enta{\regvar}{\var}  C' \typeproduce \code{\dots}, \secan{t}{w} , \regenv_E$. It is required that $\lambda = \seclev(x)=\seclev(r)$ and $\writeMap(\lambda) \sqsupseteq w$ and $t= \ttop \Rightarrow \writeMap(\lambda)= \whigh$. Assume that $\mem$ is such that  $\whstate{\while{x}{C'}}{\mem} \whsem{\tau} \whstate{C';\while{x}{C'}}{\mem}$ and consider the following type derivation

$$\begin{array}{c}
\Scale[0.94]{
\inference{
\regenv_B \enta{\regvar}{\var}  
C' \typeproduce \code{\dots}, \secan{t}{w} , \regenv_E &
 \inference{
 \regenv_B \sqsubseteq \subst{\regenv_E}{\conn{r}{x}{\writeMode}}\\
 \regenv_B \enta{\regvar}{\var}  
C' \typeproduce \code{\dots}, \secan{t}{w} , \regenv_E \\
 \writeMap(\lambda) \sqsupseteq w \\ t= \ttop \Rightarrow \writeMap(\lambda)= \whigh \\
}
 {
\Delta
}}
{
\regenv_B  \enta{\regvar}{\var}  
\left \{
\begin{array}{l}
C'; \\
\while{x}{C'}
\end{array}
\right \} \typeproduce \code{\dots}, 
\left \langle
\begin{array}{cc}
\writeMap(\lambda) \wlub w \\
t \tclub (\termMap(\lambda)\tlub t)
\end{array}
\right \rangle , \regenv_B}}

\end{array}$$
for $\Delta=\regenv_E \enta{\regvar}{\var} \while{x}{C'} \typeproduce \code{\dots}, \secan{\termMap(\lambda)\tlub t}{\writeMap(\lambda) \wlub w}, \regenv_B$. 

Notice that the proposition is true for $\regenv'=\regenv_B$ since $ w_\beta= \writeMap(\lambda) \wlub w \sqsubseteq \writeMap(\lambda) \wlub w = w_\alpha$ and $t_\beta= t \tclub (\termMap(\lambda)\tlub t) =  \termMap(\lambda)\tlub t \sqsubseteq \termMap(\lambda)\tlub t = t_\alpha$ and $\regenv_\beta= \regenv_B \sqsupseteq \regenv_B = \regenv_\alpha$. The case for $\mem$ such that $\whstate{\while{x}{C'}}{\mem} \whsem{\tau} \whstate{\findot}{\mem}$ is trivial, hence it is omitted. \\

Case $C=C_1;C_2$. 
Then we know that $\regenv \enta{\regvar}{\var}  C_1;C_2 \typeproduce \code{\dots},\secan{t_1 \tclub t_2}{w_1 \wlub w_2}, \regenv_2$ under the assumption that  $\regenv \enta{\regvar}{\var}  C_1 \typeproduce \code{\dots}, \secan{t_1}{w_1}, \regenv_1$ and $\regenv_1 \enta{\regvar}{\var}  C_2 \typeproduce \code{\dots}, \secan{t_2}{w_2}, \regenv_2$ and $t_1= \ttop \Rightarrow w_2=\whigh$. We now consider the reduction of $C$ by distinguishing two cases. \newline
Assume the rule $\wraprule{\rulename{c-1}}$ is applied. Hence there exists $ C_3$ such that $C = R[C_3]$ and $\whstate{R[C_3]}{\mem} \whsem{l} \whstate{R[C_3']}{\mem'}$, given that $\whstate{C_3}{\mem} \whsem{l} \whstate{C_3'}{\mem'}$ and $C_3' \not = \findot$. By hypothesis of type-correctness and Lemma \ref{lemma:ctxandtypes}, this type derivation for  $C$ must exist
$$\begin{array}{c} 
\inference{
\regenv \enta{\regvar}{\var}  C_3 \typeproduce \code{\dots}, \secan{t_3}{w_3},\regenv_3 &
\regenv_3 \enta{\regvar}{\var}  R[\findot] \typeproduce \code{\dots}, \secan{t_c}{w_c},\regenv_2\\
t_3= \ttop \Rightarrow w_c=\whigh 
}{\regenv \enta{\regvar}{\var}  R[C_3] \typeproduce \code{\dots}, \secan{t_3 \tclub t_c}{w_3 \wlub w_c},\regenv_2}
\end{array}$$
Notice that it must also be that  $w_3 \wlub w_c=w_1 \wlub w_2$ and  $t_3 \tclub t_c=t_1 \tclub t_2$.
By applying the inductive hypothesis on $C_3$, there must exist a $\regenv'$ such that $\regenv' \enta{\regvar}{\var}  C_3' \typeproduce \code{\dots},\regenv_3', \secan{t_3'}{w_3'}$ such that $w_3' \sqsubseteq w_3$ and $t_3' \sqsubseteq t_3$ and $\regenv_3' \sqsupseteq \regenv_3$. 
Consider this type derivation for $R[C_3']$
$$\begin{array}{c} 
\inference{
(i) &
\inference{\regenv_3 \enta{\regvar}{\var}  R[\findot] \typeproduce \code{\dots}, \secan{t_c}{w_c},\regenv_2}{\regenv_3' \enta{\regvar}{\var}  R[\findot] \typeproduce \code{\dots}, \secan{t_c}{w_c},\regenv_2' (ii)}\\
t_3'= \ttop \Rightarrow w_c=\whigh 
}
{\regenv' \enta{\regvar}{\var}  R[C_3'] \typeproduce \code{\dots}, \secan{t_3' \tclub t_c}{w_3' \wlub w_c},\regenv_2'}
\end{array}$$
where (i) stays for $\regenv' \enta{\regvar}{\var}  C_3' \typeproduce \code{\dots}, \secan{t_3'}{w_3'} ,\regenv_3'$ and holds by hypothesis, and (ii) is an application of Lemma \ref{lemma:comweak}.  We have that $ w_\beta= w_3' \wlub w_c \sqsubseteq w_3 \wlub w_c=w_1 \wlub w_2 = w_\alpha$ and $t_\beta= t_3' \tclub t_c \sqsubseteq t_3 \tclub t_c=t_1 \tclub t_2 = t_\alpha$ and $\regenv_\beta= \regenv_2' \sqsupseteq \regenv_2 = \regenv_\alpha$ because of Lemma \ref{lemma:comweak}.\\
Assume the rule $\wraprule{\rulename{c-2}}$ is applied instead. Hence $\exists C_3,C_4$ such that $C = R[C_3;C_4]$ and $\whstate{R[C_3;C_4]}{\mem} \whsem{l} \whstate{R[C_4]}{\mem'}$, given that $\whstate{C_3}{\mem} \whsem{l} \whstate{\findot}{\mem'}$. By hypothesis of type-correctness, this type derivation for  $C$ must exist
$$\begin{array}{c} 
\inference{
\regenv \enta{\regvar}{\var}  C_3 \typeproduce \code{\dots}, \secan{t_3}{w_3} ,\regenv_3&
\regenv_3 \enta{\regvar}{\var}  R[C_4] \typeproduce \code{\dots}, \secan{t_c}{w_c},\regenv_2\\
t_3= \ttop \Rightarrow w_c=\whigh 
}{\regenv \enta{\regvar}{\var}  R[C_3;C_4] \typeproduce \code{\dots}, \secan{t_3 \tclub t_c}{w_3 \wlub w_c},\regenv_2}
\end{array}$$
Notice that it must also be that  $w_3 \wlub w_c=w_1 \wlub w_2$ and  $t_3 \tclub t_c=t_1 \tclub t_2$.
By having $\regenv'=\regenv_3$ we immediately have that $\regenv_3 \enta{\regvar}{\var}  R[C_4] \typeproduce \code{\dots}, \secan{t_c}{w_c},\regenv_2$ and $ w_\beta= w_c  \sqsubseteq w_3 \wlub w_c=w_1 \wlub w_2 = w_\alpha$ and $t_\beta= t_c \sqsubseteq t_3 \tclub t_c=t_1 \tclub t_2 = t_\alpha$ and $\regenv_\beta= \regenv_2 \sqsupseteq \regenv_2 = \regenv_\alpha$.

\end{proof}


Before illustrating the details of the proof for Proposition \ref{prop:strongsectype}
we continue with some other auxiliary results that formalize features of commands that are associated to a write effect $\whigh$. The first result states that a $\whigh$ command does not produce low-distinguishable actions and does not alter the low part of the memory.

\begin{lemma}[Low transparency of $\secan{t}{\whigh}$ commands]\label{lemma:prophighcom}
Let $C$ be a \whprog\ command such that $C \not = \findot$ and there exists $\regenv$ such that $\regenv \enta{\regvar}{\var}  C \typeproduce \code{\dots}, \secan{t}{\whigh}, \regenv_\alpha$. Then for any $\mem$ we have $\whstate{C}{\mem} \whsem{l} \whstate{C'}{\mem'}$ such that $\low(l)= \tau$ and $\mem =_{\pubdata} \mem'$.
\end{lemma}
\begin{proof}

We prove the proposition by induction on the structure of $C$ and by cases on the last rule applied in the type derivation. In the proof we omit the explicit representation of the code production since it is not relevant in this context. Recall that $\forall \regenv$ $\regenv \enta{\regvar}{\var} \findot \typeproduce \code{\findot}, \secan{\tconst~(0,0)}{\whigh}, \regenv$.

\framebox[1.1\width]{Base case} 

Case $C=\wskip$. We know that $\regenv \enta{\regvar}{\var}  \wskip \typeproduce \code{\dots}, \secan{\tconst~(1,1)}{\whigh}, \regenv$ and that, for any memory $\mem$, $\whstate{\wskip}{\mem} \whsem{\tau} \whstate{\findot}{\mem}$. 

Case $C= \assign{x}{E}$. Assume that $\regenv \enta{\regvar}{\var}  \assign{x}{E} \typeproduce \code{\dots}, \secan{\tconst~(1,n+1)}{\whigh}, \regenv'$. Then $\seclev(x)= \valHigh$ and we know that, for any memory $\mem$, $\whstate{\assign{x}{E}}{\mem} \whsem{\tau} \whstate{\findot}{\mem[x \backslash |[E|](\mem)]}$ such that $\mem[x \backslash |[E|](\mem)] =_{\pubdata} \mem$.

Case $C= \out{ch}{E}$. Assume that $\regenv \enta{\regvar}{\var}  \out{ch}{E} \typeproduce \code{\dots}, \secan{\tconst~(1,n+1)}{\whigh}, \regenv'$. Then $\seclev(ch)= \valHigh$ and we know that, for any memory $\mem$, it is true that $\whstate{\out{ch}{E}}{\mem} \whsem{ch!|[E|](\mem)} \whstate{\findot}{\mem}$ such that $\low(ch!|[E|](\mem))=  \tau$. 

\framebox[1.1\width]{Inductive step} 

Case $C= \wif{E}{C_t}{C_e}$. Regardless of the typing derivation, for any memory $\mem$ there are two cases that are possible: either $\whstate{\wif{E}{C_t}{C_e}}{\mem} \whsem{\tau} \whstate{C_t}{\mem}$ or  $\whstate{\wif{E}{C_t}{C_e}}{\mem} \whsem{\tau} \whstate{C_e}{\mem}$. Hence, regardless of $\mem$, the transition produces a silent action and does not alter the memory, therefore the property trivially holds. 

Case $C= \while{x}{C'}$. Regardless of the typing derivation, for any memory $\mem$ there are two cases that are possible: either we have that $\whstate{\while{x}{C'}}{\mem} \whsem{\tau} \whstate{\findot}{\mem}$ or  $\whstate{\while{x}{C'}}{\mem} \whsem{\tau} \whstate{C';\while{x}{C'}}{\mem}$. Hence, regardless of $\mem$, the transition produces a silent action and does not alter the memory, and the property trivially holds. 

Case $C= C_1;C_2$. We proceed by cases according to the semantic rule used to reduce the first element of the pair.\\
Assume $\mem$ is considered such that  rule $\wraprule{\rulename{c-1}}$ is applied. Hence there exists $C_3 \in \{\wif{E}{C_t}{C_e},\while{x}{C'}\}$ such that $C_1;C_2 = R[C_3]$ and $\whstate{R[C_3]}{\mem} \whsem{l} \whstate{R[C_3']}{\mem'}$, given that $C_3' \not = \findot$.\\ 
As it has been observed previously, regardless of their type derivation, none of the commands produce a visible action and alter the memory, hence the property trivially holds. 
Assume $\mem$ is considered such that  rule  $\wraprule{\rulename{c-2}}$ is applied instead. Hence there exist two commands $C_3$ and $C_4$ such that $C_1;C_2 = R[C_3;C_4]$ and $\whstate{R[C_3;C_4]}{\mem} \whsem{l} \whstate{R[C_4]}{\mem'}$, given that $\whstate{C_3}{\mem} \whsem{l} \whstate{\findot}{\mem'}$. We distinguish two cases. Assume $C_3$ in $\{\wskip, \assign{x}{E}, \out{ch}{E}\}$.
If $\regenv \enta{\regvar}{\var}  C_1;C_2 \typeproduce \code{\dots}, \secan{t}{\whigh}, \regenv_\alpha$ then, by Lemma \ref{lemma:ctxandtypes}, a type derivation 

$$
\inference
{\regenv \enta{\regvar}{\var}  C_3 \typeproduce \code{\dots}, \secan{\tconst~(1,n_a)}{\whigh}, \regenv_1 \\
\regenv_1 \enta{\regvar}{\var}  R[\bullet;C_4] \typeproduce \code{\dots}, \secan{t'}{\whigh}, \regenv_\alpha}
{\regenv \enta{\regvar}{\var}  R[C_3;C_4] \typeproduce \code{\dots}, \secan{t}{\whigh}, \regenv_\alpha}
$$
must exist such that $t= \tconst~(1,n_a) \tclub t'$. By applying the inductive hypothesis on $C_3$ we know that for all $\mem$ we have $\whstate{C_3}{\mem} \whsem{l} \whstate{\findot}{\mem'}$ such that $\low(l)= \tau$ and $\mem =_{\pubdata} \mem'$, hence the property holds. 
Assume $C_3=\while{x}{C'}$. As observed previously, the command neither produces a visible action nor alter the memory. Therefore the property trivially holds. 

\end{proof}

The second result investigates the features of $\whigh$ commands further. 

\begin{lemma}[Progress properties of $\secan{t}{\whigh}$ commands]\label{lemma:prophighstep}
Let $C$ be a \whprog\ command for which there exists $\regenv$ such that $\regenv \enta{\regvar}{\var}  C \typeproduce \code{\dots}, \secan{\tconst~(m,n)}{\whigh}, \regenv_\alpha$. Then for any memory $\mem$ such that $\whstate{C}{\mem} \whsem{l} \whstate{C'}{\mem'}$ there exists $\regenv'$ 
such that $\regenv' \enta{\regvar}{\var}  C' \typeproduce \code{\dots}, \secan{\tconst~(m-1,n')}{\whigh}, \regenv_\beta$ and $\regenv_\beta \sqsupseteq \regenv_\alpha$.
\end{lemma}
\begin{proof}
We prove the proposition by induction on the structure of $C$ and by cases on the last rule applied in the type derivation. In the proof we omit the explicit representation of the code production since it is not relevant in this context. Recall that $\forall \regenv$ $\regenv \enta{\regvar}{\var} \findot \typeproduce \code{\findot}, \secan{\tconst~(0,0)}{\whigh}, \regenv$, and observe that there is no type derivation that can associate $\secan{\tconst~(m,n)}{\whigh}$ to $\while{x}{C}$.

\framebox[1.1\width]{Base case} 

Case $C= \wskip$. We know that $\regenv \enta{\regvar}{\var}  \wskip \typeproduce \code{\dots}, \secan{\tconst~(1,1)}{\whigh}, \regenv$ and that, for any memory $\mem$, $\whstate{\wskip}{\mem} \whsem{\tau} \whstate{\findot}{\mem}$. 

Case $C= \assign{x}{E}$. Assume that $\regenv \enta{\regvar}{\var}  \assign{x}{E} \typeproduce \code{\dots}, \secan{\tconst~(1,n+1)}{\whigh}, \regenv'$. Then we know that, for any memory $\mem$, it is true that  $\whstate{\assign{x}{E}}{\mem} \whsem{\tau} \whstate{\findot}{\mem[x \backslash |[E|](\mem)]}$.

Case $C= \out{ch}{E}$. Assume that $\regenv \enta{\regvar}{\var}  \out{ch}{E} \typeproduce \code{\dots}, \secan{\tconst~(1,n+1)}{\whigh}, \regenv'$. Then  we know that, for any memory $\mem$, $\whstate{\out{ch}{E}}{\mem} \whsem{ch!|[E|](\mem)} \whstate{\findot}{\mem}$. 

\framebox[1.1\width]{Inductive step} 

Case $C= \wif{E}{C_t}{C_e}$. Assume $\typeconcel{\regenv}{\wif{E}{C_t}{C_e}}{\code{\dots}}{\secan{\tconst~(m+1,n_g+max(n_t,n_e)+2)}{\whigh}}{\regenv_2 \envinters \regenv_3}{\enta{\regvar}{\var}}$. Then it must be $\regenv, \emptyRegAll \eenta{\regvar}{\var} E \etypeproduce \code{\dots}, \esecan{\valHigh}{n_g}, r, \regenv_1$ together with the derivation  $\regenv_1 \enta{\regvar}{\var}  C_t \typeproduce \code{\dots}, \secan{\tconst~(m,n_t)}{\whigh} , \regenv_2$ and $\regenv_1 \enta{\regvar}{\var}  C_e \typeproduce \code{\dots}, \secan{\tconst~(m,n_e)}{\whigh}, \regenv_3$. Hence, regardless of which branch is triggered by the reduction step, the property holds.

Case $C= C_1;C_2$. Assume $\regenv \enta{\regvar}{\var}  C_1;C_2 \typeproduce \code{\dots}, \secan{\tconst~(m,n)}{\whigh}, \regenv_\alpha$. We proceed by cases according to the semantic rule used to reduce the first element of the pair.\\
Assume $\mem$ is considered such that  rule $\wraprule{\rulename{c-1}}$ is applied. Hence there exists $C_3=\wif{E}{C_t}{C_e}$ such that $C_1;C_2 = R[C_3]$ and $\whstate{R[C_3]}{\mem} \whsem{l} \whstate{R[C_3']}{\mem'}$, given that $C_3' \not = \findot$.\\  
By Lemma \ref{lemma:ctxandtypes}, a type derivation 
$$
\inference
{\regenv \enta{\regvar}{\var}  C_3 \typeproduce \code{\dots}, \secan{\tconst~(m_a,n_a)}{\whigh}, \regenv_1 \\
\regenv_1 \enta{\regvar}{\var}  R[\findot] \typeproduce \code{\dots}, \secan{\tconst~(m_b,n_b)}{\whigh}, \regenv_\alpha}
{\regenv \enta{\regvar}{\var}  R[C_3] \typeproduce \code{\dots}, \secan{\tconst~(m,n)}{\whigh}, \regenv_\alpha}
$$
must exist such that $\tconst~(m,n)= \tconst~(m_a,n_a) \tclub \tconst~(m_b,n_b)$. By applying the inductive hypothesis on $C_3$ we know that for all $\mem$ we have $\whstate{C_3}{\mem} \whsem{l} \whstate{C_3'}{\mem'}$ and exists $\regenv'$ such that 
$\regenv' \enta{\regvar}{\var}  C_3' \typeproduce \code{\dots}, \secan{\tconst~(m_a-1,n_a')}{\whigh}, \regenv_1'$ and $\regenv_1' \sqsupseteq \regenv_1$. Hence 
$$
\inference
{\regenv' \enta{\regvar}{\var}  C_3' \typeproduce \code{\dots}, \secan{\tconst~(m_a-1,n_a')}{\whigh}, \regenv_1' \\
\regenv_1' \enta{\regvar}{\var}  R[\findot] \typeproduce \code{\dots}, \secan{\tconst~(m_b,n_b)}{\whigh}, \regenv_\alpha'}
{\regenv' \enta{\regvar}{\var}  R[C_3'] \typeproduce \code{\dots}, \secan{\tconst~(m-1,n_a'+n_b)}{\whigh}, \regenv_\alpha'}
$$
is a valid type derivation and $\regenv_\alpha' \sqsupseteq \regenv_\alpha$, by Lemma \ref{lemma:ctxandtypes} and Lemma \ref{lemma:comweak}. \\
Assume $\mem$ is considered such that  rule  $\wraprule{\rulename{c-2}}$ is applied instead. Hence there exist two commands $C_3 \in \{\wskip, \assign{x}{E}, \out{ch}{E}\}$ and $C_4$ such that $C_1;C_2 = R[C_3;C_4]$ and $\whstate{R[C_3;C_4]}{\mem} \whsem{l} \whstate{R[C_4]}{\mem'}$, given that $\whstate{C_3}{\mem} \whsem{l} \whstate{\findot}{\mem'}$. 
By the hypothesis about the type of $C$ and Lemma \ref{lemma:ctxandtypes}, a type derivation 
$$
\inference
{\regenv \enta{\regvar}{\var}  C_3 \typeproduce \code{\dots}, \secan{\tconst~(1,n_a)}{\whigh}, \regenv_1 \\
\regenv_1 \enta{\regvar}{\var}  R[\findot;C_4] \typeproduce \code{\dots}, \secan{\tconst~(m-1,n_b)}{\whigh}, \regenv_\alpha}
{\regenv \enta{\regvar}{\var}  R[C_3;C_4] \typeproduce \code{\dots}, \secan{\tconst~(m,n_a+n_b)}{\whigh}, \regenv_\alpha}
$$
hence $\regenv_1 \enta{\regvar}{\var}  R[C_4] \typeproduce \code{\dots}, \secan{\tconst~(m-1,n_b)}{\whigh}, \regenv_\alpha$ is a derivation proving the property.
\end{proof}

\begin{proof}[of Proposition \ref{prop:strongsectype}]

Let us consider the relation $R$ between \whprog\ programs defined as follows:

$$\begin{array}{lll}
H 	& = & \{ C | \exists \regenv. \regenv \enta{\regvar}{\var}  C \typeproduce \code{\dots}, \secan{t}{\whigh},\regenv' \} \cup \{ \findot \}\\
A   	& = & \{ (C,C) | \exists \regenv. \regenv \enta{\regvar}{\var}  C \typeproduce \code{\dots}, \secan{t}{w},\regenv' \} \\
X 	& = & \{(R[C_1],R[C_2]| \exists \regenv_\alpha \regenv_\beta \mbox{ such that } \\
	& 	&\regenv_\alpha \enta{\regvar}{\var}  R[C_1] \typeproduce \code{\dots}, \secan{t_1}{w_1},\regenv_\alpha' \mbox{ and } \\
	&	& \regenv_\beta \enta{\regvar}{\var}  R[C_2] \typeproduce \code{\dots}, \secan{t_2}{w_2},\regenv_\beta' \mbox{ and }\\
	&	& \regenv_\alpha \enta{\regvar}{\var}  C_1 \typeproduce \code{\dots}, \secan{\tconst~(m,n_1)}{\whigh},\regenv_1 \mbox{ and } \\
	&	& \regenv_\beta \enta{\regvar}{\var}  C_2 \typeproduce \code{\dots}, \secan{\tconst~(m,n_2)}{\whigh},\regenv_2 \mbox{ and } \\
	&	& \regenv_1 \envinters \regenv_2 \enta{\regvar}{\var}  R[\bullet] \typeproduce \code{\dots}, \secan{t_3}{w_3},\regenv_3 \}\\
R 	&  =	& A \cup X \cup H^2
\end{array}$$

We now show that $R$ is a strong bisimulation for \whprog\ programs. Recall that strong bisimulation is defined for termination transparent systems, hence we should lift the semantics of \whprog\ programs to a termination transparent system first, and then proceed with the proof. However, for improving readability, we take a different approach. Wherever it is possible (i.e. wherever the commands are not $\findot$) we reason directly on the semantics of the \whprog\ language, and we appeal to the features of the termination transparent semantics for \whprog\ programs only where it is strictly necessary (Part 2).\\

\emph{Part 1}.

We start by showing that pairs in $A$ perform transitions that correspond to low equivalent actions and preserve low equality of memories. Also, we show that transitions are performed towards pairs of commands that are either in $A$, $X$ or $H^2$. We proceed by cases on $C$.

Case $(\wskip,\wskip)$. We know that $\forall \mem =_{\pubdata} N $ $\whstate{\wskip}{\mem} \whsem{\tau} \whstate{\findot}{\mem}$ if and only if $\whstate{\wskip}{N} \whsem{\tau} \whstate{\findot}{N}$. Moreover, $(\findot,\findot) \in H^2$.

Case $(\assign{x}{E},\assign{x}{E})$. Assume that $\seclev(x)= \valHigh$. Then $\forall \mem =_{\pubdata} N $ $\whstate{\assign{x}{E}}{\mem} \whsem{\tau} \whstate{\findot}{\mem[x \backslash |[E|](\mem)]}$ if and only if $\whstate{\assign{x}{E}}{N} \whsem{\tau} \whstate{\findot}{N[x \backslash |[E|](N)]}$. Moreover $\mem[x \backslash |[E|](\mem)] =_{\pubdata} N[x \backslash |[E|](N)]$ since $\seclev(x)= \valHigh$. \newline
Assume that $\seclev(x)= \valLow$. Then we know that the derivation $\regenv \enta{\regvar}{\var} \assign{x}{E} \typeproduce \code{\dots},\secan{\tdontc}{\wdontk},\subst{\regenv'}{\conn{r}{x}{\writeMode}}$ holds under the assumption that $\regenv, \emptyRegAll \eenta{\regvar}{\var} E \etypeproduce \code{\dots}, \esecan{\valLow}{n}, r, \regenv'$. Since the security level of an expression represents the upper bound of the security levels of written registers, which in turns represent the upper bound of the security levels of read variables, the type derivation ensures that $E$ does not involve variables from $\valHigh$. Hence  $\forall \mem =_{\pubdata} N $ $|[E|](\mem)=|[E|](N)=k$. We therefore conclude that $\whstate{\assign{x}{E}}{\mem} \whsem{\tau} \whstate{\findot}{\mem[x \backslash k]}$ if and only if $\whstate{\assign{x}{E}}{N} \whsem{\tau} \whstate{\findot}{N[x \backslash k]}$. Moreover $\mem[x \backslash k] =_{\pubdata} N[x \backslash k]$.

Case $(\out{ch}{E},\out{ch}{E})$. Assume that $\seclev(ch)= \valHigh$. In this case $\forall \mem =_{\pubdata} N $ it must be that  $\whstate{\out{ch}{E}}{\mem} \whsem{ch!|[E|](\mem)} \whstate{\findot}{\mem}$ if and only if $\whstate{\out{ch}{E}}{N} \whsem{ch! |[E|](N)} \whstate{\findot}{N}$. Moreover $\low(ch!|[E|](\mem))=\low(ch!|[E|](\mem))= \tau$ since $\seclev(ch)= \valHigh$. \newline
Assume that $\seclev(ch)= \valLow$. Then we know that the derivation $\regenv \enta{\regvar}{\var} \out{ch}{E} \typeproduce \code{\dots},\secan{\tdontc}{\wdontk},\regenv'$ holds provided that $\regenv, \emptyRegAll \eenta{\regvar}{\var} E \etypeproduce \code{\dots}, \esecan{\valLow}{n}, r, \regenv'$. Since the security level of an expression represents the upper bound of the security levels of written registers, which in turns represent the upper bound of the security levels of read variables, the type derivation ensures that $E$ does not involve variables from $\valHigh$. Hence  $\forall \mem =_{\pubdata} N $ $|[E|](\mem)=|[E|](N)=k$. We therefore conclude that $\whstate{\out{ch}{E}}{\mem} \whsem{ch!k} \whstate{\findot}{\mem}$ if and only if $\whstate{\out{ch}{E}}{N} \whsem{ch!k} \whstate{\findot}{N}$. 

Case $(\wif{E}{C_t}{C_e},\wif{E}{C_t}{C_e})$. Regardless of the typing derivation used for including the pair in $A$, for any memory $\mem$ two cases are possible: either $\whstate{\wif{E}{C_t}{C_e}}{\mem} \whsem{\tau} \whstate{C_t}{\mem}$ or  $\whstate{\wif{E}{C_t}{C_e}}{\mem} \whsem{\tau} \whstate{C_e}{\mem}$. Hence, regardless of $\mem$, the transition produces a silent action and does not alter the memory. This reduces checking the strong bisimulation condition to checking that the pairs of reduced commands are in $R$. This aspect is ensured by a careful analysis of the type derivation used to include the pair in $A$.\\ An $\wifName$ statement can be typed according to two rules, $\wraprule{\rulename{\wifName}-any}$ and $\wraprule{\rulename{\wifName}-\valHigh}$. \\
Assume the rule $\wraprule{\rulename{\wifName}-any}$ is used. Then we know that the derivation
$\typeconcel{\regenv}{\wif{E}{C_t}{C_e}}{\code{\dots}} {\secan{\writeMap(\lambda) \wlub w_1 \wlub w_2}{\termMap(\lambda)\tlub t_1 \tlub t_2 }}{\regenv_F}{\enta{\regvar}{\var}}$  holds 
under the assumption that 
$\regenv, \emptyRegAll \eenta{\regvar}{\var} E \etypeproduce \code{\dots}, \esecan{\lambda}{n_g} , r, \regenv_1$ 
and  
$\regenv_1 \enta{\regvar}{\var}  C_t \typeproduce \code{\dots}, \secan{t_1}{w_1}, \regenv_2$ 
and 
$\regenv_1 \enta{\regvar}{\var}  C_e \typeproduce \code{\dots}, \secan{t_2}{w_2}, \regenv_3$ and $\writeMap(\lambda) \sqsupseteq w_i$
Assume that $\lambda= \valHigh$. Then $\writeMap(\lambda)= \valHigh$ and $w_1= w_2 = \valHigh$. Hence, all possible pairs of reduced terms, namely $(C_t,C_t)$, $(C_t,C_e)$, $(C_e,C_t)$ and $(C_e,C_e)$ are in $H^2$. Assume that $\lambda= \valLow$. Since the security level of an expression represents the upper bound of the security levels of written registers, which in turns represent the upper bound of the security levels of read variables, the type derivation ensures that $E$ does not involve variables from $\valHigh$. Then $\forall \mem =_{\pubdata} N $ $\whstate{\wif{E}{C_t}{C_e}}{\mem} \whsem{\tau} \whstate{C_{*}}{\mem}$ if and only if $\whstate{\wif{E}{C_t}{C_e}}{N} \whsem{\tau} \whstate{C_{*}}{N}$, for $C_{*} \in \{C_t,C_e\}$. Moreover, $(C_{*},C_{*}) \in A$ follows directly by Proposition \ref{prop:subjred}. \\
Assume the rule $\wraprule{\rulename{\wifName}-\valHigh}$ is used instead.
Then we know that it must be 
$\typeconcel{\regenv}{\wif{E}{C_t}{C_e}}{\code{\dots}}{\secan{\tconst~(m+1,n_g+max(n_t,n_e)+2)}{\whigh}}{\regenv_F}{\enta{\regvar}{\var}}$ under the assumption that  $\regenv_1 \enta{\regvar}{\var}  C_t \typeproduce \code{\dots}, \secan{\tconst~(m,n_t)}{\whigh}, \regenv_2$ and $\regenv_1 \enta{\regvar}{\var}  C_e \typeproduce \code{\dots}, \secan{\tconst~(m,n_e)}{\whigh}, \regenv_3$. Hence, as it was concluded in the previous subcase, all possible pairs of reduced terms, namely $(C_t,C_t)$, $(C_t,C_e)$, $(C_e,C_t)$ and $(C_e,C_e)$ are in $H^2$.

Case $(\while{x}{C},\while{x}{C})$. Regardless of the typing derivation used for including the pair in $A$, for any memory $\mem$ two cases are possible: either it is true that $\whstate{\while{x}{C}}{\mem} \whsem{\tau} \whstate{\findot}{\mem}$ or  $\whstate{\while{x}{C}}{\mem} \whsem{\tau} \whstate{C;\while{x}{C}}{\mem}$. Hence, regardless of $\mem$, the transition produces a silent action and does not alter the memory. This reduces checking the strong bisimulation condition to checking that the pairs of reduced commands are in $R$. This aspect is ensured by a careful analysis of the type derivation used to include the pair in $A$. We know that $\typeconcel{\regenv}{\while{x}{C}}{\dots}{\secan{\termMap(\lambda)\tlub t}{\writeMap(\lambda) \wlub w}}{\regenv_B}{\enta{\regvar}{\var}}$ providing that $\regenv_B \enta{\regvar}{\var}  C \typeproduce \code{\dots}, \secan{t}{w} , \regenv_E$ and $\writeMap(\lambda) \sqsupseteq w$ and  $t= \ttop \Rightarrow \writeMap(\lambda)= \whigh$, for $\lambda = \seclev(x)= \seclev(r)$. Assume that $\lambda= \valHigh$. Then $\writeMap(\lambda)= w= \whigh$, hence all possible pairs of reduced terms, namely $(C,C)$, $(C,\findot)$, $(\findot,C)$ and $(\findot,\findot)$ are in $H^2$. Assume that $\lambda= \valLow$. Then $\forall \mem =_{\pubdata} N $ $\whstate{\while{x}{C}}{\mem} \whsem{\tau} \whstate{C_{*}}{\mem}$ if and only if $\whstate{\while{x}{C}}{N} \whsem{\tau} \whstate{C_{*}}{N}$, for $C_{*} \in \{C;\while{x}{C},\findot\}$. Moreover we can conclude that the pair $(C;\while{x}{C},C;\while{x}{C}) $ is in $A$ by Proposition \ref{prop:subjred}, and $(\findot,\findot) \in H^2$ by definition. \\

Case $(C_1;C_2,C_1;C_2)$. We proceed by cases according to the semantic rule used to reduce the first element of the pair.

Assume $M$ is considered such that  rule $\wraprule{\rulename{c-1}}$ is applied. Hence there exists $C_3 \in \{\wif{E}{C_t}{C_e},\while{x}{C}\}$ such that $C_1;C_2 = R[C_3]$ and $\whstate{R[C_3]}{\mem} \whsem{l} \whstate{R[C_3']}{\mem'}$, given that $C_3' \not = \findot$.\\ 
Let us consider $C_3=\wif{E}{C_t}{C_e}$. As it has been observed previously, for any memory $\mem$ two cases are possible: either we have that  $\whstate{\wif{E}{C_t}{C_e}}{\mem} \whsem{\tau} \whstate{C_t}{\mem}$ or  $\whstate{\wif{E}{C_t}{C_e}}{\mem} \whsem{\tau} \whstate{C_e}{\mem}$. Hence, regardless of $\mem$, the transition produces a silent action and does not alter the memory. We only need to ensure that the pair of reduced commands remains in $R$. \\
Assume the rule $\wraprule{\rulename{\wifName}-any}$ was used for typing $C_3$. Then we know 
$\typeconcel{\regenv}{C_3}{\code{\dots}} {\secan{\termMap(\lambda)\tlub t_1 \tlub t_2}{\writeMap(\lambda) \wlub w_1 \wlub w_2 }}{\regenv_F}{\enta{\regvar}{\var}}$    
under the assumption that 
$\regenv, \emptyRegAll \eenta{\regvar}{\var} E \etypeproduce \code{\dots}, \esecan{\lambda}{n_g} , r, \regenv_1$ 
and  
$\regenv_1 \enta{\regvar}{\var}  C_t \typeproduce \code{\dots}, \secan{t_1}{w_1}, \regenv_2$ 
and 
$\regenv_1 \enta{\regvar}{\var}  C_e \typeproduce \code{\dots}, \secan{t_2}{w_2}, \regenv_3$ and $\writeMap(\lambda) \sqsupseteq w_i$, such that $\regenv_F=\regenv_2 \envinters \regenv_3$. 
Assume that $\lambda= \valHigh$. Then $\writeMap(\lambda)= \valHigh$ and $w_1= w_2 = \valHigh$, hence $\typeconcel{\regenv}{C_3}{\code{\dots}} {\secan{\ttop}{\whigh }}{\regenv_F}{\enta{\regvar}{\var}}$.  Since $R[C_3]$ is typable, $R[\findot]$ must be typable (Lemma \ref{lemma:ctxandtypes}) and it must have type $\secan{t}{\whigh}$ for an appropriate $t$. Hence, $R[C_t]$ must have type $\secan{t'}{\whigh}$ and $R[C_e]$ must have type $\secan{t''}{\whigh}$ (Proposition \ref{prop:subjred}). Therefore all possible pairs of reduced terms, namely $(R[C_t],R[C_t])$, $(R[C_t],R[C_e])$, $(R[C_e],R[C_t])$ and $(R[C_e],R[C_e])$ are in $H^2$. 
Assume that $\lambda= \valLow$. Since the security level of an expression represents the upper bound of the security levels of written registers, which in turns represent the upper bound of the security levels of read variables, the type derivation ensures that $E$ does not involve variables from $\valHigh$. Then for all memories $N$ such that $\mem =_{\pubdata} N $ $\whstate{\wif{E}{C_t}{C_e}}{\mem} \whsem{\tau} \whstate{C_{*}}{\mem}$ implies $\whstate{\wif{E}{C_t}{C_e}}{N} \whsem{\tau} \whstate{C_{*}}{N}$, for $C_{*} \in \{C_t,C_e\}$. Moreover, by Proposition \ref{prop:subjred}, $R[C_{*}]$ is typable, hence pairs $(R[C_{*}],R[C_{*}])$ are in $A$.\\
Assume the rule $\wraprule{\rulename{\wifName}-\valHigh}$ is used instead. Then we know that the derivation  
$\typeconcel{\regenv}{C_3}{\code{\dots}}{\secan{\tconst~(m+1,n_g+max(n_t,n_e)+2)}{\whigh}}{\regenv_2 \envinters \regenv_3}{\enta{\regvar}{\var}}$ holds under the assumption that  $\regenv, \emptyRegAll \eenta{\regvar}{\var} E \etypeproduce \code{\dots}, \esecan{\valHigh}{n_g} , r, \regenv_1$, $\regenv_1 \enta{\regvar}{\var}  C_t \typeproduce \code{\dots}, \secan{\tconst~(m,n_t)}{\whigh}, \regenv_2$ and $\regenv_1 \enta{\regvar}{\var}  C_e \typeproduce \code{\dots}, \secan{\tconst~(m,n_e)}{\whigh}, \regenv_3$. Since $R[C_3]$ is typable, $R[\findot]$ must be typable in $\regenv_2 \envinters \regenv_3$ (Lemma \ref{lemma:ctxandtypes}). Also, since $\regenv_i \sqsupseteq \regenv_2 \envinters \regenv_3$, for $i \in \{2,3\}$, $R[\findot]$ must be also typable in $\regenv_2$ and $\regenv_3$ (Lemma \ref{lemma:comweak}). Hence all possible pairs of reduced terms, namely $(R[C_t],R[C_t])$, $(R[C_t],R[C_e])$, $(R[C_e],R[C_t])$ and $(R[C_e],R[C_e])$ are in $X^2$.\\
Let us consider $C_3=\while{x}{C}$. Since $C_1;C_2$ is typable, both $C_3$ and $R[\findot]$ must be typable (Lemma \ref{lemma:ctxandtypes}). We know that $\typeconcel{\regenv}{C_3}{\code{\dots}}{\secan{\termMap(\lambda)\tlub t}{\writeMap(\lambda) \wlub w}}{\regenv_B}{\enta{\regvar}{\var}}$ for $\lambda = \seclev(x)= \seclev(r)$, providing that $\regenv_B \enta{\regvar}{\var}  C \typeproduce \code{\dots}, \secan{t}{w} , \regenv_E$ and $\writeMap(\lambda) \sqsupseteq w$ and $t= \ttop \Rightarrow \writeMap(\lambda)= \whigh$. 
Assume that $\lambda= \valHigh$. Then $\writeMap(\lambda)= \valHigh$ and $w = \valHigh$, hence 
$\typeconcel{\regenv}{C_3}{\code{\dots}} {\secan{\ttop}{\whigh }}{\regenv_B}{\enta{\regvar}{\var}}$ and 
$R[\findot]$ must have type $\secan{t}{\whigh}$ for an appropriate $t$.
Since 
we know, for the assumption on $\mem$, that $\whstate{R[C_3]}{\mem} \whsem{\tau} \whstate{R[C;C_3]}{\mem}$ it must be, by Proposition \ref{prop:subjred}, that also $R[C;C_3]$ has type $\secan{t'}{\whigh}$, for a suitable $t'$. Hence, for all memories $N$ such that $\mem =_{\pubdata} N$, either $\whstate{R[C_3]}{N} \whsem{\tau} \whstate{R[C;C_3]}{N}$ and $(R[C;C_3],R[C;C_3])$ is in $H^2$, or $\whstate{R[C_3]}{N} \whsem{\tau} \whstate{R[\findot]}{N}$ and therefore $(R[C;C_3],R[\findot])$ is in $H^2$. 
Assume that $\lambda= \valLow$. Then for all memories $N$ such that $\mem =_{\pubdata} N $ $\whstate{R[C_3]}{N} \whsem{\tau} \whstate{R[C;C_3]}{N}$. By Proposition \ref{prop:subjred} $R[C;C_3]$ is typable, hence the pair $(R[C;C_3],R[C;C_3])$ is in $A$. \\

Assume the rule $\wraprule{\rulename{c-2}}$ is applied instead. Hence there exist two commands $C_3$ and $C_4$ such that $C_1;C_2 \equiv R[C_3;C_4]$ and $\whstate{R[C_3;C_4]}{\mem} \whsem{l} \whstate{R[C_4]}{\mem'}$, given that $\whstate{C_3}{\mem} \whsem{l} \whstate{\findot}{\mem'}$. We distinguish two cases. Assume $C_3$ in $\{\wskip, \assign{x}{E}, \out{ch}{E}\}$. For what we observed previously, since $C_3$ is typable, $\forall \mem =_{\pubdata} N $  $\whstate{C_3}{\mem} \whsem{l} \whstate{\findot}{\mem'}$ if and only if $\whstate{C_3}{N} \whsem{l'} \whstate{\findot}{N'}$ such that $\low(l)=\low(l')$ and $\mem' =_{\pubdata} N'$. Also, by Proposition \ref{prop:subjred}, $R[C_4]$ is typable, hence $(R[C_4],R[C_4]) \in A$. Assume that $C_3$ is $\while{x}{C}$. 
Since $C_1;C_2$ is typable, both $C_3$ and $R[C_4]$ must be typable (Lemma \ref{lemma:ctxandtypes}). We know that $\typeconcel{\regenv}{C_3}{\code{\dots}}{\secan{\termMap(\lambda)\tlub t}{\writeMap(\lambda) \wlub w}}{\regenv_B}{\enta{\regvar}{\var}}$ for $\lambda = \seclev(x)= \seclev(r)$, providing that $\regenv_B \enta{\regvar}{\var}  C \typeproduce \code{\dots}, \secan{t}{w} , \regenv_E$ and $\writeMap(\lambda) \sqsupseteq w$ and $t= \ttop \Rightarrow \writeMap(\lambda)= \whigh$. 
Assume that $\lambda= \valHigh$. Then $\writeMap(\lambda)= \valHigh$ and $w = \valHigh$, hence 
$\typeconcel{\regenv}{C_3}{\code{\dots}} {\secan{\ttop}{\whigh }}{\regenv_B}{\enta{\regvar}{\var}}$ and 
$\regenv_B \enta{\regvar}{\var}  C \typeproduce \code{\dots}, \secan{t}{\whigh} , \regenv_E$. Since $R[C_4]$ is typable (Proposition \ref{prop:subjred}), $R[\findot]$ must have type $\secan{t}{\whigh}$ for an appropriate $t$. This implies that also $R[C_3;C_4]$ has type $\secan{\ttop}{\whigh}$. Hence, for all memories $N$ such that $\mem =_{\pubdata} N$, either $\whstate{R[C_3;C_4]}{N} \whsem{\tau} \whstate{R[C;C_3;C_4]}{N}$ and $(R[C_4],R[C;C_3;C_4])$ is in $H^2$, or we have that $\whstate{R[C_3;C_4]}{N} \whsem{\tau} \whstate{R[C_4]}{N}$ and $(R[C_4],R[C_4])$ is in $H^2$. 
Assume that $\lambda= \valLow$. Then for all memories $N$ such that $\mem =_{\pubdata} N $ $\whstate{R[C_3;C_4]}{N} \whsem{l} \whstate{R[C_4]}{N}$. By Proposition \ref{prop:subjred} $R[C_4]$ is typable, hence the pair $(R[C_4],R[C_4])$ is in $A$. \\


\emph{Part 2}. 

Let $C \in H$ such that $C \not = \findot$.  Lemma \ref{lemma:prophighcom} ensures that, for any memory $\mem$,  $C$ does not produce an action different than $\tau$ or alter the low part of $\mem$. Proposition \ref{prop:subjred} ensures that the reduced command remains in $H$. Also, recall that $\whstate{\findot}{\mem} \whsemtt{\tau} \whstate{\findot}{\mem}$. This suffices to show that, according to the termination transparent semantics of \whprog, all pairs in $H^2$ perform transitions that correspond to low equivalent actions, in which low equality of memories is preserved and such that reduced pairs of commands are still in $H^2$. \\ 

\emph{Part 3}. 

We complete the proof by showing that pairs in $X$ perform transitions that correspond to low equivalent actions and preserve low equality of memories. Also, we show that transitions are performed towards pairs of commands that are either in $X$ or in $A$. 
Let $(R[C_1],R[C_2]) \in X$. By definition there must exist $\regenv_1$ and $\regenv_2$ such that the following derivation are correct:
$$
\begin{array}{c}
\regenv_1 \enta{\regvar}{\var}  C_1 \typeproduce \code{\dots}, \secan{\tconst~(m,n_1)}{\whigh},\regenv_1' \\
\regenv_2 \enta{\regvar}{\var}  C_2 \typeproduce \code{\dots}, \secan{\tconst~(m,n_2)}{\whigh},\regenv_2' \\
\regenv_1' \envinters \regenv_2' \enta{\regvar}{\var}  R[\findot] \typeproduce \code{\dots}, \secan{t}{w},\regenv_3 \\
\end{array}
$$
Consider $\mem =_{\pubdata} N$ such that $\whstate{C_1}{\mem} \whsem{l_1} \whstate{C_1'}{\mem'}$ and $\whstate{C_2}{N} \whsem{l_2} \whstate{C_2'}{N'}$.
Lemma \ref{lemma:prophighcom} ensures that $\low(l_1)=\low(l_2)= \tau$ and $\mem' =_\pubdata N'$. 
If $m=1$ then both $C_1'$ and $C_2'$ are $\findot$ and $\regenv_1' \envinters \regenv_2' \enta{\regvar}{\var}  R[\findot] \typeproduce \code{\dots}, \secan{t}{w},\regenv_3$ ensures that $(R[\findot],R[\findot]) \in R$. 
If $m>1$ then Lemma \ref{lemma:prophighstep}  ensures that there exists $\regenv_\alpha$ and $\regenv_\beta$ for which the following derivations are correct
$$
\begin{array}{c}
\regenv_\alpha \enta{\regvar}{\var}  C_1' \typeproduce \code{\dots}, \secan{\tconst~(m-1,n_1')}{\whigh},\regenv_\alpha' \\
\regenv_\beta \enta{\regvar}{\var}  C_2' \typeproduce \code{\dots}, \secan{\tconst~(m-1,n_2')}{\whigh},\regenv_\beta' \\
\end{array}
$$
for  $\regenv_\alpha' \sqsupseteq \regenv_1'$ and $\regenv_\beta' \sqsupseteq \regenv_2'$. Hence it is true that $(R[C_1'],R[C_2']) \in X$ because $\regenv_\alpha \enta{\regvar}{\var}  R[C_1'] \typeproduce \code{\dots} ,\secan{t_a}{w_a} ,\regenv_a$ (Lemmas \ref{lemma:ctxandtypes} and \ref{lemma:comweak}), $\regenv_\beta \enta{\regvar}{\var}  R[C_2'] \typeproduce \code{\dots} ,\secan{t_b}{w_b}, \regenv_b$ (Lemmas \ref{lemma:ctxandtypes}  and \ref{lemma:comweak}) and $\regenv_\alpha' \envinters \regenv_\beta' \enta{\regvar}{\var}  R[\findot] \typeproduce \code{\dots} ,\secan{t_c}{w_c} ,\regenv_c$ (Lemma \ref{lemma:comweak}). 
\end{proof}

We have shown that a type-correct \whprog\ program is strongly secure. We now have to show that the type system in Figures \ref{table:typesystematom} and \ref{table:typesystem} translates a secure \whprog\ program into a secure \swhprog\ program. Since \swhprog\ programs form a subclass of \whprog\ programs, this can be shown by retyping a translated program and showing it is indeed type-correct. 

In order to state this result formally, we consider that now programs have variables in the set $\var'= \var \cup \regvar$. Moreover, from the set of register variables $\regvar= \{r_1,\dots,r_n\}$ we define a new set of register variables, $\regvar'= \{s_1,\dots,s_n\} \cup  \{s_1',\dots,s_n'\}$, such that for a (former) register variable $r_i$ there are a corresponding register variable $\compa{r_i}=s_i$ and a corresponding shadow variable $\shado{r_i}=s_i'$. We assume that $\seclev(r_i)= \seclev(s_i)= \seclev(s_i')$. 

In order to show that the code $D$ obtained from the compilation of a \whprog\ program $C$ is retypable, we proceed via a stronger property. In particular we show that the recompilation of $D$ (i) maps two related input register records into two related output register records and (ii) computes a security annotation which is equivalent (up-to slowdown) to the one corresponding to $D$. We formalize the first aspect as follows.

\begin{definition}[Register Correspondence]
Let $\regenv$ and $\regenv'$ be two register records defined over $\regvar$ and $\regvar'$ respectively. We say that they are corresponding, written as $\regenv \regcorresp \regenv'$, if:
\begin{itemize}
\item $\conn{r_i}{x}{\writeMode} \in \regenv \Rightarrow \conn{s_i}{x}{\writeMode} \in \regenv'$
\item $\conn{r_i}{x}{\readMode} \in \regenv \Rightarrow \conn{s_i}{r}{\writeMode} \in \regenv'$
 \end{itemize}
\end{definition}
Notice that for any register variable $r_i$ associated in $\regenv$ we explicitly require the corresponding register variable $s_i$ (not the shadow register variable $s_i'$) to be associated in $\regenv'$. 

In order to state and prove the result about recompilation of compiled programs, we begin by stating an auxiliary result that formalize the type correctness of the code corresponding to the evaluation of a \whprog\ expression.

\begin{proposition}[Retyping of expressions]\label{ref:exprrety}
Let $E$ be a \whprog\ expression such that there exists $\regenv$ for which $\regenv, \activeReg \eenta{\regvar}{\var} E \etypeproduce \code{D}, \esecan{\lambda}{n}, r, \regenv_\alpha$ holds. Then for all $\regenv'$ such that $\regenv \regcorresp \regenv'$ we  have that there exists a derivation $\regenv' \enta{\regvar'}{\var'}  D \typeproduce \code{D'}, \secan{t}{w},\regenv_\alpha'$ such that $\regenv_\alpha \regcorresp \regenv_\alpha'$ and: 
\begin{itemize}
\item if $\lambda= \valHigh$ or $n=0$, then $w= \whigh$ and $t=(n,p)$;
\item if $\lambda= \valLow$ and $n>0$, then $w= \wdontk$ and $t= \tdontc$.
\end{itemize}
\end{proposition}
\begin{proof}

We prove the proposition by induction on the structure of $E$ and by cases on the last rule applied in the type derivation. 

\framebox[1.1\width]{Base case} 

Case $E=k$. Assume the following derivation exists
$$\inference
{r \in \regvar & r \not \in \activeReg & \seclev(r)=\valHigh}
{\regenv, \activeReg \eenta{\regvar}{\var} k \etypeproduce \code{\assign{r}{k}}, \esecan{\valHigh}{1}, r, \subst{\regenv}{\breakconn{r}}}$$
and consider $\regenv'$, such that $\regenv \regcorresp \regenv'$, together with $s = \compa{r}$. Then the following derivation exists.
$$\inference
{\inference{s \in \regvar'  & \seclev(s)=\valHigh}
{\regenv', \{\} \eenta{\regvar'}{\var'} k \etypeproduce \code{\assign{s}{k}}, \esecan{\valHigh}{1}, s, \subst{\regenv'}{\breakconn{s}}}
}
{\regenv' \enta{\regvar'}{\var'} \assign{r}{k} \typeproduce \code{\assign{s}{k};\assign{r}{s}}, \secan{\tconst~(1,2)}{\whigh}, \subst{\regenv'}{\conn{s}{r}{\writeMode}}} 
$$
We have that $\subst{\regenv}{\breakconn{r}} \regcorresp \subst{\regenv'}{\conn{s}{r}{\writeMode}}$ because:
\begin{itemize}
\item for any $r' \not = r$, if $\conn{r'}{x}{\writeMode} \in \subst{\regenv}{\breakconn{r}}$ then $\conn{s'}{x}{\writeMode} \in \subst{\regenv'}{\conn{s}{r}{\writeMode}}$ because $s' \not = \compa{r}$;
\item for any $r' \not = r$, if $\conn{r'}{x}{\readMode} \in \subst{\regenv}{\breakconn{r}}$ then $\conn{s'}{r'}{\writeMode} \in  \subst{\regenv'}{\conn{s}{r}{\writeMode}}$ because $s' \not = \compa{r}$;
\item $r$ is unassociated in $\subst{\regenv}{\breakconn{r}}$.
\end{itemize}
Assume the following derivation exists  
$$\inference
{r \in \regvar & r \not \in \activeReg & \seclev(r)=\valLow}
{\regenv, \activeReg \eenta{\regvar}{\var} k \etypeproduce \code{\assign{r}{k}}, \esecan{\valLow}{1}, r, \subst{\regenv}{\breakconn{r}}}$$
instead and consider $\regenv'$, such that $\regenv \regcorresp \regenv'$, together with $s = \compa{r}$. Then the following derivation exists
$$
\inference
{\inference{s \in \regvar' & \seclev(s)=\valLow}
{\regenv', \{\} \eenta{\regvar'}{\var'} k \etypeproduce \code{\assign{s}{k}}, \esecan{\valLow}{1}, s, \subst{\regenv'}{\breakconn{s}}}
}
{\regenv' \enta{\regvar'}{\var'} \assign{r}{k} \typeproduce \code{\assign{s}{k};\assign{r}{s}}, \secan{\tdontc}{\wdontk}, \subst{\regenv'}{\conn{s}{r}{\writeMode}}} 
$$
We have that $\subst{\regenv}{\breakconn{r}} \regcorresp \subst{\regenv'}{\conn{s}{r}{\writeMode}}$ because:
\begin{itemize}
\item for any $r' \not = r$, if $\conn{r'}{x}{\writeMode} \in \subst{\regenv}{\breakconn{r}}$ then $\conn{s'}{x}{\writeMode} \in \subst{\regenv'}{\conn{s}{r}{\writeMode}}$ because $s' \not = \compa{r}$;
\item for any $r' \not = r$, if $\conn{r'}{x}{\readMode} \in \subst{\regenv}{\breakconn{r}}$ then $\conn{s'}{r'}{\writeMode} \in  \subst{\regenv'}{\conn{s}{r}{\writeMode}}$ because $s' \not = \compa{r}$;
\item $r$ is unassociated in $\subst{\regenv}{\breakconn{r}}$.
\end{itemize}

Case $E=x$. Assume the following derivation exists
$$\inference{x \in \var & r \in \regvar & \conn{r}{x}{} \in \regenv}
{\regenv, \activeReg \eenta{\regvar}{\var} x \etypeproduce \code{\findot},\esecan{\seclev(r)}{0}, r, \regenv }$$
and consider $\regenv'$ such that $\regenv \regcorresp \regenv'$. Then the following derivation exists
$$
\inference{}{\regenv' \enta{\regvar'}{\var'} \findot \typeproduce \code{\findot}, \secan{\tconst~(0,0)}{\whigh}, \regenv'}$$
Assume the following derivation exists instead 
$$\inference
{x \in \var & r \in \regvar & r   \not \in \activeReg & \seclev(r) = \valHigh \sqsupseteq \seclev(x)} 
{\regenv, \activeReg \eenta{\regvar}{\var} x \etypeproduce \code{\assign{r}{x}}, \esecan{\valHigh}{1}, r, \subst{\regenv}{\conn{r}{x}{\readMode}}}$$
Then for $\regenv'$, such that $\regenv \regcorresp \regenv'$ and $s = \compa{r}$ the following derivation exists
$$\inference
{\inference{s \in \regvar'  & \seclev(s)=\valHigh}
{\regenv', \{\} \eenta{\regvar'}{\var'} x \etypeproduce \code{\assign{s}{x}}, \esecan{\valHigh}{1}, s, \subst{\regenv'}{\conn{s}{x}{\readMode}}}
}
{\regenv' \enta{\regvar'}{\var'} \assign{r}{x} \typeproduce \code{\assign{s}{x};\assign{r}{s}}, \secan{\tconst~(1,2)}{\whigh}, \subst{\regenv'}{\conn{s}{r}{\writeMode}}}$$
We have that $\subst{\regenv}{\conn{r}{x}{\readMode}} \regcorresp \subst{\regenv'}{\conn{s}{r}{\writeMode}}$:
\begin{itemize}
\item for any $r' \not = r$ and $y \not = x$, if $\conn{r'}{y}{\writeMode} \in \subst{\regenv}{\conn{r}{x}{\readMode}}$ then $\conn{s'}{y}{\writeMode} \in \subst{\regenv'}{\conn{s}{r}{\writeMode}}$ because $s' \not = \compa{r}$;
\item for any $r' \not = r$ and $y \not = x$, if $\conn{r'}{y}{\readMode} \in \subst{\regenv}{\conn{r}{x}{\readMode}}$ then $\conn{s'}{r'}{\writeMode} \in  \subst{\regenv'}{\conn{s}{r}{\writeMode}}$ because $s' \not = \compa{r}$;
\item $\conn{r}{x}{\readMode} \in \subst{\regenv}{\conn{r}{x}{\readMode}}$ and  $\conn{s}{r}{\writeMode} \in  \subst{\regenv'}{\conn{s}{r}{\writeMode}}$.
\end{itemize}
If
$$\inference
{x \in \var & r \in \regvar & r   \not \in \activeReg & \seclev(r) = \valLow \sqsupseteq \seclev(x)} 
{\regenv, \activeReg \eenta{\regvar}{\var} x \etypeproduce \code{\assign{r}{x}}, \esecan{\valLow}{1}, r, \subst{\regenv}{\conn{r}{x}{\readMode}}}$$
then for $s = \compa{r}$ and $\regenv'$ such that $\regenv \regcorresp \regenv'$ the following derivation exists
$$\inference
{\inference{s \in \regvar'  & \seclev(s)=\valLow}
{\regenv', \{\} \eenta{\regvar'}{\var'} x \etypeproduce \code{\assign{s}{x}}, \esecan{\valLow}{1},s, \subst{\regenv'}{\conn{s}{x}{\readMode}}}
}
{\regenv' \enta{\regvar'}{\var'} \assign{r}{x} \typeproduce \code{\assign{s}{x};\assign{r}{s}}, \secan{\tdontc}{\wdontk}, \subst{\regenv'}{\conn{s}{r}{\writeMode}}}$$
We have that $\subst{\regenv}{\conn{r}{x}{\readMode}} \regcorresp \subst{\regenv'}{\conn{s}{r}{\writeMode}}$:
\begin{itemize}
\item for any $r' \not = r$ and $y \not = x$, if $\conn{r'}{y}{\writeMode} \in \subst{\regenv}{\conn{r}{x}{\readMode}}$ then $\conn{s'}{y}{\writeMode} \in \subst{\regenv'}{\conn{s}{r}{\writeMode}}$ because $s' \not = \compa{r}$;
\item for any $r' \not = r$ and $y \not = x$, if $\conn{r'}{y}{\readMode} \in \subst{\regenv}{\conn{r}{x}{\readMode}}$ then $\conn{s'}{r'}{\writeMode} \in  \subst{\regenv'}{\conn{s}{r}{\writeMode}}$ because $s' \not = \compa{r}$;
\item $\conn{r}{x}{\readMode} \in \subst{\regenv}{\conn{r}{x}{\readMode}}$ and  $\conn{s}{r}{\writeMode} \in  \subst{\regenv'}{\conn{s}{r}{\writeMode}}$.
\end{itemize}

\framebox[1.1\width]{Inductive step} 

Case $E=\wop{E_1}{E_2}$. 

Consider that the following derivation exists
$$\inference
{r,r_a \in \regvar \\
\regenv, \activeReg \eenta{\regvar}{\var} E_1 \etypeproduce \code{D_1}, \esecan{\valHigh}{n_1}, r, \regenv_1 \\ 
\regenv_1, \activeReg\cup\{r\} \eenta{\regvar}{\var} E_2 \etypeproduce \code{D_2}, \esecan{\valHigh}{n_2}, r_a, \regenv_2\\
\valHigh=\seclev(r)=\seclev(r_a)}
{\regenv, \activeReg \eenta{\regvar}{\var} \wop{E_1}{E_2} \etypeproduce 
\left \{
\begin{array}{l}
D_1;\\
D_2;\\
\assign{r}{\wop{r}{r_a}}
\end{array}
\right \}
, \esecan{\valHigh}{n_1+n_2+1}, r, \subst{\regenv_2}{\breakconn{r}}}$$
and consider $\regenv'$ such that $\regenv \regcorresp \regenv'$. Then the following derivations exist by inductive hypothesis
$$
\begin{array}{c}
\regenv' \enta{\regvar'}{\var'} D_1 \typeproduce \code{D_1'}, \secan{\tconst~(n_1,p_1)}{\whigh}, \regenv_1'\\
  \regenv_1' \enta{\regvar'}{\var'} D_2 \typeproduce \code{D_2'}, \secan{\tconst~(n_2,p_2)}{\whigh}, \regenv_2'
  \end{array}$$
such that $\regenv_2 \regcorresp \regenv_2'$.
In order for completing the case, we have to show that $\regenv_2' \enta{\regvar'}{\var'} \assign{r}{\wop{r}{r_a}} \typeproduce \code{D'}, \secan{\tconst~(1,p)}{\whigh}, \regenv_3$ such that $\subst{\regenv_2}{\breakconn{r}} \regcorresp \regenv_3$. 
We distinguish four cases. \\
Case 1: there exist $\conn{r}{x}{\readMode}, \conn{r_a}{y}{\readMode}  \in \regenv_2$. Since $\regenv_2 \regcorresp \regenv_2'$, then we know $\conn{s}{r}{\writeMode}, \conn{s_a}{r_a}{\writeMode}  \in \regenv_2'$ and the following type derivation exists
$$\inference
{\inference
{s,s_a \in \regvar'  & \seclev(s)= \valHigh}
{\regenv_2', \{\} \eenta{\regvar'}{\var'} \wop{r}{r_a} \etypeproduce \code{\assign{s}{\wop{s}{s_a}}}, \esecan{\valHigh}{1}, s,\subst{\regenv_2'}{\breakconn{s}}
}}
{\regenv_2' \enta{\regvar'}{\var'}  \assign{r}{\wop{r}{r_a}}  \typeproduce \code{\assign{s}{\wop{s}{s_a}};\assign{r}{s}}, \secan{\tconst~(1,2)}{\whigh}, \subst{\regenv_2'}{\conn{s}{r}{\writeMode}}}$$
We have that $\subst{\regenv_2}{\breakconn{r}} \regcorresp \subst{\regenv_2'}{\conn{s}{r}{\writeMode}}$:
\begin{itemize}
\item for any $r'' \not = r$, if $\conn{r''}{z}{\writeMode} \in \subst{\regenv_2}{\breakconn{r}}$ then $\conn{s''}{z}{\writeMode} \in\subst{\regenv_2'}{\conn{s}{r}{\writeMode}}$ because $s'' \not = \compa{r}$;
\item for any $r'' \not = r$, if $\conn{r''}{z}{\readMode} \in \subst{\regenv_2}{\breakconn{r}}$ then $\conn{s''}{r''}{\writeMode} \in  \subst{\regenv_2'}{\conn{s}{r}{\writeMode}}$ because $s'' \not = \compa{r}$;
\item $r$ is unassociated in $\subst{\regenv_2}{\breakconn{r}}$.
\end{itemize}
Case 2: there exists $\conn{r}{x}{\readMode} \in \regenv_2$ but there is not $\conn{r_a}{y}{\readMode}  \in \regenv_2$. We therefore distinguish two further subcases.\\
Case 2a:  $r_a$ is unassociated in  $\regenv_2$. Since $\regenv_2 \regcorresp \regenv_2'$ we know $\conn{s}{r}{\writeMode}  \in \regenv_2'$ and the following type derivation for $s_a = \compa{r_a}$ exists
$$
\Scale[0.94]{
\inference
{\inference
{s,s_a \in \regvar'  & \seclev(s)= \valHigh\\
\regenv_2', \{s\} \eenta{\regvar'}{\var'} r_a \etypeproduce \code{\assign{s_a}{r_a}}, \esecan{\valHigh}{1}, s_a,\subst{\regenv_2'}{\conn{s_a}{r_a}{\readMode}}}
{\regenv_2', \{\} \eenta{\regvar'}{\var'} \wop{r}{r_a} \etypeproduce 
\left \{
\begin{array}{l}
\assign{s_a}{r_a};\\
\assign{s}{\wop{s}{s_a}}
\end{array}
\right \}
, \esecan{\valHigh}{2}, s,\subst{\subst{\regenv_2'}{\breakconn{s
 }}}{\conn{s_a}{r_a}{\readMode}}
}}
{\regenv_2' \enta{\regvar'}{\var' }  \assign{r}{\wop{r}{r_a}}  \typeproduce 
\left \{
\begin{array}{l}
\assign{s_a}{r};\\
\assign{s}{\wop{s}{s_a}}; \\
\assign{r}{s}
\end{array}
\right \}
, \secan{\tconst~(1,3)}{\whigh}, 
\subst{\subst{\regenv_2'}{\conn{s}{r}{\writeMode}}}{\conn{s_a}{r_a}{\readMode}}}}
$$
We have that $\subst{\regenv_2}{\breakconn{r}} \regcorresp \subst{\subst{\regenv_2'}{\conn{s}{r}{\writeMode}}}{\conn{s_a}{r_a}{\readMode}}$:
\begin{itemize}
\item for any $r''$ such that $r'' \not = r$ and $r'' \not = r_a$, if $\conn{r''}{z}{\writeMode} \in \subst{\regenv_2}{\breakconn{r}}$ then $\conn{s''}{z}{\writeMode} \in \subst{\subst{\regenv_2'}{\conn{s}{r}{\writeMode}}}{\conn{s_a}{r_a}{\readMode}}$ because $s'' \not = \compa{r}$ and $s'' \not = \compa{r_a}$;
\item for any $r''$ such that $r'' \not = r$ and $r'' \not = r_a$, if $\conn{r''}{z}{\readMode} \in \subst{\regenv_2}{\breakconn{r}}$ then $\conn{s''}{r''}{\writeMode} \in \subst{\subst{\regenv_2'}{\conn{s}{r}{\writeMode}}}{\conn{s_a}{r_a}{\readMode}}$ because $s'' \not = \compa{r}$ and $s'' \not = \compa{r_a}$;
\item $r$ and $r_a$ are unassociated in $\subst{\regenv_2}{\breakconn{r}}$.
\end{itemize}
Case 2b:  there exists $\conn{r_a}{y}{\writeMode} \in \regenv_2$. Since $\regenv_2 \regcorresp \regenv_2'$ we know  $\conn{s}{r}{\writeMode}  \in \regenv_2'$ and $\conn{s_a}{y}{\writeMode}  \in \regenv_2'$ and the following type derivation exists for $s_a' = \shado{r_a}$
$$
\Scale[0.94]{\inference
{\inference
{s, s_a' \in \regvar'  & \seclev(s)= \valHigh \\
\regenv_2', \{s\} \eenta{\regvar'}{\var'} r_a \etypeproduce \code{\assign{s_a'}{r_a}}, \esecan{\valHigh}{1}, s_a',\subst{\regenv_2'}{\conn{s_a'}{r_a}{\readMode}}}
{\regenv_2', \{\} \eenta{\regvar'}{\var'} \wop{r}{r_a} \etypeproduce 
\left \{
\begin{array}{l}
\assign{s_a'}{r_a};  \\
\assign{s}{\wop{s}{s_a'}}
\end{array}
\right \}
, \esecan{\valHigh}{2}, s,\subst{\subst{\regenv_2'}{\breakconn{s
 }}}{\conn{s_a'}{r_a}{\readMode}}
}}
{\regenv_2' \enta{\regvar'}{\var'}  \assign{r}{\wop{r}{r_a}}  \typeproduce 
\left \{
\begin{array}{l}
\assign{s_a'}{r_a};\\
\assign{s}{\wop{s}{s_a'}}; \\
\assign{r}{s}
\end{array}
\right \}
, \secan{\tconst~(1,3)}{\whigh}, 
\subst{\subst{\regenv_2'}{\conn{s}{r}{\writeMode}}}{\conn{s_a'}{r_a}{\readMode}}}}
$$
We have that $\subst{\regenv_2}{\breakconn{r}} \regcorresp \subst{\subst{\regenv_2'}{\conn{s}{r}{\writeMode}}}{\conn{s_a'}{r_a}{\readMode}}$:
\begin{itemize}
\item for any $r''$ such that $r'' \not = r$ and $r'' \not = r_a$, if $\conn{r''}{z}{\writeMode} \in \subst{\regenv_2}{\breakconn{r}}$ then $\conn{s''}{z}{\writeMode} \in \subst{\subst{\regenv_2'}{\conn{s}{r}{\writeMode}}}{\conn{s_a'}{r_a}{\readMode}}$ because $s'' \not = \compa{r}$ and $s'' \not = \shado{r_a}$;
\item for any $r''$ such that $r'' \not = r$ and $r'' \not = r_a$, if $\conn{r''}{z}{\readMode} \in \subst{\regenv_2}{\breakconn{r}}$ then $\conn{s''}{r''}{\writeMode} \in \subst{\subst{\regenv_2'}{\conn{s}{r}{\writeMode}}}{\conn{s_a'}{r_a}{\readMode}}$ because $s'' \not = \compa{r}$ and $s'' \not = \shado{r_a}$;
\item $\conn{r_a}{y}{\writeMode} \in \subst{\regenv_2}{\breakconn{r}}$ and $\conn{s_a}{y}{\writeMode}  \in \subst{\subst{\regenv_2'}{\conn{s}{r}{\writeMode}}}{\conn{s_a'}{r_a}{\readMode}}$ because $s_a'= \shado{r_a}$;
\item $r$ is unassociated in $\subst{\regenv_2}{\breakconn{r}}$.
\end{itemize}
Case 3: there exists $\conn{r_a}{y}{\readMode} \in \regenv_2$ but there is not $\conn{r}{x}{\readMode}  \in \regenv_2$. We therefore distinguish two further subcases.\\
Case 3a:  $r$ is unassociated in $\regenv_2$. Since $\regenv_2 \regcorresp \regenv_2'$ we know $\conn{s_a}{r_a}{\writeMode}  \in \regenv_2'$ and the following type derivation for $s = \compa{r}$ exists
$$\inference
{\inference
{s,s_a \in \regvar'  & \seclev(s)= \valHigh\\
\regenv_2', \{\} \eenta{\regvar'}{\var'} r \etypeproduce \code{\assign{s}{r}}, \esecan{\valHigh}{1}, s,\subst{\regenv_2'}{\conn{s}{r}{\readMode}}}
{\regenv_2', \{\} \eenta{\regvar'}{\var'} \wop{r}{r_a} \etypeproduce \code{\assign{s}{r}; \assign{s}{\wop{s}{s_a}}}, \esecan{\valHigh}{2}, s,\subst{\regenv_2'}{\breakconn{s
 }}}
}
{\regenv_2' \enta{\regvar'}{\var'}  \assign{r}{\wop{r}{r_a}}  \typeproduce 
\left \{
\begin{array}{l}
\assign{s}{r};\\
\assign{s}{\wop{s}{s_a}}; \\
\assign{r}{s}
\end{array}
\right \}
, \secan{\tconst~(1,3)}{\whigh}, 
\subst{\regenv_2'}{\conn{s}{r}{\writeMode}}}$$
We have that $\subst{\regenv_2}{\breakconn{r}} \regcorresp \subst{\regenv_2'}{\conn{s}{r}{\writeMode}}$:
\begin{itemize}
\item for any $r''$ such that $r'' \not = r$, if $\conn{r''}{z}{\writeMode} \in \subst{\regenv_2}{\breakconn{r}}$ then $\conn{s''}{z}{\writeMode} \in \subst{\regenv_2'}{\conn{s}{r}{\writeMode}}$ because $s'' \not = \compa{r}$;
\item for any $r''$ such that $r'' \not = r$, if $\conn{r''}{z}{\readMode} \in \subst{\regenv_2}{\breakconn{r}}$ then $\conn{s''}{r''}{\writeMode} \in \subst{\regenv_2'}{\conn{s}{r}{\writeMode}}$ because $s'' \not = \compa{r}$;
\item $r$ is unassociated in $\subst{\regenv_2}{\breakconn{r}}$.
\end{itemize}
Case 3b:  there exists $\conn{r}{x}{\writeMode} \in \regenv_2$. Since $\regenv_2 \regcorresp \regenv_2'$ we know $\conn{s_a}{r_a}{\writeMode}  \in \regenv_2'$ and $\conn{s}{x}{\writeMode}  \in \regenv_2'$ so the following type derivation exists for $s'= \shado{r}$
$$\inference
{\inference
{s',s_a \in \regvar'  & \seclev(s')= \valHigh\\
\regenv_2', \{\} \eenta{\regvar'}{\var'} r \etypeproduce \code{\assign{s'}{r}}, \esecan{\valHigh}{1}, s',\subst{\regenv_2'}{\conn{s'}{r}{\readMode}}}
{\regenv_2', \{\} \eenta{\regvar'}{\var'} \wop{r}{r_a} \etypeproduce \code{\assign{s'}{r}; \assign{s'}{\wop{s'}{s_a}}}, \esecan{\valHigh}{2}, s',\subst{\regenv_2'}{\breakconn{s
 }}
}}
{\regenv_2' \enta{\regvar'}{\var' }  \assign{r}{\wop{r}{r_a}}  \typeproduce 
\left \{
\begin{array}{l}
\assign{s'}{r};\\
\assign{s'}{\wop{s'}{s_a}}; \\
\assign{r}{s'}
\end{array}
\right \}
, \secan{\tconst~(1,3)}{\whigh}, 
\subst{\regenv_2'}{\conn{s'}{r}{\writeMode}}}$$
We have that $\subst{\regenv_2}{\breakconn{r}} \regcorresp \subst{\regenv_2'}{\conn{s'}{r}{\writeMode}}$:
\begin{itemize}
\item for any $r''$ such that $r'' \not = r$, if $\conn{r''}{z}{\writeMode} \in \subst{\regenv_2}{\breakconn{r}}$ then $\conn{s''}{z}{\writeMode} \in \subst{\regenv_2'}{\conn{s'}{r}{\writeMode}}$ because $s'' \not = \shado{r}$;
\item for any $r''$ such that $r'' \not = r$, if $\conn{r''}{z}{\readMode} \in \subst{\regenv_2}{\breakconn{r}}$ then $\conn{s''}{r''}{\writeMode} \in \subst{\regenv_2'}{\conn{s'}{r}{\writeMode}}$ because $s'' \not = \shado{r}$;
\item $r$ is unassociated in $\subst{\regenv_2}{\breakconn{r}}$.
\end{itemize}
Case 4: neither $\conn{r}{x}{\readMode}$ nor $\conn{r_a}{y}{\readMode}$ are in $\regenv_2$. This case is provable by considering the combined cases 2 and 3.\\

Consider that the following derivation exists instead
$$\inference
{r,r_a \in \regvar \\
\regenv, \activeReg \eenta{\regvar}{\var} E_1 \etypeproduce \code{D_1}, \esecan{\valLow}{n_1}, r, \regenv_1 \\ 
\regenv_1, \activeReg\cup\{r\} \eenta{\regvar}{\var} E_2 \etypeproduce \code{D_2}, \esecan{\valLow}{n_2}, r_a, \regenv_2\\
\valLow=\seclev(r)=\seclev(r_a)}
{\regenv, \activeReg \eenta{\regvar}{\var} \wop{E_1}{E_2} \etypeproduce 
\left \{
\begin{array}{l}
D_1;\\
D_2;\\
\assign{r}{\wop{r}{r_a}}
\end{array}
\right \}, \esecan{\valLow}{n_1+n_2+1}, r, \subst{\regenv_2}{\breakconn{r}}}$$
and consider $\regenv'$ such that $\regenv \regcorresp \regenv'$. Then the following derivations exist by inductive hypothesis

$$\regenv' \enta{\regvar'}{\var'} D_1 \typeproduce \code{D_1'}, \secan{t_1}{w_1}, \regenv_1' \ \ \ \regenv_1' \enta{\regvar'}{\var'} D_2 \typeproduce \code{D_2'}, \secan{t_2}{w_2}, \regenv_2'$$
such that:
\begin{itemize}
\item if $n_i>0$ then $w_i = \wdontk$ and $t_i= \tdontc$;
\item if $n_i=0$ then $w_i = \whigh$ and $t_i= \tconst~(0,0)$;
\item $\regenv_2 \regcorresp \regenv_2'$.
\end{itemize}
In order for completing the case, we have to show that $\regenv_2' \enta{\regvar'}{\var'} \assign{r}{\wop{r}{r_a}} \typeproduce \code{D'}, \secan{\tdontc}{\wdontk}, \regenv_3$ such that $\subst{\regenv_2}{\breakconn{r}} \regcorresp \regenv_3$. In fact, recall that for all $w$ it holds that $w \wlub \wdontk = \wdontk$ and for any $t \sqsubseteq \tdontc$ it holds that $t \tclub \tdontc = \tdontc$.
We distinguish four cases. \\
 Case 1: there exist $\conn{r}{x}{\readMode}, \conn{r_a}{y}{\readMode}  \in \regenv_2$. Since $\regenv_2 \regcorresp \regenv_2'$, then we know $\conn{s}{r}{\writeMode}, \conn{s_a}{r_a}{\writeMode}  \in \regenv_2'$ and the following type derivation exists
$$\inference
{\inference
{s,s_a \in \regvar'  & \seclev(s)= \valLow}
{\regenv_2', \{\} \eenta{\regvar'}{\var'} \wop{r}{r_a} \etypeproduce \code{\assign{s}{\wop{s}{s_a}}}, \esecan{\valLow}{1}, s,\subst{\regenv_2'}{\breakconn{s}}
}}
{\regenv_2' \enta{\regvar'}{\var'}  \assign{r}{\wop{r}{r_a}}  \typeproduce \code{\assign{s}{\wop{s}{s_a}};\assign{r}{s}}, \secan{\tdontc}{\wdontk}, \subst{\regenv_2'}{\conn{s}{r}{\writeMode}}}$$
We have that $\subst{\regenv_2}{\breakconn{r}} \regcorresp \subst{\regenv_2'}{\conn{s}{r}{\writeMode}}$:
\begin{itemize}
\item for any $r'' \not = r$, if $\conn{r''}{z}{\writeMode} \in \subst{\regenv_2}{\breakconn{r}}$ then $\conn{s''}{z}{\writeMode} \in\subst{\regenv_2'}{\conn{s}{r}{\writeMode}}$ because $s'' \not = \compa{r}$;
\item for any $r'' \not = r$, if $\conn{r''}{z}{\readMode} \in \subst{\regenv_2}{\breakconn{r}}$ then $\conn{s''}{r''}{\writeMode} \in  \subst{\regenv_2'}{\conn{s}{r}{\writeMode}}$ because $s'' \not = \compa{r}$;
\item $r$ is unassociated in $\subst{\regenv_2}{\breakconn{r}}$.
\end{itemize}
Case 2: there exists $\conn{r}{x}{\readMode} \in \regenv_2$ but there is not $\conn{r_a}{y}{\readMode}  \in \regenv_2$. We therefore distinguish two further subcases.\\
Case 2a:  $r_a$ is unassociated in $\regenv_2$. Since $\regenv_2 \regcorresp \regenv_2'$ we know $\conn{s}{r}{\writeMode}  \in \regenv_2'$ and the following type derivation for $s_a = \compa{r_a}$ exists
$$\Scale[0.93]{\inference
{\inference
{s,s_a \in \regvar'  & \seclev(s)= \valLow\\
\regenv_2', \{s\} \eenta{\regvar'}{\var'} r_a \etypeproduce \code{\assign{s_a}{r_a}}, \esecan{\valLow}{1}, s_a,\subst{\regenv_2'}{\conn{s_a}{r_a}{\readMode}}}
{\regenv_2', \{\} \eenta{\regvar'}{\var'} \wop{r}{r_a} \etypeproduce 
\left \{
\begin{array}{l}
\assign{s_a}{r_a};\\
\assign{s}{\wop{s}{s_a}}
\end{array}
\right \}
, \esecan{\valLow}{2}, s,\subst{\subst{\regenv_2'}{\breakconn{s
 }}}{\conn{s_a}{r_a}{\readMode}}
}}
{\regenv_2' \enta{\regvar'}{\var' }  \assign{r}{\wop{r}{r_a}}  \typeproduce 
\left \{
\begin{array}{l}
\assign{s_a}{r};\\
\assign{s}{\wop{s}{s_a}}; \\
\assign{r}{s}
\end{array}
\right \}
, \secan{\tdontc}{\wdontk}, 
\subst{\subst{\regenv_2'}{\conn{s}{r}{\writeMode}}}{\conn{s_a}{r_a}{\readMode}}}}
$$
We have that $\subst{\regenv_2}{\breakconn{r}} \regcorresp \subst{\subst{\regenv_2'}{\conn{s}{r}{\writeMode}}}{\conn{s_a}{r_a}{\readMode}}$:
\begin{itemize}
\item for any $r''$ such that $r'' \not = r$ and $r'' \not = r_a$, if $\conn{r''}{z}{\writeMode} \in \subst{\regenv_2}{\breakconn{r}}$ then $\conn{s''}{z}{\writeMode} \in \subst{\subst{\regenv_2'}{\conn{s}{r}{\writeMode}}}{\conn{s_a}{r_a}{\readMode}}$ because $s'' \not = \compa{r}$ and $s'' \not = \compa{r_a}$;
\item for any $r''$ such that $r'' \not = r$ and $r'' \not = r_a$, if $\conn{r''}{z}{\readMode} \in \subst{\regenv_2}{\breakconn{r}}$ then $\conn{s''}{r''}{\writeMode} \in \subst{\subst{\regenv_2'}{\conn{s}{r}{\writeMode}}}{\conn{s_a}{r_a}{\readMode}}$ because $s'' \not = \compa{r}$ and $s'' \not = \compa{r_a}$;
\item $r$ and $r_a$ are unassociated in $\subst{\regenv_2}{\breakconn{r}}$.
\end{itemize}
Case 2b:  there exists $\conn{r_a}{y}{\writeMode} \in \regenv_2$. Since $\regenv_2 \regcorresp \regenv_2'$ we know  $\conn{s}{r}{\writeMode}  \in \regenv_2'$ and $\conn{s_a}{y}{\writeMode}  \in \regenv_2'$ and the following type derivation exists for $s_a' = \shado{r_a}$
$$\Scale[0.93]{\inference
{\inference
{s, s_a' \in \regvar'  & \seclev(s)= \valLow \\
\regenv_2', \{s\} \eenta{\regvar'}{\var'} r_a \etypeproduce \code{\assign{s_a'}{r_a}}, \esecan{\valLow}{1}, s_a',\subst{\regenv_2'}{\conn{s_a'}{r_a}{\readMode}}}
{\regenv_2', \{\} \eenta{\regvar'}{\var'} \wop{r}{r_a} \etypeproduce 
\left \{
\begin{array}{l}
\assign{s_a'}{r_a};\\
\assign{s}{\wop{s}{s_a'}}
\end{array}
\right \}
, \esecan{\valLow}{2}, s,\subst{\subst{\regenv_2'}{\breakconn{s
 }}}{\conn{s_a'}{r_a}{\readMode}}
}}
{\regenv_2' \enta{\regvar'}{\var'}  \assign{r}{\wop{r}{r_a}}  \typeproduce 
\left \{
\begin{array}{l}
\assign{s_a'}{r_a};\\
\assign{s}{\wop{s}{s_a'}}; \\
\assign{r}{s}
\end{array}
\right \}
, \secan{\tdontc}{\wdontk}, 
\subst{\subst{\regenv_2'}{\conn{s}{r}{\writeMode}}}{\conn{s_a'}{r_a}{\readMode}}}}
$$
We have that $\subst{\regenv_2}{\breakconn{r}} \regcorresp \subst{\subst{\regenv_2'}{\conn{s}{r}{\writeMode}}}{\conn{s_a'}{r_a}{\readMode}}$:
\begin{itemize}
\item for any $r''$ such that $r'' \not = r$ and $r'' \not = r_a$, if $\conn{r''}{z}{\writeMode} \in \subst{\regenv_2}{\breakconn{r}}$ then $\conn{s''}{z}{\writeMode} \in \subst{\subst{\regenv_2'}{\conn{s}{r}{\writeMode}}}{\conn{s_a'}{r_a}{\readMode}}$ because $s'' \not = \compa{r}$ and $s'' \not = \shado{r_a}$;
\item for any $r''$ such that $r'' \not = r$ and $r'' \not = r_a$, if $\conn{r''}{z}{\readMode} \in \subst{\regenv_2}{\breakconn{r}}$ then $\conn{s''}{r''}{\writeMode} \in \subst{\subst{\regenv_2'}{\conn{s}{r}{\writeMode}}}{\conn{s_a'}{r_a}{\readMode}}$ because $s'' \not = \compa{r}$ and $s'' \not = \shado{r_a}$;
\item $\conn{r_a}{y}{\writeMode} \in \subst{\regenv_2}{\breakconn{r}}$ and $\conn{s_a}{y}{\writeMode}  \in \subst{\subst{\regenv_2'}{\conn{s}{r}{\writeMode}}}{\conn{s_a'}{r_a}{\readMode}}$ because $s_a'= \shado{r_a}$;
\item $r$ is unassociated in $\subst{\regenv_2}{\breakconn{r}}$.
\end{itemize}
Case 3: there exists $\conn{r_a}{y}{\readMode} \in \regenv_2$ but there is not $\conn{r}{x}{\readMode}  \in \regenv_2$. We therefore distinguish two further subcases.\\
Case 3a:  $r$ is unassociated in $\regenv_2$. Since $\regenv_2 \regcorresp \regenv_2'$ we know $\conn{s_a}{r_a}{\writeMode}  \in \regenv_2'$ and the following type derivation for $s = \compa{r}$ exists
$$\inference
{\inference
{s,s_a \in \regvar'  & \seclev(s)= \valLow\\
\regenv_2', \{\} \eenta{\regvar'}{\var'} r \etypeproduce \code{\assign{s}{r}}, \esecan{\valLow}{1}, s,\subst{\regenv_2'}{\conn{s}{r}{\readMode}}}
{\regenv_2', \{\} \eenta{\regvar'}{\var'} \wop{r}{r_a} \etypeproduce \code{\assign{s}{r}; \assign{s}{\wop{s}{s_a}}}, \esecan{\valLow}{2}, s,\subst{\regenv_2'}{\breakconn{s
 }}}
}
{\regenv_2' \enta{\regvar'}{\var'}  \assign{r}{\wop{r}{r_a}}  \typeproduce 
\left \{
\begin{array}{l}
\assign{s}{r};\\
\assign{s}{\wop{s}{s_a}}; \\
\assign{r}{s}
\end{array}
\right \}
, \secan{\tdontc}{\wdontk}, 
\subst{\regenv_2'}{\conn{s}{r}{\writeMode}}}$$
We have that $\subst{\regenv_2}{\breakconn{r}} \regcorresp \subst{\regenv_2'}{\conn{s}{r}{\writeMode}}$:
\begin{itemize}
\item for any $r''$ such that $r'' \not = r$, if $\conn{r''}{z}{\writeMode} \in \subst{\regenv_2}{\breakconn{r}}$ then $\conn{s''}{z}{\writeMode} \in \subst{\regenv_2'}{\conn{s}{r}{\writeMode}}$ because $s'' \not = \compa{r}$;
\item for any $r''$ such that $r'' \not = r$, if $\conn{r''}{z}{\readMode} \in \subst{\regenv_2}{\breakconn{r}}$ then $\conn{s''}{r''}{\writeMode} \in \subst{\regenv_2'}{\conn{s}{r}{\writeMode}}$ because $s'' \not = \compa{r}$;
\item $r$ is unassociated in $\subst{\regenv_2}{\breakconn{r}}$.
\end{itemize}
Case 3b:  there exists $\conn{r}{x}{\writeMode} \in \regenv_2$. Since $\regenv_2 \regcorresp \regenv_2'$ we know $\conn{s_a}{r_a}{\writeMode}  \in \regenv_2'$ and $\conn{s}{x}{\writeMode}  \in \regenv_2'$ so the following type derivation exists for $s'= \shado{r}$
$$\inference
{\inference
{s',s_a \in \regvar'  & \seclev(s')= \valLow\\
\regenv_2', \{\} \eenta{\regvar'}{\var'} r \etypeproduce \code{\assign{s'}{r}}, \esecan{\valLow}{1}, s',\subst{\regenv_2'}{\conn{s'}{r}{\readMode}}}
{\regenv_2', \{\} \eenta{\regvar'}{\var'} \wop{r}{r_a} \etypeproduce \code{\assign{s'}{r}; \assign{s'}{\wop{s'}{s_a}}}, \esecan{\valLow}{2}, s',\subst{\regenv_2'}{\breakconn{s
 }}
}}
{\regenv_2' \enta{\regvar'}{\var' }  \assign{r}{\wop{r}{r_a}}  \typeproduce 
\left \{
\begin{array}{l}
\assign{s'}{r};\\
\assign{s'}{\wop{s'}{s_a}}; \\
\assign{r}{s'}
\end{array}
\right \}
, \secan{\tdontc}{\wdontk}, 
\subst{\regenv_2'}{\conn{s'}{r}{\writeMode}}}$$
We have that $\subst{\regenv_2}{\breakconn{r}} \regcorresp \subst{\regenv_2'}{\conn{s'}{r}{\writeMode}}$:
\begin{itemize}
\item for any $r''$ such that $r'' \not = r$, if $\conn{r''}{z}{\writeMode} \in \subst{\regenv_2}{\breakconn{r}}$ then $\conn{s''}{z}{\writeMode} \in \subst{\regenv_2'}{\conn{s'}{r}{\writeMode}}$ because $s'' \not = \shado{r}$;
\item for any $r''$ such that $r'' \not = r$, if $\conn{r''}{z}{\readMode} \in \subst{\regenv_2}{\breakconn{r}}$ then $\conn{s''}{r''}{\writeMode} \in \subst{\regenv_2'}{\conn{s'}{r}{\writeMode}}$ because $s'' \not = \shado{r}$;
\item $r$ is unassociated in $\subst{\regenv_2}{\breakconn{r}}$.
\end{itemize}
Case 4: neither $\conn{r}{x}{\readMode}$ nor $\conn{r_a}{y}{\readMode}$ are in $\regenv_2$. This case is provable by considering the combined cases 2 and 3.\\
 
 \end{proof}

We can now present the details for the recompilation of \swhprog\ programs. In order to relate the type annotations produced in the two compilations, we consider a relation $\nextTerm \subseteq \termlattice \times \termlattice $ such that $\ttop \nextTerm \ttop$, $\tdontc \nextTerm \tdontc$ and $\tconst~(m,n) \nextTerm \tconst~(n,p)$. 

\begin{proposition}[Type prediction]\label{thm:typeprediction}
Let $C$ be a \whprog\ program. If $\regenv \enta{\regvar}{\var}  C \typeproduce \code{D}, \secan{t}{w}, \regenv'$, then for any $\regenv_1$ such that $\regenv \regcorresp \regenv_1$ there exists a derivation $\regenv_1 \enta{\regvar'}{\var'} D \typeproduce D',\secan{t'}{w},\regenv_1'$ such that $t \nextTerm{} t'$ and $\regenv' \regcorresp \regenv_1'$. 
\end{proposition}
\begin{proof}

We prove the proposition by induction on the structure of $C$ and by cases on the last rule applied in the type derivation.

\framebox[1.1\width]{Base case} 

Case $C=\wskip$. Assume 
$$\inference
{ - }
{\regenv \enta{\regvar}{\var}  \wskip \typeproduce \code{\wskip}, \secan{\tconst~(1,1)}{\whigh}, \regenv}$$
and let $\regenv_1$ be a register record such that $\regenv \regcorresp \regenv_1$. Then the following derivation holds
$$\inference
{ - }
{\regenv_1 \enta{\regvar'}{\var'}  \wskip \typeproduce \code{\wskip}, \secan{\tconst~(1,1)}{\whigh}, \regenv_1}$$

Case $C= \assign{x}{E}$. Assume 
$$\inference
{x \in \var &r\in \regvar & \seclev(r)=\valHigh & \regenv, \emptyRegAll \eenta{\regvar}{\var} E \etypeproduce \code{D}, \esecan{\valHigh}{n}, r, \regenv' }
{
\regenv \enta{\regvar}{\var} \assign{x}{E} \typeproduce \code{D;\assign{x}{r}}, 
\secan{\tconst~(1,n+1)}{\whigh}, 
\subst{\regenv'}{\conn{r}{x}{\writeMode}}}$$
and let $\regenv_1$ be a register record such that $\regenv \regcorresp \regenv_1$. By applying Proposition \ref{ref:exprrety} on $E$ we have that $\regenv_1 \enta{\regvar'}{\var'} D \typeproduce \code{D'}, 
\secan{\tconst~(n,p)}{\whigh},\regenv_1'$ such that $\regenv' \regcorresp \regenv_1'$.We now distinguish two cases.\\
Case 1: there exists $\conn{r}{z}{\readMode} \in \regenv'$. Then since  $\regenv' \regcorresp \regenv_1'$ there exists $\conn{s}{r}{\writeMode} \in \regenv_1'$ and the following derivation is correct
$$\inference
{\regenv_1 \enta{\regvar'}{\var'} D \typeproduce \code{D'}, 
\secan{\tconst~(n,p)}{\whigh},\regenv_1' \\
\regenv_1' \enta{\regvar'}{\var'} \assign{x}{r} \typeproduce \code{\assign{x}{s}},\secan{\tconst~(1,1)}{\whigh}, \subst{\regenv_1'}{\conn{s}{x}{\writeMode}}}
{
\regenv_1 \enta{\regvar'}{\var'} D;\assign{x}{r} \typeproduce \code{D';\assign{x}{s}}, 
\secan{\tconst~(n+1,p+1)}{\whigh}, 
\subst{\regenv_1'}{\conn{s}{x}{\writeMode}}}$$
and $\subst{\regenv'}{\conn{r}{x}{\writeMode}} \regcorresp \subst{\regenv_1'}{\conn{s}{x}{\writeMode}}$.\\
Case 2: There is not $\conn{r}{z}{\readMode} \in \regenv'$. We distinguish two further cases.\\
Case 2a: $r$ is unassociated in $\regenv'$. Then
$$\inference
{\regenv_1 \enta{\regvar'}{\var'} D \typeproduce \code{D'}, 
\secan{\tconst~(n,p)}{\whigh},\regenv_1' \\
(ii)}
{
\regenv_1 \enta{\regvar'}{\var'} D;\assign{x}{r} \typeproduce 
\left \{
\begin{array}{l}
D'\\
\assign{s}{r};\\
\assign{x}{s}
\end{array}
\right \}, 
\secan{\tconst~(n+1,p+2)}{\whigh}, 
\subst{\regenv_1'}{\conn{s}{x}{\writeMode}}}$$
where 
$$
(ii)= \inference
{
s \in \regvar' & s = \compa{r} \\
\regenv_1',\{ \} \eenta{\regvar'}{\var'} r \typeproduce \code{\assign{s}{r}},\esecan{1}{\whigh}, s, \subst{\regenv_1'}{\conn{s}{r}{\readMode}}}
{
\regenv_1' \enta{\regvar'}{\var'} \assign{x}{r} \typeproduce \code{\assign{s}{r};\assign{x}{s}}, 
\secan{\tconst~(1,2)}{\whigh}, 
\subst{\regenv_1'}{\conn{s}{x}{\writeMode}}}$$
and $\subst{\regenv'}{\conn{r}{x}{\writeMode}} \regcorresp \subst{\regenv_1'}{\conn{s}{x}{\writeMode}}$.\\
Case 2b: $\conn{r}{z}{\writeMode} \in \regenv'$. Then since  $\regenv' \regcorresp \regenv_1'$ there exists $\conn{s}{z}{\writeMode} \in \regenv_1'$. Then the following derivation is correct
$$\inference
{\regenv_1 \enta{\regvar'}{\var'} D \typeproduce \code{D'}, 
\secan{\tconst~(n,p)}{\whigh},\regenv_1' \\
(ii)}
{
\regenv_1 \enta{\regvar'}{\var'} D;\assign{x}{r} \typeproduce 
\left \{
\begin{array}{l}
D';\\
\assign{s}{r};\\
 \assign{x}{s}
\end{array}
\right \} 
\secan{\tconst~(n+1,p+2)}{\whigh}, \subst{\regenv_1'}{\conn{s}{x}{\writeMode}}
}$$
where 
$$
(ii)= \inference
{s \in \regvar' & s = \compa{r}\\
\regenv_1',\{ \} \eenta{\regvar'}{\var'} r \typeproduce \code{\assign{s}{r}},\esecan{1}{\whigh}, s, \subst{\regenv_1'}{\conn{s}{r}{\readMode}}}
{
\regenv_1' \enta{\regvar'}{\var'} \assign{x}{r} \typeproduce \code{\assign{s}{r};\assign{x}{s}}, 
\secan{\tconst~(1,2)}{\whigh}, \subst{\regenv_1'}{\conn{s}{x}{\writeMode}}}$$
and $\subst{\regenv'}{\conn{r}{x}{\writeMode}} \regcorresp \subst{\regenv_1'}{\conn{s}{x}{\writeMode}}$.\\
Assume 
$$\inference
{x \in \var &r\in \regvar & \seclev(r) = \valLow & \regenv, \emptyRegAll \eenta{\regvar}{\var} E \etypeproduce \code{D}, \esecan{\valLow}{n}, r, \regenv' }
{
\regenv \enta{\regvar}{\var} \assign{x}{E} \typeproduce \code{D;\assign{x}{r}}, 
\secan{\tdontc}{\wdontk}, 
\subst{\regenv'}{\conn{r}{x}{\writeMode}}}$$ instead,
and let $\regenv_1$ be a register record such that $\regenv \regcorresp \regenv_1$. By applying Proposition \ref{ref:exprrety} on $E$ we have that $\regenv_1 \enta{\regvar'}{\var'} D \typeproduce \code{D'}, 
\secan{t}{w},\regenv_1'$ such that:
\begin{itemize}
\item if $n>0$ then $w = \wdontk$ and $t= \tdontc$;
\item if $n=0$ then $w = \whigh$ and $t= \tconst~(0,0)$;
\item $\regenv_2 \regcorresp \regenv_2'$.
\end{itemize}
We now distinguish two cases.\\
Case 1: there exists $\conn{r}{z}{\readMode} \in \regenv'$. Then since  $\regenv' \regcorresp \regenv_1'$ there exists $\conn{s}{r}{\writeMode} \in \regenv_1'$ and the following derivation is correct
$$\inference
{\regenv_1 \enta{\regvar'}{\var'} D \typeproduce \code{D'}, 
\secan{t}{w},\regenv_1' \\
\regenv_1' \enta{\regvar'}{\var'} \assign{x}{r} \typeproduce \code{\assign{x}{s}},\secan{\tdontc}{\wdontk}, \subst{\regenv_1'}{\conn{s}{x}{\writeMode}}}
{
\regenv_1 \enta{\regvar'}{\var'} D;\assign{x}{r} \typeproduce \code{D';\assign{x}{s}}, 
\secan{\tdontc}{\wdontk}, 
\subst{\regenv_1'}{\conn{s}{x}{\writeMode}}}$$
and $\subst{\regenv'}{\conn{r}{x}{\writeMode}} \regcorresp \subst{\regenv_1'}{\conn{s}{x}{\writeMode}}$.\\
Case 2: There is not $\conn{r}{z}{\readMode} \in \regenv'$. We distinguish two further cases.\\
Case 2a: $r$ is unassociated in $\regenv'$. Then
$$\inference
{\regenv_1 \enta{\regvar'}{\var'} D \typeproduce \code{D'}, 
\secan{t}{w},\regenv_1' \\
(ii)}
{
\regenv_1 \enta{\regvar'}{\var'} D;\assign{x}{r} \typeproduce \code{D';\assign{s}{r};\assign{x}{s}}, 
\secan{\tdontc}{\wdontk}, 
\subst{\regenv_1'}{\conn{s}{x}{\writeMode}}}$$
where 
$$
(ii)= \inference
{
s \in \regvar' & s = \compa{r} \\
\regenv_1',\{ \} \eenta{\regvar'}{\var'} r \typeproduce \code{\assign{s}{r}},\esecan{1}{\valLow},s, \subst{\regenv_1'}{\conn{s}{r}{\readMode}}}
{
\regenv_1' \enta{\regvar'}{\var'} \assign{x}{r} \typeproduce \code{\assign{s}{r};\assign{x}{s}}, 
\secan{\tdontc}{\wdontk}, 
\subst{\regenv_1'}{\conn{s}{x}{\writeMode}}}$$
and $\subst{\regenv'}{\conn{r}{x}{\writeMode}} \regcorresp \subst{\regenv_1'}{\conn{s}{x}{\writeMode}}$.\\
Case 2b: $\conn{r}{z}{\writeMode} \in \regenv'$. Then since  $\regenv' \regcorresp \regenv_1'$ there exists $\conn{s}{z}{\writeMode} \in \regenv_1'$. Then the following derivation is correct
$$\inference
{
\regenv_1 \enta{\regvar'}{\var'} D \typeproduce \code{D'}, 
\secan{t}{w},\regenv_1' \\
(ii)}
{
\regenv_1 \enta{\regvar'}{\var'} D;\assign{x}{r} \typeproduce \code{D';\assign{s}{r};\assign{x}{s}}, 
\secan{\tdontc}{\wdontk}, 
\subst{\regenv_1'}{\conn{s}{x}{\writeMode}}}$$
where 
$$
(ii)= \inference
{
s \in \regvar' & s = \compa{r}\\
\regenv_1',\{ \} \eenta{\regvar'}{\var'} r \typeproduce \code{\assign{s}{r}},\esecan{1}{\valLow}, s, \subst{\regenv_1'}{\conn{s}{r}{\readMode}}}
{
\regenv_1' \enta{\regvar'}{\var'} \assign{x}{r} \typeproduce \code{\assign{s}{r};\assign{x}{s}}, 
\secan{\tdontc}{\wdontk}, 
\subst{\regenv_1'}{\conn{s}{x}{\writeMode}}}
$$
and $\subst{\regenv'}{\conn{r}{x}{\writeMode}} \regcorresp \subst{\regenv_1'}{\conn{s}{x}{\writeMode}}$.

Case $C=  \out{ch}{E}$. 
Assume 
$$\inference
{r\in \regvar & \regenv, \emptyRegAll \eenta{\regvar}{\var} E \etypeproduce \code{D}, \esecan{\valHigh}{n}, r, \regenv' }
{
\regenv \enta{\regvar}{\var} \out{ch}{E} \typeproduce \code{D;\out{ch}{r}}, 
\secan{\tconst~(1,n+1)}{\whigh}, 
\regenv'}$$
and let $\regenv_1$ be a register record such that $\regenv \regcorresp \regenv_1$. By applying Proposition \ref{ref:exprrety} on $E$ we have that $\regenv_1 \enta{\regvar'}{\var'} D \typeproduce \code{D'}, 
\secan{\tconst~(n,p)}{\whigh},\regenv_1'$ such that $\regenv' \regcorresp \regenv_1'$.We now distinguish two cases.\\
Case 1: there exists $\conn{r}{z}{\readMode} \in \regenv'$. Then since  $\regenv' \regcorresp \regenv_1'$ there exists $\conn{s}{r}{\writeMode} \in \regenv_1'$ and the following derivation is correct
$$\inference
{\regenv_1 \enta{\regvar'}{\var'} D \typeproduce \code{D'}, 
\secan{\tconst~(n,p)}{\whigh},\regenv_1' \\
\regenv_1' \enta{\regvar'}{\var'} \out{ch}{r} \typeproduce \code{\out{ch}{s}},\secan{\tconst~(1,1)}{\whigh}, \regenv_1'}
{
\regenv_1 \enta{\regvar'}{\var'} D;\out{x}{r} \typeproduce \code{D';\out{ch}{s}}, 
\secan{\tconst~(n+1,p+1)}{\whigh}, \regenv_1'}$$
Case 2: There is not $\conn{r}{z}{\readMode} \in \regenv'$. We distinguish two further cases.\\
Case 2a: $r$ is unassociated in $\regenv'$. Then
$$\inference
{\regenv_1 \enta{\regvar'}{\var'} D \typeproduce \code{D'}, 
\secan{\tconst~(n,p)}{\whigh},\regenv_1' \\
(ii)}
{
\regenv_1 \enta{\regvar'}{\var'} D;\out{ch}{r} \typeproduce 
\left \{
\begin{array}{c}
D';\\
\assign{s}{r};\\
\out{ch}{s}
\end{array}
\right \}
, 
\secan{\tconst~(n+1,p+2)}{\whigh}, 
\subst{\regenv_1'}{\conn{s}{r}{\readMode}}}$$
where 
$$
(ii)= \inference
{
s \in \regvar' & s = \compa{r} \\
\regenv_1',\{ \} \eenta{\regvar'}{\var'} r \typeproduce \code{\assign{s}{r}},\esecan{1}{\whigh}, s, \subst{\regenv_1'}{\conn{s}{r}{\readMode}}}
{
\regenv_1' \enta{\regvar'}{\var'} \out{ch}{r} \typeproduce \code{\assign{s}{r};\out{ch}{s}}, 
\secan{\tconst~(1,2)}{\whigh}, 
\subst{\regenv_1'}{\conn{s}{r}{\readMode}}}$$
and $\regenv' \regcorresp \subst{\regenv_1'}{\conn{s}{r}{\readMode}}$ because $r$ is unassociated in $\regenv'$.\\
Case 2b: $\conn{r}{z}{\writeMode} \in \regenv'$. Then since  $\regenv' \regcorresp \regenv_1'$ there exists $\conn{s}{z}{\writeMode} \in \regenv_1'$. Consider $s' = \shado{r}$, then the following derivation is correct
$$\inference
{\regenv_1 \enta{\regvar'}{\var'} D \typeproduce \code{D'}, 
\secan{\tconst~(n,p)}{\whigh},\regenv_1' \\
(ii)}
{
\regenv_1 \enta{\regvar'}{\var'} D;\assign{x}{r} \typeproduce 
\left \{
\begin{array}{c}
D';\\
\assign{s'}{r};\\
\out{ch}{s'}
\end{array}
\right \}
, 
\secan{\tconst~(n+1,p+2)}{\whigh}, 
\subst{\regenv_1'}{\conn{s'}{r}{\readMode}}}$$
where 
$$
(ii)= \inference
{
s' \in \regvar' & s'= \shado{r} \\
\regenv_1',\{ \} \eenta{\regvar'}{\var'} r \typeproduce \code{\assign{s'}{r}},\esecan{1}{\whigh}, s', \subst{\regenv_1'}{\conn{s'}{r}{\readMode}}}
{
\regenv_1' \enta{\regvar'}{\var'} \assign{x}{r} \typeproduce \code{\assign{s'}{r};\out{ch}{s'}}, 
\secan{\tconst~(1,2)}{\whigh}, 
\subst{\regenv_1'}{\conn{s'}{r}{\readMode}}}
$$
and $\regenv' \regcorresp \subst{\regenv_1'}{\conn{s'}{r}{\readMode}}$.\\
Assume 
$$\inference
{r\in \regvar & \regenv, \emptyRegAll \eenta{\regvar}{\var} E \etypeproduce \code{D}, \esecan{\valLow}{n}, r, \regenv' }
{
\regenv \enta{\regvar}{\var} \out{ch}{E} \typeproduce \code{D;\out{ch}{r}}, 
\secan{\tdontc}{\wdontk}, 
\regenv'}$$ instead,
and let $\regenv_1$ be a register record such that $\regenv \regcorresp \regenv_1$. By applying Proposition \ref{ref:exprrety} on $E$ we have that $\regenv_1 \enta{\regvar'}{\var'} D \typeproduce \code{D'}, 
\secan{t}{w},\regenv_1'$ such that:
\begin{itemize}
\item if $n>0$ then $w = \wdontk$ and $t= \tdontc$;
\item if $n=0$ then $w = \whigh$ and $t= \tconst~(0,0)$;
\item $\regenv_2 \regcorresp \regenv_2'$.
\end{itemize}
We now distinguish two cases.\\
Case 1: there exists $\conn{r}{z}{\readMode} \in \regenv'$. Then since  $\regenv' \regcorresp \regenv_1'$ there exists $\conn{s}{r}{\writeMode} \in \regenv_1'$ and the following derivation is correct
$$\inference
{\regenv_1 \enta{\regvar'}{\var'} D \typeproduce \code{D'}, 
\secan{t}{w},\regenv_1' \\
\regenv_1' \enta{\regvar'}{\var'} \out{ch}{r} \typeproduce \code{\out{ch}{s}},\secan{\tdontc}{\wdontk}, \regenv_1'}
{
\regenv_1 \enta{\regvar'}{\var'} D;\out{ch}{r} \typeproduce \code{D';\out{ch}{s}}, 
\secan{\tdontc}{\wdontk}, 
\regenv_1'}$$
and $\regenv' \regcorresp \regenv_1'$.\\
Case 2: There is not $\conn{r}{z}{\readMode} \in \regenv'$. We distinguish two further cases.\\
Case 2a: $r$ is unassociated in $ \regenv'$. Then
$$\inference
{\regenv_1 \enta{\regvar'}{\var'} D \typeproduce \code{D'}, 
\secan{t}{w},\regenv_1' \\
(ii)}
{
\regenv_1 \enta{\regvar'}{\var'} D;\out{x}{r} \typeproduce \code{D';\assign{s}{r};\out{x}{s}}, 
\secan{\tdontc}{\wdontk}, 
\subst{\regenv_1'}{\conn{s}{r}{\readMode}}}$$
where 
$$
(ii)= \inference
{
s \in \regvar' & s  = \compa{r} \\
\regenv_1',\{ \} \eenta{\regvar'}{\var'} r \typeproduce \code{\assign{s}{r}},\esecan{1}{\valLow},s, \subst{\regenv_1'}{\conn{s}{r}{\readMode}}}
{
\regenv_1' \enta{\regvar'}{\var'} \out{ch}{r} \typeproduce \code{\assign{s}{r};\out{ch}{s}}, 
\secan{\tdontc}{\wdontk}, 
\subst{\regenv_1'}{\conn{s}{r}{\readMode}}}$$
and $\regenv' \regcorresp \subst{\regenv_1'}{\conn{s}{r}{\readMode}}$ because $r$ is unassociated in $\regenv'$.\\
Case 2b: $\conn{r}{z}{\writeMode} \in \regenv'$. Then since  $\regenv' \regcorresp \regenv_1'$ there exists $\conn{s}{z}{\writeMode} \in \regenv_1'$. Consider $s'= \shado{r}$, then the following derivation is correct
$$\inference
{\regenv_1 \enta{\regvar'}{\var'} D \typeproduce \code{D'}, 
\secan{t}{w},\regenv_1' \\
(ii)}
{
\regenv_1 \enta{\regvar'}{\var'} D;\out{ch}{r} \typeproduce \code{D';\assign{s'}{r};\out{ch}{s'}}, 
\secan{\tdontc}{\wdontk}, 
\subst{\regenv_1'}{\conn{s'}{r}{\readMode}}}$$
where 
$$
(ii)= \inference
{
s' \in \regvar' &s' = \shado{r} \\
\regenv_1',\{ \} \eenta{\regvar'}{\var'} r \typeproduce \code{\assign{s'}{r}},\esecan{1}{\valLow}, s', \subst{\regenv_1'}{\conn{s'}{r}{\readMode}}}
{
\regenv_1' \enta{\regvar'}{\var'} \out{ch}{r} \typeproduce \code{\assign{s'}{r};\out{ch}{s'}}, 
\secan{\tdontc}{\wdontk}, 
\subst{\regenv_1'}{\conn{s'}{r}{\readMode}}}
$$
and $\regenv' \regcorresp \subst{\regenv_1'}{\conn{s'}{r}{\readMode}}$.


\framebox[1.1\width]{Inductive Step} 

Case $C=C_1;C_2$. Assume
$$\inference
{\regenv \enta{\regvar}{\var}  C_1 \typeproduce \code{D_1}, \secan{t_1}{w_1}, \regenv_\alpha \\
\regenv_\alpha \enta{\regvar}{\var}  C_2 \typeproduce \code{D_2}, \secan{t_2}{w_2}, \regenv_\beta &
}
{\regenv \enta{\regvar}{\var}  C_1;C_2 \typeproduce \code{D_1;D_2}, \secan{t_1 \tclub t_2}{w_1 \wlub w_2},\regenv_\beta}\\
$$
and let $\regenv_1$ be a register record such that $\regenv \regcorresp \regenv_1$. Then the following type derivation exists
$$\inference
{\regenv_1 \enta{\regvar'}{\var'}  D_1 \typeproduce \code{D_1'}, \secan{t_1'}{w_1}, \regenv_\alpha' \\
\regenv_\alpha' \enta{\regvar'}{\var'}  D_2 \typeproduce \code{D_2'}, \secan{t_2'}{w_2}, \regenv_\beta' &
}
{\regenv_1 \enta{\regvar'}{\var'}  D_1;D_2 \typeproduce \code{D;D'}, \secan{t_1 \tclub t_2}{w_1 \wlub w_2},\regenv_\beta'}\\
$$
and $\regenv_\beta \regcorresp \regenv_\beta'$ by applying the inductive hypothesis on derivations for $D_1$ and $D_2$.

Case $C=\wif{E}{C_t}{C_e}$. Assume 
$$
\Scale[0.94]{\inference
{r \in \regvar \\
\regenv, \emptyRegAll \eenta{\regvar}{\var} E \etypeproduce \code{D_g}, \esecan{\valHigh}{n_g} , r, \regenv_\alpha\\ 
\regenv_\alpha \enta{\regvar}{\var}  C_t \typeproduce \code{D_t}, \secan{t_1}{\whigh}, \regenv_\beta &
\regenv_\alpha \enta{\regvar}{\var}  C_e \typeproduce \code{D_e}, \secan{t_2}{\whigh}, \regenv_\gamma  
}
{\typeconce{\regenv}{\wif{E}{C_t}{C_e}}{
\left \{ 
\begin{array}{l}
D_g; \\
\wifa{r}{D_t;\wskip}{D_e;\wskip}
\end{array} 
\right \}
} 
{\secan{\ttop}{\whigh}
}
{\regenv_\beta \envinters \regenv_\gamma}{\enta{\regvar}{\var}}}\\}
$$
and let $\regenv_1$ be a register record such that $\regenv \regcorresp \regenv_1$. By applying the Proposition \ref{ref:exprrety} on $E$ we have that $\regenv_1 \enta{\regvar'}{\var'}  D_g \typeproduce \code{D_g'}, \secan{\tconst~(n_g,p_g)}{\whigh}, \regenv_\alpha'$ such that $\regenv_\alpha \regcorresp \regenv_\alpha'$. We distinguish two cases.\\
Case 1: there exists $\conn{r}{z}{\readMode} \in \regenv_\alpha$. Then since  $\regenv_\alpha \regcorresp \regenv_\alpha'$ there exists $\conn{s}{r}{\writeMode} \in \regenv_\alpha'$ and the following derivation is correct
$$
\inference
{ 
\regenv_1 \enta{\regvar'}{\var'}  D_g \typeproduce \code{D_g'}, \secan{\tconst~(n_g,p_g)}{\whigh}, \regenv_\alpha' \\
s \in \regvar' & s = \compa{r} \\
\regenv_\alpha' \enta{\regvar'}{\var'}  D_t \typeproduce \code{R_t}, \secan{t_1'}{\whigh}, \regenv_\beta' &
\regenv_\alpha' \enta{\regvar'}{\var'}  D_e \typeproduce \code{R_e}, \secan{t_2'}{\whigh}, \regenv_\gamma' \\
R_t'=R_t;\wskip;\wskip & R_e'=R_e;\wskip;\wskip 
}
{\typeconce{\regenv'}{
\left \{ 
\begin{array}{l}
D_g;\\
\wifa{r}{D_t;\wskip}{D_e;\wskip}
\end{array} 
\right \}
}{
\left \{ 
\begin{array}{l}
D_g'; \\
\wifa{s}{R_t'}{R_e'}
\end{array} 
\right \}
} 
{
\secan{\ttop}{\whigh}}
{\regenv_\beta' \envinters \regenv_\gamma'}{\enta{\regvar'}{\var'}}}\\
$$
and $\regenv_\beta \envinters \regenv_\gamma \regcorresp \regenv_\beta' \envinters \regenv_\gamma'$ because the correspondence between records is preserved by record intersection.\\
Case 2: There is not $\conn{r}{z}{\readMode} \in \regenv_\alpha$. We distinguish two further cases.\\
Case 2a: $r$ is unassociated in $ \regenv_\alpha$.
Then the following derivation is correct
$$
\inference
{ 
\regenv_1 \enta{\regvar'}{\var'}  D_g \typeproduce \code{D_g'}, \secan{\tconst~(n_g,p_g)}{\whigh}, \regenv_\alpha' \\
s \in \regvar' & s = \compa{r}\\
D_a=\assign{s}{r}\\
\subst{\regenv_\alpha'}{\conn{s}{r}{\readMode}} \enta{\regvar'}{\var'}  D_t \typeproduce \code{R_t}, \secan{t_1'}{\whigh}, \regenv_\beta' \\
\subst{\regenv_\alpha'}{\conn{s}{r}{\readMode}} \enta{\regvar'}{\var'} D_e \typeproduce \code{R_e}, \secan{t_2'}{\whigh}, \regenv_\gamma' \\
R_t'=R_t;\wskip;\wskip & R_e'=R_e;\wskip;\wskip 
}
{\typeconce{\regenv'}{
\left \{ 
\begin{array}{l}
D_g;\\
\wifa{r}{D_t;\wskip}{D_e;\wskip}
\end{array} 
\right \}
}{
\left \{ 
\begin{array}{l}
D_g'; \\
D_a; \\
\wifa{s}{R_t'}{R_e'}
\end{array} 
\right \}
} 
{
\secan{\ttop}{\whigh}}
{\regenv_\beta' \envinters \regenv_\gamma'}{\enta{\regvar'}{\var'}}}\\
$$
In particular notice that $\regenv_\alpha \regcorresp \subst{\regenv_\alpha'}{\conn{s}{r}{\readMode}}$, since $r$ is unassociated in $\regenv_\alpha$ and $\regenv_\beta \envinters \regenv_\gamma \regcorresp \regenv_\beta' \envinters \regenv_\gamma'$ because the correspondence between records is preserved by record intersection.\\
Case 2b: $\conn{r}{z}{\writeMode} \in \regenv_\alpha$.
Then the following derivation is correct
$$
\inference
{ 
\regenv_1 \enta{\regvar'}{\var'}  D_g \typeproduce \code{D_g'}, \secan{\tconst~(n_g,p_g)}{\whigh}, \regenv_\alpha' \\
s' \in \regvar' & s'=\shado{r} \\
D_a=\assign{s'}{r} \\
\subst{\regenv_\alpha'}{\conn{s'}{r}{\readMode}} \enta{\regvar'}{\var'}  D_t \typeproduce \code{R_t}, \secan{t_1'}{\whigh}, \regenv_\beta' \\
\subst{\regenv_\alpha'}{\conn{s'}{r}{\readMode}} \enta{\regvar'}{\var'}  D_e \typeproduce \code{R_e}, \secan{t_2'}{\whigh}, \regenv_\gamma' \\
R_t'=R_t;\wskip;\wskip & R_e'=R_e;\wskip;\wskip 
}
{\typeconce{\regenv'}{
\left \{ 
\begin{array}{l}
D_g;\\
\wifa{r}{D_t;\wskip}{D_e;\wskip}
\end{array} 
\right \}
}{
\left \{ 
\begin{array}{l}
D_g'; \\
D_a;\\
\wifa{s'}{R_t'}{R_e'}
\end{array} 
\right \}
} 
{
\secan{\ttop}{\whigh}}
{\regenv_\beta' \envinters \regenv_\gamma'}{\enta{\regvar'}{\var'}}}\\
$$
In particular notice that $\regenv_\alpha \regcorresp \subst{\regenv_\alpha'}{\conn{s'}{r}{\readMode}}$, since $\regenv_\alpha \regcorresp \regenv_\alpha'$ and $\conn{r}{z}{\writeMode} \in \regenv_\alpha$ implies $\conn{s}{z}{\writeMode} \in \subst{\regenv_\alpha'}{\conn{s'}{r}{\readMode}}$ and $\regenv_\beta \envinters \regenv_\gamma \regcorresp \regenv_\beta' \envinters \regenv_\gamma'$ because the correspondence between records is preserved by record intersection.

Assume 
$$
\Scale[0.9]{
\inference
{r \in \regvar \\
\regenv, \emptyRegAll \eenta{\regvar}{\var} E \etypeproduce \code{D_g}, \esecan{\valLow}{n_g} , r, \regenv_\alpha\\ 
\regenv_\alpha \enta{\regvar}{\var}  C_t \typeproduce \code{D_t}, \secan{t_1}{w_1}, \regenv_\beta &
\regenv_\alpha \enta{\regvar}{\var}  C_e \typeproduce \code{D_e}, \secan{t_2}{w_2}, \regenv_\gamma  
}
{\typeconce{\regenv}{
\left \{ 
\begin{array}{l}
\wifa{E}{C_t}{C_e}
\end{array} 
\right \}
}{
\left \{ 
\begin{array}{l}
D_g; \\
\wifa{r}{D_t;\wskip}{D_e;\wskip}
\end{array} 
\right \}
} 
{\secan{\tdontc \tlub t_1 \tlub t_2}{\wdontk}
}
{\regenv_\beta \envinters \regenv_\gamma}{\enta{\regvar}{\var}}}\\}
$$
and let $\regenv_1$ be a register record such that $\regenv \regcorresp \regenv_1$. By applying the Proposition \ref{ref:exprrety} on $E$ we have that $\regenv_1 \enta{\regvar'}{\var'}  D_g \typeproduce \code{D_g'}, \secan{t}{w}, \regenv_\alpha'$ such that:
\begin{itemize}
\item if $n>0$ then $w = \wdontk$ and $t= \tdontc$;
\item if $n=0$ then $w = \whigh$ and $t= \tconst~(0,0)$;
\item $\regenv_\alpha \regcorresp \regenv_\alpha'$. 
\end{itemize}
We distinguish two cases.\\
Case 1: there exists $\conn{r}{z}{\readMode} \in \regenv_\alpha$. Then since  $\regenv_\alpha \regcorresp \regenv_\alpha'$ there exists $\conn{s}{r}{\writeMode} \in \regenv_\alpha'$ and the following derivation is correct
$$
\Scale[0.9]{
\inference
{ 
\regenv_1 \enta{\regvar'}{\var'}  D_g \typeproduce \code{D_g'}, \secan{t}{w}, \regenv_\alpha' \\
s \in \regvar' & s= \compa{r}\\
\regenv_\alpha' \enta{\regvar'}{\var'}  D_t \typeproduce \code{R_t}, \secan{t_1'}{w_1}, \regenv_\beta' &
\regenv_\alpha' \enta{\regvar'}{\var'}  D_e \typeproduce \code{R_e}, \secan{t_2'}{w_2}, \regenv_\gamma' \\
R_t'=R_t;\wskip;\wskip & R_e'=R_e;\wskip;\wskip 
}
{\typeconce{\regenv'}{
\left \{ 
\begin{array}{l}
D_g;\\
\wifa{r}{D_t;\wskip}{D_e;\wskip}
\end{array} 
\right \}
}
{
\left \{ 
\begin{array}{l}
D_g'; \\
\wifa{s}{R_t'}{R_e'}
\end{array} 
\right \}
} 
{
\secan{\tdontc \tlub t_1' \tlub t_2'}{\wdontk}}
{\regenv_\beta' \envinters \regenv_\gamma'}{\enta{\regvar'}{\var'}}}\\}
$$
and $\regenv_\beta \envinters \regenv_\gamma \regcorresp \regenv_\beta' \envinters \regenv_\gamma'$ because the correspondence between records is preserved by record intersection.\\
Case 2: There is not $\conn{r}{z}{\readMode} \in \regenv_\alpha$. We distinguish two further cases.\\
Case 2a: $r$ is unassociated in $\regenv_\alpha$.
Then the following derivation is correct
$$
\Scale[0.9]{
\inference
{ 
\regenv_1 \enta{\regvar'}{\var'}  D_g \typeproduce \code{D_g'}, \secan{t}{w}, \regenv_\alpha' \\
s \in \regvar' & s= \compa{r} \\
D_a=\assign{s}{r}\\
\subst{\regenv_\alpha'}{\conn{s}{r}{\readMode}} \enta{\regvar'}{\var'}  D_t \typeproduce \code{R_t}, \secan{t_1'}{w_1}, \regenv_\beta' \\
\subst{\regenv_\alpha'}{\conn{s}{r}{\readMode}}  \enta{\regvar'}{\var'} D_e \typeproduce \code{R_e}, \secan{t_2'}{w_2}, \regenv_\gamma' \\
R_t'=R_t;\wskip;\wskip & R_e'=R_e;\wskip;\wskip 
}
{\typeconce{\regenv'}{
\left \{ 
\begin{array}{l}
D_g;\\
\wifa{r}{D_t;\wskip}{D_e;\wskip}
\end{array} 
\right \}
}{
\left \{ 
\begin{array}{l}
D_g'; \\
D_a; \\
\wifa{s}{R_t'}{R_e'}
\end{array} 
\right \}
} 
{
\secan{\tdontc \tlub t_1' \tlub t_2'}{\wdontk}}
{\regenv_\beta' \envinters \regenv_\gamma'}{\enta{\regvar'}{\var'}}}\\}
$$
In particular notice that $\regenv_\alpha \regcorresp \subst{\regenv_\alpha'}{\conn{s}{r}{\readMode}}$, since $r$ is unassociated in $\regenv_\alpha$ and $\regenv_\beta \envinters \regenv_\gamma \regcorresp \regenv_\beta' \envinters \regenv_\gamma'$ because the correspondence between records is preserved by record intersection.\\
Case 2b: $\conn{r}{z}{\writeMode} \in \regenv_\alpha$.
Then the following derivation is correct
$$
\Scale[0.9]{
\inference
{ 
\regenv_1 \enta{\regvar'}{\var'}  D_g \typeproduce \code{D_g'}, \secan{t}{w}, \regenv_\alpha' \\
s' \in \regvar' & s'= \shado{r} \\
D_a=\assign{s'}{r} \\
\subst{\regenv_\alpha'}{\conn{s'}{r}{\readMode}} \enta{\regvar'}{\var'}  D_t \typeproduce \code{R_t}, \secan{t_1'}{w_1}, \regenv_\beta' \\
\subst{\regenv_\alpha'}{\conn{s'}{r}{\readMode}} \enta{\regvar'}{\var'}  D_e \typeproduce \code{R_e}, \secan{t_2'}{w_2}, \regenv_\gamma' \\
R_t'=R_t;\wskip;\wskip & R_e'=R_e;\wskip;\wskip 
}
{\typeconce{\regenv'}{
\left \{ 
\begin{array}{l}
D_g;\\
\wifa{r}{D_t;\wskip}{D_e;\wskip}
\end{array} 
\right \}
}{
\left \{ 
\begin{array}{l}
D_g'; \\
D_a;\\
\wifa{s'}{R_t'}{R_e'}
\end{array} 
\right \}
} 
{
\secan{\tdontc \tlub t_1 \tlub t_2}{\wdontk}}
{\regenv_\beta' \envinters \regenv_\gamma'}{\enta{\regvar'}{\var'}}}\\}
$$
In particular notice that $\regenv_\alpha \regcorresp \subst{\regenv_\alpha'}{\conn{s'}{r}{\readMode}}$, since $\regenv_\alpha \regcorresp \regenv_\alpha'$ and $\conn{r}{z}{\writeMode} \in \regenv_\alpha$ implies $\conn{s}{z}{\writeMode} \in \subst{\regenv_\alpha'}{\conn{s'}{r}{\readMode}}$ and $\regenv_\beta \envinters \regenv_\gamma \regcorresp \regenv_\beta' \envinters \regenv_\gamma'$ because the correspondence between records is preserved by record intersection.

Assume 
$$
\inference
{r \in \regvar \\
\regenv, \emptyRegAll \eenta{\regvar}{\var} E \etypeproduce \code{D_g}, \esecan{\valHigh}{n_g} , r, \regenv_\alpha\\ 
\regenv_\alpha \enta{\regvar}{\var}  C_t \typeproduce \code{D_t}, \secan{\tconst~(m,n_t)}{\whigh}, \regenv_\beta \\
\regenv_\alpha \enta{\regvar}{\var}  C_e \typeproduce \code{D_e}, \secan{\tconst~(m,n_e)}{\whigh}, \regenv_\gamma\\ 
D_t'= D_t;\wskip^{n_e-n_t};\wskip  &  D_e'= D_e;\wskip^{n_t-n_e};\wskip\\
n= \max(n_t,n_e) \\
}
{\typeconce{\regenv}{\wif{E}{C_t}{C_e}}{
\left \{ 
\begin{array}{l}
D_g; \\
\wifa{r}{D_t'}{D_e'}
\end{array} 
\right \}
} 
{\secan{\tau}{\whigh}
}
{\regenv_\beta \envinters \regenv_\gamma}{\enta{\regvar}{\var}}}\\
$$
for $\tau= \tconst~(m+1,n_g+n+2)$ and let $\regenv_1$ be a register record such that $\regenv \regcorresp \regenv_1$. By applying the Proposition \ref{ref:exprrety} on $E$ we have that $\regenv_1 \enta{\regvar'}{\var'}  D_g \typeproduce \code{D_g'}, \secan{\tconst~(n_g,p_g)}{\whigh}, \regenv_\alpha'$ such that $\regenv_\alpha \regcorresp \regenv_\alpha'$. We distinguish two cases.\\
Case 1: there exists $\conn{r}{z}{\readMode} \in \regenv_\alpha$. Then since  $\regenv_\alpha \regcorresp \regenv_\alpha'$ there exists $\conn{s}{r}{\writeMode} \in \regenv_\alpha'$ and the following derivation is correct
$$
\inference
{ 
\regenv_1 \enta{\regvar'}{\var'}  D_g \typeproduce \code{D_g'}, \secan{\tconst~(n_g,p_g)}{\whigh}, \regenv_\alpha' \\
s \in \regvar' & s = \compa{r}\\
\regenv_\alpha' \enta{\regvar'}{\var'}  D_t' \typeproduce \code{R_t}, \secan{\tconst~(n+1,p_t)}{\whigh}, \regenv_\beta' \\
\regenv_\alpha' \enta{\regvar'}{\var'}  D_e' \typeproduce \code{R_e}, \secan{\tconst~(n+1,p_e)}{\whigh}, \regenv_\gamma' \\
R_t'=R_t;\wskip^{p_e-p_t};\wskip & R_e'=R_e;\wskip^{p_t-p_e};\wskip  \\
p = \max(p_t,p_e)
}
{\typeconce{\regenv'}{
\left \{ 
\begin{array}{l}
D_g;\\
\wifa{r}{D_t'}{D_e'}
\end{array} 
\right \}
}{
\left \{ 
\begin{array}{l}
D_g'; \\
\wifa{s}{R_t'}{R_e'}
\end{array} 
\right \}
} 
{
\secan{\tau'}{\whigh}}
{\regenv_\beta' \envinters \regenv_\gamma'}{\enta{\regvar'}{\var'}}}\\
$$
Notice that $\tau' = \tconst~(n_g,p_g) \tclub \tconst~(n+2,p+2)= \tconst~(n_g+n+2,p_g+p+2)$ and $\regenv_\beta \envinters \regenv_\gamma \regcorresp \regenv_\beta' \envinters \regenv_\gamma'$ because the correspondence between records is preserved by record intersection.\\
Case 2: There is not $\conn{r}{z}{\readMode} \in \regenv_\alpha$. We distinguish two further cases.\\
Case 2a: $r$ is unassociated in $\regenv_\alpha$.
Then the following derivation is correct
$$
\inference
{ 
\regenv_1 \enta{\regvar'}{\var'}  D_g \typeproduce \code{D_g'}, \secan{\tconst~(n_g,p_g)}{\whigh}, \regenv_\alpha' \\
s \in \regvar' & s = \compa{r}\\
D_a=\assign{s}{r}\\ 
\subst{\regenv_\alpha'}{\conn{s}{r}{\readMode}} \enta{\regvar'}{\var'}  D_t \typeproduce \code{R_t}, \secan{\tconst~(n+1,p_t)}{\whigh}, \regenv_\beta' \\
\subst{\regenv_\alpha'}{\conn{s}{r}{\readMode}}  \enta{\regvar'}{\var'}D_e \typeproduce \code{R_e}, \secan{\tconst~(n+1,p_e)}{\whigh}, \regenv_\gamma' \\
R_t'=R_t;\wskip^{p_e-p_t};\wskip & R_e'=R_e;\wskip^{p_t-p_e};\wskip  \\
p= \max(p_t,p_e)
}
{\typeconce{\regenv'}{
\left \{ 
\begin{array}{l}
D_g;\\
\wifa{r}{D_t'}{D_e'}
\end{array} 
\right \}
}{
\left \{ 
\begin{array}{l}
D_g'; \\
D_a; \\
\wifa{s}{R_t'}{R_e'}
\end{array} 
\right \}
} 
{
\secan{\tau'}{\whigh}}
{\regenv_\beta' \envinters \regenv_\gamma'}{\enta{\regvar'}{\var'}}}\\
$$
for $\tau' = \tconst~(n_g,p_g)\tclub \tconst~(n+2,1+p+2)$. In particular notice that $\tconst~(n_g,p_g) \tclub \tconst~(n+2,1+p+2)= \tconst~(n_g+n+2,1+p_g+p+2)$ and $\regenv_\alpha \regcorresp \subst{\regenv_\alpha'}{\conn{s}{r}{\readMode}}$, since $r$ is unassociated in $\regenv_\alpha$ and $\regenv_\beta \envinters \regenv_\gamma \regcorresp \regenv_\beta' \envinters \regenv_\gamma'$ because the correspondence between records is preserved by record intersection.\\
Case 2b: $\conn{r}{z}{\writeMode} \in \regenv_\alpha$.
Then the following derivation is correct
$$
\inference
{ 
\regenv_1 \enta{\regvar'}{\var'}  D_g \typeproduce \code{D_g'}, \secan{\tconst~(n_g,p_g)}{\whigh}, \regenv_\alpha' \\
s' \in \regvar' & s'= \shado{r} \\
D_a=\assign{s'}{r}\\ 
\subst{\regenv_\alpha'}{\conn{s'}{r}{\readMode}} \enta{\regvar'}{\var'}  D_t \typeproduce \code{R_t}, \secan{\tconst~(n+1,p_t)}{\whigh}, \regenv_\beta' \\
\subst{\regenv_\alpha'}{\conn{s'}{r}{\readMode}} \enta{\regvar'}{\var'}  D_e \typeproduce \code{R_e}, \secan{\tconst~(n+1,p_e)}{\whigh}, \regenv_\gamma' \\
R_t'=R_t;\wskip^{p_e-p_t};\wskip & R_e'=R_e;\wskip^{p_t-p_e};\wskip  \\
p= \max(p_t,p_e)
}
{\typeconce{\regenv'}{
\left \{ 
\begin{array}{l}
D_g;\\
\wifa{r}{D_t'}{D_e'}
\end{array} 
\right \}
}{
\left \{ 
\begin{array}{l}
D_g'; \\
D_a; \\
\wifa{s'}{R_t'}{R_e'}
\end{array} 
\right \}
} 
{
\secan{\tau'}{\whigh}}
{\regenv_\beta' \envinters \regenv_\gamma'}{\enta{\regvar'}{\var'}}}\\
$$
for $\tau'= \tconst~(n_g,p_g)\tclub \tconst~(n+2,1+p+2)$. In particular notice that $\tconst~(n_g,p_g) \tclub \tconst~(n+2,1+p+2)= \tconst~(n_g+n+2,1+p_g+p+2)$ and $\regenv_\alpha \regcorresp \subst{\regenv_\alpha'}{\conn{s'}{r}{\readMode}}$, since $\regenv_\alpha \regcorresp \regenv_\alpha'$ and $\conn{r}{z}{\writeMode} \in \regenv_\alpha$ implies $\conn{s}{z}{\writeMode} \in \subst{\regenv_\alpha'}{\conn{s'}{r}{\readMode}}$ and $\regenv_\beta \envinters \regenv_\gamma \regcorresp \regenv_\beta' \envinters \regenv_\gamma'$ because the correspondence between records is preserved by record intersection.

Case $C=\while{x}{C_1}$. Assume $\valHigh = \seclev(x)=\seclev(r)$ and
$$\inference
{x \in \var & r \in \regvar \\
A=\assign{r}{x};\assign{x}{r}; \\
\regenv_\alpha=\subst{\regenv}{\conn{r}{x}{\writeMode}} \\
\regenv_B \sqsubseteq \regenv_\alpha & \regenv_B \sqsubseteq \subst{\regenv_E}{\conn{r}{x}{\writeMode}}\\
\regenv_B \enta{\regvar}{\var}  C \typeproduce \code{D}, \secan{t}{\whigh} , \regenv_E\\
}
{\typeconce{\regenv}{\while{x}{C}}{
\left \{ 
\begin{array}{l}
A; \\
\while{r}{\code{D;A;\wskip}}\\
\end{array} 
\right \}
} 
{
\secan{\ttop}{\whigh}
}
{\regenv_B}{\enta{\regvar}{\var}}}$$
and let $\regenv_1$ be a register record such that $\regenv \regcorresp \regenv_1$. We now have to show that 
$$\typeconce{\regenv_1}{
\left \{ 
\begin{array}{l}
A;\\
\while{r}{\code{D;A;\wskip}}
\end{array} 
\right \}
}
{
\code{\dots}
} 
{
\secan{\ttop}{\whigh}
}
{\regenv_B'}{\enta{\regvar}{\var}}$$
and $\regenv_B \regcorresp \regenv_B'$.
We begin by constructing the derivation for $A$ as follows 
$$\inference{
(i)
&
(ii)
}
{\regenv_1 \enta{\regvar'}{\var'} \assign{r}{x};\assign{x}{r};  \typeproduce 
\left \{ 
\begin{array}{l}
\assign{s}{x}; \\
\assign{r}{s};\\
\assign{x}{s};
\end{array} 
\right \}
, \secan{\tconst~(2,3)}{\whigh} ,\subst{\regenv_1}{\conn{s}{x}{\writeMode}}}$$ 
where 
$$ 
\begin{array}{rcl}
(i) & = &
\inference{s \in \regvar' & s = \compa{r} \\
\regenv_1, \{\} \eenta{\regvar'}{\var} x \typeproduce \code{\assign{s}{x}},\esecan{\valHigh}{1},s,\subst{\regenv_1}{\conn{s}{x}{\readMode}}}
{\regenv_1 \enta{\regvar'}{\var'} \assign{r}{x} \typeproduce 
\left \{ 
\begin{array}{l}
\assign{s}{x}; \\
\assign{r}{s};
\end{array} 
\right \}, \secan{\tconst~(1,2)}{\whigh}, \subst{\regenv_1}{\conn{s}{r}{\writeMode}}} \\
(ii) & = & 
\inference{\subst{\regenv_1}{\conn{s}{r}{\writeMode}},\{\} \eenta{\regvar'}{\var'} r \typeproduce \code{\findot},\esecan{\valHigh}{0},s,\subst{\regenv_1}{\conn{s}{r}{\writeMode}}}
{\subst{\regenv_1}{\conn{s}{r}{\writeMode}} \enta{\regvar'}{\var'} \assign{x}{r} \typeproduce \code{\assign{x}{s}}, \secan{\tconst~(1,1)}{\whigh}, \subst{\regenv_1}{\conn{s}{x}{\writeMode}}}
\end{array}$$
Notice that $\subst{\regenv}{\conn{r}{x}{\writeMode}} \regcorresp \subst{\regenv_1}{\conn{s}{x}{\writeMode}}$. We now need to show that
$$\subst{\regenv_1}{\conn{s}{x}{\writeMode}} \enta{\regvar'}{\var'} \while{r}{\code{D;A;\wskip}} \typeproduce \code{\dots},\secan{\ttop}{\whigh},\regenv_B'$$
Consider the following derivation
$$\inference{
r \in \var'  & s' \in \regvar' & s' = \shado{r}\\
A' = \assign{s'}{r};\assign{r}{s'}; \\
\regenv_1'=\subst{\subst{\regenv_1}{\conn{s}{x}{\writeMode}}}{\conn{s'}{r}{\writeMode}} 
& \regenv_B' \sqsubseteq \regenv_1'
& \regenv_B' \sqsubseteq \subst{\subst{\regenv_E'}{\conn{s}{x}{\writeMode}}}{\conn{s'}{r}{\writeMode}}\\
(iii)\\
}
{\subst{\regenv_1}{\conn{s}{x}{\writeMode}} \enta{\regvar'}{\var'} 
{
\left \{
\begin{array}{l}
\while{r}\{\\
\begin{array}{ll}
 & D;\\
 & \assign{r}{x};\\
 & \assign{x}{r};\\
 & \wskip; \}
\end{array}\\
\end{array}
\right \}
} 
\typeproduce 
\left \{
\begin{array}{l}
A';\\
\while{s'}\{\\
\begin{array}{ll}
& D';\\
& \assign{s}{x};\\
& \assign{r}{s};\\
& \assign{x}{s};\\
& \wskip;\\
& A';\\
& \wskip; \}
\end{array}\\
\end{array}
\right \}
,\secan{\ttop}{\whigh},\regenv_B'}
$$
Notice that $\subst{\regenv}{\conn{r}{x}{\writeMode}} \regcorresp  \subst{\subst{\regenv_1}{\conn{s}{x}{\writeMode}}}{\conn{s'}{r}{\writeMode}} = \regenv_1'$, since $r$ is already in correspondence with $s$. Consider the smallest $\regenv_B' \sqsubseteq \regenv_1'$ such that $\regenv_B \regcorresp \regenv_B'$ (the domain of $\regenv_B'$ is $\compa{\Dom(\regenv_B)}$). Then the following derivation
$$
\Scale[0.93]{
(iii)= \inference{(iv) & (v) & (vi)}
{\regenv_B' \enta{\regvar'}{\var'} 
\left \{
\begin{array}{l}
D;\\
\assign{r}{x};\\
\assign{x}{r};\\
\wskip; 
\end{array}
\right \}
\typeproduce 
\left \{
\begin{array}{l}
D';\\
\assign{s}{x};\\
\assign{r}{s};\\
\assign{x}{s};\\
\wskip; 
\end{array}
\right \}
,
\left \langle
\begin{array}{l}
\whigh, \\
t'\tclub \tconst~(2,3) \tclub \tconst~(1,1)
\end{array}
\right \rangle
,\regenv_E'}}
$$
holds for
$$
\begin{array}{rl}
(iv) = & \regenv_B' \enta{\regvar'}{\var'} D \typeproduce \code{D'},\secan{t'}{\whigh},\regenv_E'\\
(v) = & \regenv_E' \enta{\regvar'}{\var'} 
\left \{
\begin{array}{l}
\assign{r}{x};\\
\assign{x}{r};
\end{array}
\right \} \typeproduce 
\left \{
\begin{array}{l}
\assign{s}{x};\\
\assign{r}{s};\\
\assign{x}{s};\\
\end{array}
\right \}
,\secan{\tconst~(1,3)}{\whigh},\subst{\regenv_E'}{\conn{s}{x}{\writeMode}}\\
(vi) =  & \subst{\regenv_E'}{\conn{s}{x}{\writeMode}} \enta{\regvar'}{\var'} \wskip \typeproduce \code{\wskip},\secan{\tconst~(1,1)}{\whigh},\subst{\regenv_E'}{\conn{s}{x}{\writeMode}}\\
\end{array}
$$
In particular:
\begin{itemize}
\item the type derivation $(iv)$ follows from the inductive hypothesis on $D$, which also guarantees that $\regenv_E \regcorresp \regenv_E'$; 
\item type derivation $(v)$ is identical, up to environment renaming, to the same derivation calculated previously. It also holds that $\subst{\regenv_E}{\conn{s}{x}{\writeMode}} \regcorresp \subst{\regenv_E'}{\conn{s}{x}{\writeMode}}$, since $\regenv_E \regcorresp \regenv_E'$;
\item type derivation $(vi)$ is straightforward.
\end{itemize}
The case is concluded by observing that $\regenv_B \regcorresp \regenv_B'$ and $\regenv_B \sqsubseteq \subst{\regenv_E}{\conn{s}{x}{\writeMode}}$ and $\subst{\regenv_E}{\conn{s}{x}{\writeMode}} \regcorresp \subst{\regenv_E'}{\conn{s}{x}{\writeMode}}$ implies that $\regenv_B' \sqsubseteq \subst{\subst{\regenv_E'}{\conn{s}{x}{\writeMode}}}{\conn{s'}{r}{\writeMode}}$.

Assume $\valLow = \seclev(x)=\seclev(r)$ and 
$$\inference
{x \in \var & r \in \regvar \\
A=\assign{r}{x};\assign{x}{r}; \\
\regenv_\alpha=\subst{\regenv}{\conn{r}{x}{\writeMode}} \\
\regenv_B \sqsubseteq \regenv_\alpha & \regenv_B \sqsubseteq \subst{\regenv_E}{\conn{r}{x}{\writeMode}}\\
\regenv_B \enta{\regvar}{\var}  C \typeproduce \code{D}, \secan{t}{w} , \regenv_E\\}
{\typeconce{\regenv}{\while{x}{C}}{
\left \{ 
\begin{array}{l}
A; \\
\while{r}{\code{D;A;\wskip}}\\
\end{array} 
\right \}
} 
{
\secan{\tdontc}{\wdontk}
}
{\regenv_B}{\enta{\regvar}{\var}}}$$
instead. Notice that for $\seclev(x)= \seclev(r)= \valLow$ we have that $t$ must be $t\sqsubset \ttop$, hence $\termMap(\lambda)\tlub t= \tdontc$. Let $\regenv_1$ be a register record such that $\regenv \regcorresp \regenv_1$. We now have to show that 
$$\typeconce{\regenv_1}{
\left \{ 
\begin{array}{l}
A;\\
\while{r}{\code{D;A;\wskip}}
\end{array} 
\right \}
}
{
\code{\dots}
} 
{
\secan{\tdontc}{\wdontk}
}
{\regenv_B'}{\enta{\regvar}{\var}}$$
and $\regenv_B \regcorresp \regenv_B'$.
We begin by constructing the derivation for $A$ as follows 
$$\inference{
(i)
&
(ii)
}
{\regenv_1 \enta{\regvar'}{\var'} \assign{r}{x};\assign{x}{r};  \typeproduce 
\left \{ 
\begin{array}{l}
\assign{s}{x}; \\
\assign{r}{s};\\
\assign{x}{s};
\end{array} 
\right \}
, \secan{\tdontc}{\wdontk} ,\subst{\regenv_1}{\conn{s}{x}{\writeMode}}}$$ 
where 
$$ 
\begin{array}{rcl}
(i) & = &
\inference{s \in \regvar' & s = \compa{r} \\
\regenv_1, \{\} \eenta{\regvar'}{\var} x \typeproduce \code{\assign{s}{x}},\esecan{\valLow}{1},s,\subst{\regenv_1}{\conn{s}{x}{\readMode}}}
{\regenv_1 \enta{\regvar'}{\var'} \assign{r}{x} \typeproduce 
\left \{ 
\begin{array}{l}
\assign{s}{x}; \\
\assign{r}{s};
\end{array} 
\right \}, \secan{\tdontc}{\wdontk}, \subst{\regenv_1}{\conn{s}{r}{\writeMode}}} \\
(ii) & = & 
\inference{\subst{\regenv_1}{\conn{s}{r}{\writeMode}},\{\} \eenta{\regvar'}{\var'} r \typeproduce \code{\findot},\esecan{\valLow}{0},s,\subst{\regenv_1}{\conn{s}{r}{\writeMode}}}
{\subst{\regenv_1}{\conn{s}{r}{\writeMode}} \enta{\regvar'}{\var'} \assign{x}{r} \typeproduce \code{\assign{x}{s}}, \secan{\tdontc}{\wdontk}, \subst{\regenv_1}{\conn{s}{x}{\writeMode}}}
\end{array}$$
Notice that $\subst{\regenv}{\conn{r}{x}{\writeMode}} \regcorresp \subst{\regenv_1}{\conn{s}{x}{\writeMode}}$. We now need to show that
$$\subst{\regenv_1}{\conn{s}{x}{\writeMode}} \enta{\regvar'}{\var'} \while{r}{\code{D;A;\wskip}} \typeproduce \code{\dots},\secan{\tdontc}{\wdontk},\regenv_B'$$
Consider the following derivation
$$\inference{
r \in \var'  & s' \in \regvar' & s' = \shado{r}\\
A' = \assign{s'}{r};\assign{r}{s'}; \\
\regenv_1'=\subst{\subst{\regenv_1}{\conn{s}{x}{\writeMode}}}{\conn{s'}{r}{\writeMode}} 
& \regenv_B' \sqsubseteq \regenv_1'
& \regenv_B' \sqsubseteq \subst{\subst{\regenv_E'}{\conn{s}{x}{\writeMode}}}{\conn{s'}{r}{\writeMode}}\\
(iii)\\
}
{\subst{\regenv_1}{\conn{s}{x}{\writeMode}} \enta{\regvar'}{\var'} 
{
\left \{
\begin{array}{l}
\while{r}\{\\
\begin{array}{ll}
 & D;\\
 & \assign{r}{x};\\
 & \assign{x}{r};\\
 & \wskip; \}
\end{array}\\
\end{array}
\right \}
} 
\typeproduce 
\left \{
\begin{array}{l}
A';\\
\while{s'}\{\\
\begin{array}{ll}
& D';\\
& \assign{s}{x};\\
& \assign{r}{s};\\
& \assign{x}{s};\\
& \wskip;\\
& A';\\
& \wskip; \}
\end{array}\\
\end{array}
\right \}
,\secan{\tdontc}{\wdontk},\regenv_B'}
$$
Notice that $\subst{\regenv}{\conn{r}{x}{\writeMode}} \regcorresp  \subst{\subst{\regenv_1}{\conn{s}{x}{\writeMode}}}{\conn{s'}{r}{\writeMode}} = \regenv_1'$, since $r$ is already in correspondence with $s$. Consider the smallest $\regenv_B' \sqsubseteq \regenv_1'$ such that $\regenv_B \regcorresp \regenv_B'$ (the domain of $\regenv_B'$ is $\compa{\Dom(\regenv_B)}$). Then the following derivation
$$
(iii)= \inference{(iv) & (v) & (vi)}
{\regenv_B' \enta{\regvar'}{\var'} 
\left \{
\begin{array}{l}
D;\\
\assign{r}{x};\\
\assign{x}{r};\\
\wskip; 
\end{array}
\right \}
\typeproduce 
\left \{
\begin{array}{l}
D';\\
\assign{s}{x};\\
\assign{r}{s};\\
\assign{x}{s};\\
\wskip; 
\end{array}
\right \}
,
\left \langle
\begin{array}{l}
w \wlub \wdontk \wlub \whigh  \\
t' \tclub \tdontc \tclub \tconst~(1,1)
\end{array}
\right \rangle
,\regenv_E'}
$$
holds for
$$
\begin{array}{rl}
(iv) = & \regenv_B' \enta{\regvar'}{\var'} D \typeproduce \code{D'},\secan{t'}{w},\regenv_E'\\
(v) = & \regenv_E' \enta{\regvar'}{\var'} 
\left \{
\begin{array}{l}
\assign{r}{x};\\
\assign{x}{r};
\end{array}
\right \} \typeproduce 
\left \{
\begin{array}{l}
\assign{s}{x};\\
\assign{r}{s};\\
\assign{x}{s};\\
\end{array}
\right \}
,\secan{\tdontc}{\wdontk},\subst{\regenv_E'}{\conn{s}{x}{\writeMode}}\\
(vi) =  & \subst{\regenv_E'}{\conn{s}{x}{\writeMode}} \enta{\regvar'}{\var'} \wskip \typeproduce \code{\wskip},\secan{\tconst~(1,1)}{\whigh},\subst{\regenv_E'}{\conn{s}{x}{\writeMode}}\\
\end{array}
$$
In particular:
\begin{itemize}
\item the type derivation $(iv)$ follows from the inductive hypothesis on $D$, which also guarantees that $\regenv_E \regcorresp \regenv_E'$; 
\item type derivation $(v)$ is identical, up to environment renaming, to the same derivation calculated previously. It also holds that $\subst{\regenv_E}{\conn{s}{x}{\writeMode}} \regcorresp \subst{\regenv_E'}{\conn{s}{x}{\writeMode}}$, since $\regenv_E \regcorresp \regenv_E'$;
\item type derivation $(vi)$ is straightforward.
\end{itemize}
The case is concluded by observing that $\regenv_B \regcorresp \regenv_B'$ and $\regenv_B \sqsubseteq \subst{\regenv_E}{\conn{s}{x}{\writeMode}}$ and $\subst{\regenv_E}{\conn{s}{x}{\writeMode}} \regcorresp \subst{\regenv_E'}{\conn{s}{x}{\writeMode}}$ implies that $\regenv_B' \sqsubseteq \subst{\subst{\regenv_E'}{\conn{s}{x}{\writeMode}}}{\conn{s'}{r}{\writeMode}}$. Also, $w \wlub \wdontk \wlub \whigh = \wdontk$ and  $t' \tclub \tdontc \tclub \tconst~(1,1) = \tdontc$ (since $t \nextTerm{} t'$ and $t \sqsubset \ttop$ implies that $t' \sqsubset \ttop$). 

\end{proof}

Proposition \ref{thm:typeprediction} states that the an \swhprog\ program obtained from a type-correct \whprog\ program is also type-correct, and from Preposition \ref{prop:strongsectype} we have that any type-correct program is strongly secure. We therefore conclude that the type system in Figures \ref{table:typesystematom} and  \ref{table:typesystem} compiles a \whprog\ program $C$ into an \swhprog\ program $D$ which enjoys Strong Security. Notice that the definition of Strong Security for \swhprog\ programs is obtained by considering the semantics of \swhprog\ programs as a LTS $\iwhlts=\{\{\langle D,\mem \rangle \}, \whsem{}, \{ch!n|ch \in \{\low,\high\} \mbox{ and } n \in \mathbb{N}\} \cup \{\tau\} \}$, where $\whsem{}$ is defined by rules in Figure \ref{table:whilesem} for $\allvari= \var \cup \regvar$.

\subsection{From \swhprog\ programs to \assem\ programs}\label{Sec:iWh-asse}

We now introduce the compilation function that translate \swhprog\ programs into \assem\ programs. The syntax for \assem\ programs is defined in Figure \ref{table:assemsyntax}. The translation $\whtoasse$ between \swhprog\ and \assem\ instructions is defined in Figure \ref{table:transiwa}.

\begin{figure}[h]
\centering
$
\begin{array}{| l c l |}
 \hline 
\whtoasse(\wskip,l) 				& = &    ([l: \nop], \emptyLab)\\
 \hline
 \whtoasse(\assign{x}{r},l) 			& = &   ([l: \storeName \ \vartoptr(x) \ r], \emptyLab) \\
 \hline
 \whtoasse(\assign{r}{n},l) 		& = &  ([l: \movekName\ r\ n], \emptyLab)\\
 \whtoasse(\assign{r}{r'},l)			& = &  ([l: \moverName\ r\ r'], \emptyLab)\\
 \whtoasse(\assign{r}{x},l)			& = &  ([l: \loadName\ \vartoptr(x)\ r], \emptyLab)\\
 \whtoasse(\assign{r}{\wop{r}{r'}},l) 	& = &  ([l: \aopName \ r \ r'], \emptyLab)\\
\hline
\whtoasse(\out{ch}{r},l) 			& = & ([l: \out{ch}{r}], \emptyLab)\\
\hline
	\whtoasse
	\left (
	
	\wifa{r}{\{D_1;\wskip\}}{\{D_2;\wskip\}},l \right ) & = & 
	
	\begin{array}{l}
	([l: \jzName\ br \ r] \concatprograms P_1 \concatprograms [l_1: \jmpName\ ex]  \\		
								\concatprograms  P_2 \concatprograms [l_2: \nop], ex) \\
								 	\mbox{where } \\
								 \hspace{3ex} \whtoasse(D_1,\emptyLab)=(P_1,l_1) \\
								 \hspace{3ex} \whtoasse(D_2,br)=(P_2,l_2)	\\
								 \mbox{for fresh labels } br, ex \\
	\end{array}\\
\hline
	\whtoasse(\while{r}{\{D;\wskip\}},l) 	& = & ([l: \jzName\ ex \ r] \concatprograms P \concatprograms [l': \jmpName\ l], ex) \\
								&   & 	\mbox{where } \\
								&   & \hspace{3ex} \whtoasse(D,\emptyLab)=(P,l') \\
								&   & \mbox{for fresh label } ex\\
 \hline 
	\whtoasse(D_1; D_2	,l)		& = & (P_1 \concatprograms P_2, l')\\
							&   & 	\mbox{where } \\
							&   & \hspace{3ex} \whtoasse(D_1,l)=(P_1,l_1) \\
							&   & \hspace{3ex} \whtoasse(D_2,l_1)=(P_2,l')	\\
\hline
									
\end{array}$
\caption{Translation between \swhprog\ and \assem\ programs\label{table:transiwa}}
\end{figure}

Beside the input \swhprog\ code and the output \assem\ program, the function $\whtoasse$ requires a label as input and produces a label as output. The role of these labels is clarified in the explanation of $\wifName$ and $\whileName$ instructions compilation. 

All cases $\wskip$, $\assign{x}{r}$, $\assign{r}{F}$ and $\out{ch}{r}$ correspond to a straightforward mapping between \swhprog\ instructions and \assem\ instructions. The input label is used to label the produced instruction, and the empty label $\emptyLab$ is given as output. The first instruction in the compilation of the $\wifName$ statement is $l:\jzName\ br\ r$, that corresponds to the register variable evaluation in the \swhprog\ language. Branches $D_1$ and $D_2$ are translated into \assem\ programs $P_1$ and $P_2$ respectively. In particular, $P_2$ requires the label $br$ to be used as argument for $\whtoasse$ in order to move the control flow from $l: \jzName\ br\ r$ to $P_2$ in case the content of $r$ is 0. The trailing $\wskip$ statements are converted into a $l_1: \jmpName\ ex$ and a $l_2:\ \nop$ respectively, to guarantee that the first instruction to be executed after any of the two branches is the one labeled with $ex$, if any. The  compilation of the $\whileName$ instruction produces a conditional jump on $r$ as the first instruction, similarly to the $\wifName$ case. The compilation of the loop body follows, and the $\wskip$ statement is replaced by the $l':\ \jmpName\ l$ instruction, that moves back the control to the register evaluation when the execution of the loop body is completed. 
 
The relation between an \swhprog\ program $D$ and its corresponding \assem\ program $\whtoasse(D)$ is stated in terms of the \assem\ semantics. Recall from Section \ref{sec:procdescr} that the semantic of \assem\ programs is defined in terms of a labelled transition system $\asselts=\{ \{\langle P,k,\RegName,\HeapName\rangle \}, \{ch!k| ch \in \{\low,\high\} \cup \{\tau\}\},\asem{} \}$ where the first two elements of the state $P$ and $k \in \Constant$ are the fault-tolerant component $T$,  the elements $\RegName$ and $\HeapName$ are the faulty part $F$ and transitions are defined in Figure \ref{table:abssemantics}.

In general, a \assem\ program behaves exactly as the corresponding \swhprog\ program. However, this does not hold for \swhprog\ programs that manipulate constants outside $\Constant$. We therefore postulate the existence of a technique for approximating the resource footprint for the source \swhprog\ program statically (cf. the approach to transparency in \cite{DelTedesco+:FTNInterference}) and, from now onwards, we focus only on \swhprog\ programs that have a resource footprint compatible with the capabilities of the \assem\ machine. 

Before proceeding further, we need some notational conventions. We say an \swhprog\ memory $\mem$ is equivalent to the \assem\ storage $(\RegName,\HeapName)$, written as $\mem \memequiva (\RegName,\HeapName)$ if $\restr{\mem}{\regvar}= \RegName$ and $\restr{\mem_w}{\var}= \vartoptr \circ \HeapName$. 
 
\begin{definition}[Strong Coupling]
A relation $R$ between \swhprog\ programs and $(P, k)$ pairs, where $P$ is a \assem\ program and $k \in \Constant$ is an exact bisimulation if for any $(D,(P,n)) \in R$, $\forall \mem \memequiva (\RegName,\HeapName)$, $\whstate{D}{\mem} \whsem{l} \whstate{D'}{\mem'}$  if and only if $\absStateJump{P}{n}{\RegName}{\HeapName} \asem{l}
\absStateJump{P}{n'}{\RegName'}{\HeapName'}$ such that $\mem' \memequiva (\RegName',\HeapName')$ and $(D',(P,n')) \in R$. We say that an \swhprog\ program $D$ is strongly coupled with a \assem\ program $P$, written as $D \strongsim P$, if there exists an exact bisimulation $R$ such that $(D,(P,0)) \in R$.
\end{definition}

We can now characterize the existing relation between an \swhprog\ program $D$ and its translated version $P$ in terms of strong coupling.

\begin{proposition}[Strong coupling between \swhprog\ and \assem\ programs]\label{prop:iwhileandrisc}
Let $D$ be an \swhprog\ program and $P$ be the first component of $\whtoasse(D,\emptyLab)$. Then $D \strongsim P$. 
\end{proposition}

The result in this section can be combined with the ones in Section \ref{Sec:Wh-iWh} to show that we can translate a type-correct \whprog\ program into a strongly secure \assem\ program. In particular, 
Let $C$ be a \whprog\ program such that $\{\} \enta{\regvar}{\var}  C \typeproduce \code{D}, \esecan{t}{w}, \regenv'$ and let $\whtoasse(D,\emptyLab)=(P,l)$. Then $P$ is strongly secure. 

\section{Example of a Compilation}\label{appendix:example}
In this section we show an example of a simple hash calculation that we can implement\footnote{ Notice that this statement is not completely accurate. The inaccuracy arises from the fact that $guard>0$ is not a valid expression in the language. This is not a big issue, since RISC can be smoothly extended by including a $\jlez$ instruction without breaking any of the results.} within our language. Informally, the hashing algorithm works as follows:
\begin{itemize}
\item define two public integers $i$ and $j$ such that $1 < j \ll i$
\item pick a public message $m$ in the interval $[0,…,2^i-1]$, while the public hash $h$ will be in $[0,…,2^j-1]$ 
\item define a public value $p$, the smallest prime such that $p> 2^j$ 
\item pick two secret values $q$ and $r$ such that $1\leq q \leq p-1$, $0 \leq r \leq p-1$ 
\item calculate the hash as $h_{q,r}(m)=[(q*m+r) \mod p] \mod 2^i$ 
\end{itemize}

In the \whprog\ language a possible implementation of the hash program proceeds as follows:
\[
\begin{array}{l}
limit := 1;  \\
\while{i}{\{limit:=limit*2; i:=i-1;\} }\\
source:=q*m; \\
source:=source+r; \\
guard:=source -p; \\
\while{guard>0}{ \{guard:=guard-p; source := source -p;\} }\\
guard:=source -limit; \\
\while{guard > 0} do{\{ guard := guard -limit; source := source -limit;\} }
\end{array}
\]
	
This correspond to the following \assem\ code
\[
\Scale[0.9]{
\begin{array}{rllcrll}
 & \movek r_{lim} 1 		& 						& & & \% &  \\
 & \storeName \ limit\  r_{lim}	& \conn{r_{lim}}{limit}{W} 	& & & \storeName \ guard \ r_g & \conn{r_{g}}{guard}{W}\\
 & \loadName\ r_i\ i 		& 						& & &  \loadName \ r_s \ source & \\
 & \storeName\ i \ r_i 		& \conn{r_{i}}{i}{W} 		& & &  \sub\ r_s \ r_{hlim} & \\
 loop1: & \jzName \ exit\_loop1\  r_i  			 &		& & & \storeName \ source \ r_s &\conn{r_{s}}{source}{W}\\
 & \movek \ r_2 \ 2  					&			& & &  \loadName\ r_g\ guard & \\
 & \mul \ r_{lim} \ r_2  					&			& & &   \storeName \ guard \ r_g &\conn{r_{g}}{guard}{W}\\
 & \storeName \ limit \ r_{lim} &\conn{r_{lim}}{limit}{W} 	& & &  \jmp loop3 & \\
 & \movek\ r_1\ 1 && & & \\
 & \sub\ r_i\ r_1   	&& & & \\
 & \storeName\ i\ r_i 	& \conn{r_{i}}{i}{W}& & &\\
 & \loadName\ r_i\ i & & & &\\
 & \storeName\ i \ r_i & \conn{r_{i}}{i}{W}& & &\\
 & \jmpName \ loop1 & & & &\\
exit\_loop1: &  \loadName\ r_g \ q & & & &\\
& \loadName\  r_m \ m & \conn{r_{m}}{m}{R} & & &\\
 & \mul \ r_g\ r_m & & & &\\
 & \storeName\ source\ r_g & \conn{r_{g}}{source}{W}& & &\\
  & \loadName\ r_r \ r & \conn{r_{r}}{r}{R}& & &\\
  & \addName \ r_g \ r_r && & & \\
  & \storeName\ source \ r_g & \conn{r_{g}}{source}{W}& & &\\
  & \loadName\ r_p \ p & \conn{r_{p}}{p}{R}& & &\\
  & \sub \ r_g \ r_p & & & &\\ 
  & \storeName\ guard \ r_g & \conn{r_{g}}{guard}{W} & & &\\
  & \loadName \ r_g \ guard && & & \\
  & \storeName \ guard \ r_g &\conn{r_{g}}{guard}{W}& & & \\
loop2: & \jlez \ exit\_loop2  \ r_g & & & &\\
& \sub \ r_g \ r_p && & & \\
 & \storeName \ guard \ r_g & \conn{r_{g}}{guard}{W}& & & \\
 & \loadName \ r_s  \ source && & & \\
 & \sub \ r_s \ r_p && & & \\
 & \storeName \ source \ r_s& \conn{r_{s}}{source}{W}& & &\\
 &  \loadName \ r_g \ guard && & & \\
 & \storeName \ guard \ r_g & \conn{r_{g}}{guard}{W}& & &\\
 & \jmp loop2 && & & \\
 exit\_loop2:& \loadName \ r_{hlim}\ limit & \conn{r_{hlim}}{limit}{W}& & &\\
  & \loadName\ r_g \ source & & & &\\
  & \sub \ r_g \ r_{hlim} & & & &\\
  & \storeName\ guard\ r_g&  \conn{r_{g}}{guard}{W}& & &\\
  &  \loadName \ r_g \ guard & & & &\\
  & \storeName\ guard \ r_g & \conn{r_{g}}{guard}{W}& & &\\
  loop3: &  \jlez\ exit\_loop3 \ r_g  && & & \\
  & \sub\ r_g \ r_{hlim} & & & &\\
  & \% & & & & \\
\end{array}
}
\]

 The security levels of resources involved in this example are assigned as follows:
 \begin{itemize} 
 \item level $\valLow$ variables: $limit$, $i$, $m$, $p$;
 \item level $\valHigh$ variables: $q$, $source$, $guard$, $r$;
 \item level $\valLow$ registers: $r_{lim}$, $r_i$, $r_2$, $r_1$; 
 \item level $\valHigh$ registers: $r_g$, $r_m$, $r_p$, $r_s$, $r_{hlim}$, $r_r$.
 \end{itemize}

\end{document}